\documentclass[11pt,xcolor=dvipsnames, reqno]{amsart}
\usepackage[a4paper, top=3cm, bottom=4cm, left=3cm, right=3cm]{geometry}
\usepackage[T2A, T1]{fontenc}
\usepackage[utf8]{inputenc}

\usepackage[english]{babel}
\usepackage{amsmath,amssymb,amstext}
\usepackage{amsthm}
\usepackage{amsfonts}
\usepackage{amsxtra}
\usepackage{graphicx}
\usepackage{wrapfig}
\usepackage[dvipsnames]{xcolor}
\usepackage{subcaption}
\usepackage{braket}
\usepackage{enumerate}
\usepackage[super]{nth}
\usepackage{pdfpages}
\usepackage{graphicx}

\usepackage{hyperref}
\usepackage[backend=bibtex,style=numeric,sorting=none,firstinits=true, maxbibnames=99]{biblatex}

\hypersetup{breaklinks=true}

\usepackage{xifthen}

\allowdisplaybreaks

\bibliography{Literature.bib}

\newcommand{\R}{\mathbb{R}}
\newcommand{\N}{\mathbb{N}}
\newcommand{\bC}{\mathbb{C}}
\newcommand{\Z}{\mathbb{Z}}

\newcommand{\B}{\mathcal{B}}

\newcommand{\cB}{\mathcal{B}}

\newcommand{\cF}{\mathcal{F}}
\newcommand{\cG}{\mathcal{G}}

\newcommand{\dr}[1]{\mathrm{d}{#1}\, }

 \newcommand{\Span}{\operatorname{span}}

\newcommand{\ga}{\gamma}
\newcommand{\be}{\beta}
\newcommand{\de}{\delta}
\newcommand{\dt}{\tilde\delta}
\newcommand{\ho}{h_\Omega}
\newcommand{\dz}{\delta_0} % important
\newcommand{\dc}{\delta_c}
\newcommand{\pt}{\tilde p}
\newcommand{\rt}{\tilde r}
\newcommand{\ps}{{p^*}}
\newcommand{\qs}{{q^*}}
\newcommand{\om}{\Omega}
\newcommand{\obd}{\Omega^b_{\ha\sqrt{d}\dz}}
\newcommand{\obz}{\Omega^b_{\sqrt{d}\dz}}

\newcommand{\ol}{\Omega_l}
\newcommand{\old}{\Omega_l^{(d-1)}}
\newcommand{\olt}{\tilde\Omega_l}
\newcommand{\oltd}{\tilde\Omega_l^{(d-1)}}
\newcommand{\olh}{\hat\Omega_l}
\newcommand{\olhd}{\hat\Omega_l^{(d-1)}}
\newcommand{\dg}{\frac{d-1}{\ga}} % important
\newcommand{\dgo}{{\frac{d-1}{\ga}+1}} % important
\newcommand{\ls}{\lesssim} % important
\newcommand{\gs}{\gtrsim} % important
\newcommand{\dep}{d,\, \ga,\, c,\,\ho,\, L}
\newcommand{\deps}{d,\, \ga}
\newcommand{\depc}{d,\, \ga,\, c}
\newcommand{\dwc}{d,\, \ga,\, L}
\newcommand{\hxm}{h_{x,\min}}
\newcommand{\hym}{h_{y,\min}}
\newcommand{\uoml}{\bigcup_{l=1}^L\ol}
\newcommand{\gai}{\frac{1}{\ga}}
\newcommand{\ha}{\frac{1}{2}}
\newcommand{\di}{\frac{1}{d}}
\newcommand{\ome}{\omega}
\newcommand{\ze}{\zeta}

\newcommand{\cps}{C_{PS}}
\newcommand{\cp}{C_{P}}
\newcommand{\cs}{C_{S}}
\newcommand{\odo}{\omega_{d-1}}
\newcommand{\tm}{\tilde M}
\newcommand{\td}{{\tilde D}}
\newcommand{\tb}{{\tilde B}}
\newcommand{\tf}{\tilde f}
\newcommand{\tu}{\tilde u}
\newcommand{\hd}{{\hat D}}
\newcommand{\at}{{\tilde a}}
\newcommand{\ah}{{\hat a}}

\newcommand{\eps}{{\epsilon}}

\newcommand{\supp}{{\operatorname{supp}\,}}

\newcommand{\dist}[1]{\operatorname{dist}\left(#1\right)}

\newcommand\no[2]{\left\lVert#1\right\rVert_{#2}}
\newcommand\nor[3]{\left\lVert#1\right\rVert_{#2}^{#3}}
\newcommand\nos[3]{\left\lvert#1\right\rvert_{#2}^{#3}}
\newcommand\noi[1]{\left\lvert#1\right\rvert_{\infty}}
\newcommand\ev[1]{N\left(-\Delta_{#1}^{N}+V\right)}
\newcommand\evv[2]{N\left(-\Delta_{#1}^{N}+#2\right)}

\newcommand\evd[1]{N\left(-\Delta_{#1}^{D}+V\right)}
\newcommand\inx[2]{\int_{#1}\dr{#2}}
\newcommand\inu[3]{\int_{#1}^{#2}\dr{#3}}
\newcommand\ob[1]{\om^b_{#1}}

\let\epsilon\varepsilon
\let\phi\varphi

\makeatletter
\providecommand*{\cupdot}{%
  \mathbin{%
    \mathpalette\@cupdot{}%
  }%
}
\newcommand*{\@cupdot}[2]{%
  \ooalign{%
    $\m@th#1\cup$\cr
    \sbox0{$#1\cup$}%
    \dimen@=\ht0 %
    \sbox0{$\m@th#1\cdot$}%
    \advance\dimen@ by -\ht0 %
    \dimen@=.5\dimen@
    \hidewidth\raise\dimen@\box0\hidewidth
  }%
}

\providecommand*{\bigcupdot}{%
  \mathop{%
    \vphantom{\bigcup}%
    \mathpalette\@bigcupdot{}%
  }%
}
\newcommand*{\@bigcupdot}[2]{%
  \ooalign{%
    $\m@th#1\bigcup$\cr
    \sbox0{$#1\bigcup$}%
    \dimen@=\ht0 %
    \advance\dimen@ by -\dp0 %
    \sbox0{\scalebox{2}{$\m@th#1\cdot$}}%
    \advance\dimen@ by -\ht0 %
    \dimen@=.5\dimen@
    \hidewidth\raise\dimen@\box0\hidewidth
  }%
}
\makeatother

\newcommand{\normiii}[1]{{\left\vert\kern-0.25ex\left\vert\kern-0.25ex\left\vert #1 
    \right\vert\kern-0.25ex\right\vert\kern-0.25ex\right\vert}}

\newtheorem{theorem}{Theorem}[section]
\newtheorem{lemma}[theorem]{Lemma}
\newtheorem{definition}[theorem]{Definition}
\newtheorem{corollary}[theorem]{Corollary}
\newtheorem{remark}[theorem]{Remark}

\allowdisplaybreaks

\title[Semiclassical estimates for Schr\"odinger operators on H\"older domains]{Semiclassical estimates for Schr\"odinger operators with Neumann boundary conditions on H\"older domains}
\author{Charlotte Dietze}
\address{Department of Mathematics, LMU Munich, Theresienstr.~39, 80333 Munich, Germany\newline
Institut des Hautes Études Scientifiques, 35 route de Chartres, 91440 Bures-sur-Yvette, France}
\email{dietze@math.lmu.de}
\date{\today}

\makeatletter
\@namedef{subjclassname@2020}{\textup{2020} Mathematics Subject Classification}
\makeatother

\subjclass[2020]{35P15, 35P20}
\keywords{Neumann Laplacian, H\"older domains, Cwikel-Lieb-Rozenblum inequality, Semiclassical asymptotics, Weyl's law.}

\begin{document}
\maketitle
\begin{abstract}
    % We prove Weyl's law for Schr\"odinger operators with Neumann boundary conditions on bounded H\"older domains under suitable assumptions on the H\"older exponent and the external potential.
    % Furthermore, we provide a universal bound for the number of negative eigenvalues of Schr\"odinger operators with Neumann boundary conditions on bounded Hölder domains. Our bound yields same semiclassical behaviour as the Weyl asymptotics for smooth domains. We also provide an example of a Schr\"odinger operator with Neumann boundary conditions on a bounded H\"older domain which exhibits a different semiclassical behaviour. 
We prove a universal bound for the number of negative eigenvalues of Schr\"odinger operators with Neumann boundary conditions on bounded H\"older domains, under suitable assumptions on the H\"older exponent and the external potential. Our bound yields the same semiclassical behaviour as the Weyl asymptotics for smooth domains. We also discuss different cases where Weyl's law holds and fails.
\end{abstract}

%\tableofcontents

\section{Introduction}

The celebrated correspondence principle, which goes back to  Niels Bohr in the early days of quantum mechanics, states that quantum systems exhibit classical behaviour in the limit of large quantum numbers. In the context of spectral analysis of Schr\"odinger operators, this leads to the semiclassical approximation, which suggests that any bound state can be related to a volume of size $(2\pi)^d$ in the phase space $\R^d \times \R^d$ \cite[Section 4.1.1]{lieb2010stability}. In particular, the number of negative eigenvalues $N\left(- \Delta_{\Omega} - \lambda\right)$ of $-\Delta_{\Omega} - \lambda$ on a domain $\Omega \subset \R^d$ with suitable boundary conditions can be approximated by its semiclassical analogue
\begin{align}\label{eq:intweyldconst-0}
    N\left(- \Delta_{\Omega} - \lambda\right)  \approx \frac{1}{(2\pi)^{d} } \left|\left\{(p,x) \in \R^d \times \Omega  \bigm|   |p|^2 - \lambda <0 \right\}\right| =  \frac{|B_1(0)|}{(2\pi)^d} |\Omega|\lambda^\frac{d}{2} 
\end{align}
in the large coupling limit $\lambda \to \infty$, with $|B_1(0)|$ the volume of the unit ball in $\R^d$. More generally, for a general potential $V: \Omega \to (-\infty,0]$ one might expect that 
\begin{align}\label{eq:intweyldconst-1}
    N\left(- \Delta_{\Omega} + \lambda V \right)  \approx \frac{1}{(2\pi)^{d} } \left|\left\{(p,x) \in \R^d \times \Omega \bigm|  |p|^2 + \lambda V(x) <0 \right\}\right| =  \frac{|B_1(0)|}{(2\pi)^d}  \int_{\Omega} |\lambda V|^{\frac d 2}.
\end{align}

Rigorous justifications of \eqref{eq:intweyldconst-0} and \eqref{eq:intweyldconst-1} have been shown for a large class of smooth domains $\Omega$ and potentials $V$. On the other hand, in general, implementing the semiclassical approximation for rough domains and potentials is difficult. In this paper, we will discuss the validity of \eqref{eq:intweyldconst-1} for H\"older domains $\Omega$ and suitable $L^p$-integrable potentials $V$. We will focus on Neumann boundary conditions as Dirichlet boundary conditions have been well understood.

\subsection{Main results} Let $d\in \N$, $d\geq 2$ and let $\Omega \subset \R^d$ be a domain, that is, an open bounded and connected subset of $\R^d$. For $\ga\in\left(0,1\right]$, a $\ga$-Hölder domain  is a domain $\Omega\subset \R^d$ which is locally the subgraph of a $\gamma$-Hölder continuous function $f$, that is, there exists a constant $c>0$ such that 
\begin{equation} \label{eq:Holder-intro}
|f(x)-f(y)|\le c|x-y|^\ga
\end{equation}
for all $x,y$ in the domain of $f$, which is an open subset of $\R^{d-1}$. In the case $\gamma=1$, we call $\Omega$ a Lipschitz domain. See Section \ref{de:gahodomain} for details.

\bigskip

Denote the Dirichlet Laplacian on $\Omega$ by $- \Delta^D_\Omega$ and the number of negative eigenvalues of $-\Delta^D_\Omega - \lambda$ by $N\left(-\Delta^D_\Omega - \lambda\right)$, where $\lambda > 0$. One of the fundamental results of spectral theory is Weyl's law \cite{weyl, weyl1912asymptotische, rozenblum1971distribution, rozenbljum1972eigenvalues}, which states that
\begin{equation}\label{eq:intweyldconst}
    N\left(- \Delta^D_\Omega - \lambda\right) = \frac{|B_1(0)|}{(2\pi)^d} |\Omega|\lambda^\frac{d}{2} + o\left(\lambda^\frac{d}{2}\right) \mathrm{\ as\ } \lambda \rightarrow \infty,
\end{equation}
thus rigorously justifying \eqref{eq:intweyldconst-0}, see also \cite{ivrii1980second} for a second order result. This asymptotics also holds for the Neumann Laplacian\footnote{Here $-\Delta_{\om}^{N}$ is the self-adjoint operator generated by the quadratic form $\int_\om  |\nabla u|^2$ for all $u\in H^1(\om)$.} $- \Delta^N_\Omega$ for Lipschitz domains $\Omega$, or more generally for extension domains, see \cite{birman1970principal} and also
\cite[Theorem 3.20]{frank2021schrodinger}. % and Remark \ref{re:extweyl} in Appendix \ref{app:extop}. 
% Appendix \ref{app:extop}.
% extension domains $\Omega$, see Lemma \ref{le:manybodyinequality} for a definition and an equivalent characterisation of extension domains $\Omega$. 
% Examples for extension domains include Lipschitz domains, so in particular smooth domains and $C^1$-domains satisfy the extension property. 

\bigskip

%While the Weyl asymptotics for the Neumann Laplacian hold for Lipschitz domains $\Omega$, or more generally for extension domains, they do \textit{not} hold for arbitrary domains $\Omega$.
In general, the Weyl asymptotics for the Neumann Laplacian does \textit{not} hold for arbitrary domains $\Omega$. 
%Recall that for $\om\subset\R^d$ open, the Neumann Laplacian $-\Delta_{\om}^{N}$ is a self-adjoint operator generated by the quadratic form $\int_\om  |\nabla u|^2$ for all $u\in H^1(\om)$.  
%
%\bigskip
%
It is well-known, see for example \cite{hempel1991essential}, that there are domains $\Omega$ such that zero is contained in the essential spectrum of $- \Delta^N_\Omega$. 
% and therefore\footnote{Strictly speaking, 
%% if $-\Delta^N_\Omega - \lambda$ has no negative eigenvalues, then 
%$N\left(-\Delta^N_\Omega - \lambda\right)$ denotes the dimension of the image of the spectral projection $1_{(-\infty,0)}(-\Delta^N_\Omega - \lambda)$.},
%\begin{equation}
%    N\left(-\Delta^N_\Omega - \lambda\right) = \infty \quad \text{ for all } \lambda>0.
%\end{equation}
%
%\bigskip
% Hölder domains do \textit{not} satisfy the extension property, see \cite[p.~213]{liebloss} and \cite{labutin}. 
Interestingly, Netrusov and Safarov showed that the Weyl asymptotics \eqref{eq:intweyldconst} holds for the Neumann Laplacian $- \Delta^N_\Omega$, with any $\gamma$-Hölder domain $\Omega$, if and only if $\gamma>(d-1)/d$ (see \cite[Corollary 1.6 and Theorem 1.10]{netrusov2005weyl}). 
% holds:

%Interestingly, if $\Omega$ is a $\gamma$-Hölder domain with Hölder exponent 
%$\gamma \in ((d-1)/d,1)$, then Netrusov and Safarov \cite[Corollary 1.6]{netrusov2005weyl} showed that the Weyl asymptotics \eqref{eq:intweyldconst} holds for the Neumann Laplacian $- \Delta^N_\Omega$. 
%% holds:
%%\begin{equation}\label{eq:nmlaplacianlimit}
%%    N\left(-\Delta^N_\Omega - \lambda\right) = \frac{|B_1(0)|}{(2\pi)^d} |\Omega|\lambda^\frac{d}{2} + o\left(\lambda^\frac{d}{2}\right) \mathrm{\ as\ } \lambda \rightarrow \infty
%%\end{equation}
%They also showed that Weyl's law fails for $\gamma \le (d-1)/d$ \cite[Theorem 1.10]{netrusov2005weyl}. %, namely for any $\gamma \in \left(0,\frac{d-1}{d}\right]$ there exists a $\gamma$-Hölder domain $\om$ such that \eqref{eq:nmlaplacianlimit} does not hold \cite[Theorem 1.10]{netrusov2005weyl}. 

\bigskip

In the present paper, we are interested in $\evv{\Omega}{V}$ with a  potential $V:\Omega \to (-\infty, 0]$ on a  Hölder domain $\Omega$. Unlike the case of constant potentials, the problem with a general potential $V$ is more subtle as the following theorem shows.

\begin{theorem}[Example with non-semiclassical behaviour]\label{th:example} Let $d\ge 2$.     For every $\gamma \in \left(\tfrac{d-1}{d}, 1\right)$ there exists a $\gamma$-Hölder domain $\Omega \subset \R^d$ and $V : \Omega \rightarrow \left( -\infty, 0\right]$ with $V \in L^\frac{d}{2}(\Omega)$ such that
    \begin{equation}\label{eq:thexeq}
        \limsup_{\lambda \rightarrow \infty} \lambda^{-d/2} \evv{\Omega}{\lambda V}  = \infty .
    \end{equation}
\end{theorem}

%Theorem \ref{th:example} shows that unless we assume stronger assumptions on the potential (such as in Theorem \ref{th:weightednorm} below), we cannot even expect semiclassical behaviour. 

%\bigskip

We prove Theorem \ref{th:example} by constructing a $\gamma$-Hölder domain $\Omega$ in the same way as Netrusov and Safarov in \cite[Theorem 1.10]{netrusov2005weyl}, and we choose the potential $V$ to be growing to infinity near the boundary in such a way that $- \Delta^N_\Omega + \lambda V$ can support significantly more  than $ \lambda^{ \frac{d}{2}}$ bound states, see Section \ref{ss:strex} for more details.

\bigskip

Our next result is a universal bound on the number of negative eigenvalues of the Schr\"odinger operator $-\Delta^{N}_\Omega + V$ on $L^2(\Omega)$ with a suitable potential $V$ on a H\"older domain $\Omega$. The Cwikel-Lieb-Rozenblum inequality \cite{cwikel1977weak,lieb1976bounds,rozenblum1972} states that for any open set $\om\subset\R^d$ for $d\ge3$, there exists a constant $C=C(d)>0$ such that for every $V : \om \rightarrow \left(- \infty, 0 \right]$, we have
    \begin{equation}\label{eq:clrintbasic2}
        \evd{\om} \leq C \int_{\R^d} |V|^{\frac d 2}.
    \end{equation}
%This inequality was first proved independently by Cwikel \cite{cwikel1977weak}, Lieb \cite{lieb1976bounds} and Rozenblum \cite{rozenblum1972}. 
Here $\|V\|_p$ is the $L^p$ norm of $V$. In the case of the Neumann Laplacian on a Lipschitz domain $\Omega$, if $d\ge 3$, then 
    \begin{equation}\label{eq:clrintbasicneu}
        \ev{\om} \leq C_\om\left(1 + \int_{\R^d} |V|^{\frac d 2}\right)
    \end{equation}
 for a finite constant $C_\om>0$ independent of the potential $V$, 
 see for example \cite[Corollary 4.37]{frank2021schrodinger}. In dimension $d=2$, \eqref{eq:clrintbasic2} and \eqref{eq:clrintbasicneu} still hold, provided that $\no{V}{\frac{d}{2}}$ is replaced by the  Orlicz norm $\no{V}{\cB}$ \cite{frank2019bound,solomyak1994piecewise}. % (see Appendix \ref{se:orlicz}).

% Again, if $d=2$, the norm  $\no{V}{\frac{d}{2}}$ in \eqref{eq:clrintbasicneu} should be replaced by the  Orlicz norm  \cite{solomyak1994piecewise}. 
 
 \bigskip
 
On the other hand, Theorem \ref{th:example} implies that \eqref{eq:clrintbasicneu} cannot be extended to Hölder domains. Our main result is a replacement of \eqref{eq:clrintbasicneu} for Hölder domains under suitable assumptions on the potential $V$. We need to use a weighted $L^{p}$-norm with a weight that grows to infinity as one approaches $\partial \Omega$. A simplified version of our result reads as follows.

\begin{theorem}[Cwikel-Lieb-Rozenblum type bound]\label{th:weightednorm}
   Let $d\ge 2$. Let $\gamma\in \left[\tfrac{d-1}{d}, 1\right)$ and let $\Omega$ be a $\gamma$-Hölder domain. Then there exists a constant $C_\om = C_\om(d, \gamma, \Omega)> 0$ such that for every $V : \Omega \rightarrow \left(- \infty, 0 \right]$ with $\normiii{V} < \infty$, we have
    \begin{equation}\label{eq:clrthm}
        \ev{\Omega} \leq C_\om\left(1 + \normiii{V}^\frac{d}{2}\right) .
    \end{equation}
    Here the norm $ \normiii{V}=\no{V}{\pt,\be}$ is given in Definition \ref{de:hxlseminorm}, with $\beta$ and $\pt$  chosen as in \eqref{eq:ptdef}. 
   Moreover, if $\gamma \in \left[\tfrac{2(d-1)}{2d - 1} , 1 \right)$, then $\normiii{V}$ can be replaced by $\no{V}{p}$ where $p=p_{d,\gamma}>\frac d 2 $ is a constant  depending only on $d$ and $\gamma$ satisfying
       \begin{equation}
   \lim_{\gamma\to 1}     p_{d,\gamma} = \frac{d}{2}. 
     \end{equation}
\end{theorem}

A more precise statement is given in Theorem \ref{th:main} below. %Note that the only difference between \eqref{eq:clrthm} and \eqref{eq:clrintbasicneu} is that $\no{V}{\frac{d}{2}}$  is replaced by $\normiii{V}$. 
The norm  $\normiii{V}$ is stronger than $\no{V}{\frac{d}{2}}$ (in particular, the potential $V$ in Theorem \ref{th:example} satisfies $\normiii{V}=\infty$). Nevertheless, by \eqref{eq:clrthm}, we still get the correct semiclassical behaviour as soon as $\normiii{V}<\infty$, namely
\begin{equation}\label{eq:clrnicesemib}
    \evv{\om}{\lambda V}=\mathcal{O}\left(\lambda^{\frac{d}{2}}\right) \ \textrm{as } \lambda\to\infty. 
\end{equation}

\medskip

On the technical level, our norm $\normiii{V}$ is chosen carefully to capture the correct leading order behavior of the number of bound states close to the boundary. By following Rozenblum's method \cite{rozenblum1972},
% , see \eqref{eq:evltsim1plusintsomega} in Section \ref{s:ingredient}, 
it is possible to obtain  the following bound
\begin{equation}\label{eq:evltsim1plusintsomega-intro}
    \ev{\Omega} \lesssim 1 + \int_\Omega |V|^\frac{d}{2} + \int_{\mathrm{close\ to\ } \partial \Omega} |V|^{\tilde{p}}
\end{equation}
for some $\tilde{p}> \frac d 2$ (this bound could also be obtained from the analysis in \cite{frank2010equivalence} and \cite{labutin}). However, \eqref{eq:evltsim1plusintsomega-intro} is insufficient to deduce \eqref{eq:clrnicesemib}.

\medskip

As a consequence of Theorem \ref{th:weightednorm}, we are able to come back to sharp semiclassics. 

\begin{theorem}[Weyl's law for Schr\"odinger operators on H\"older domains]\label{th:weyllawforapotential} Let $d\ge 2$. Let $\gamma \in \left[\frac{d-1}{d}, 1\right)$ and let $\Omega \subset \R^d$ be a $\gamma$-Hölder domain. Let $V : \Omega \rightarrow (-\infty, 0]$ be measurable and such that $\normiii{V} < \infty$, where the norm $\normiii{\cdot}$ is the same as in Theorem \ref{th:weightednorm}. Then
\begin{equation}
    N \left(-\Delta^N_\Omega + \lambda V\right) = (2 \pi)^{-d} \left|B_1(0)\right| \lambda^{\frac{d}{2}} \int_\Omega |V|^{\frac{d}{2}} + o \left(\lambda^{\frac{d}{2}}\right) \mathrm{\ as\ } \lambda \rightarrow \infty .
\end{equation}
\end{theorem}

Let us remark that Weyl's law for constant potentials fails for $\gamma = \tfrac{d-1}{d}$ \cite[Theorem 1.10]{netrusov2005weyl}. This does not contradict  Theorem \ref{th:weyllawforapotential} since $\normiii{-1_{\Omega}} = \infty$ for $\gamma = \tfrac{d-1}{d}$.
% A corresponding result for $V\in L^\infty(\om)$ can be easily proved using \cite[Corollary 1.6]{netrusov2005weyl}, see Theorem \ref{th:weylbdpot}. \charlotte{can we prove a new version of this?} The main point of Theorem \ref{th:weyllawforapotential} is that we can cover the case of unbounded potentials $V$. 

%The proof of the Theorem \ref{th:weyllawforapotential} is inspired by Weyl's proof strategy for the Weyl law for constant potentials on bounded domains with Dirichlet boundary conditions \cite{weyl,weyl1912asymptotische} combined with the Cwikel-Lieb-Rozenblum type bound (Theorem \ref{th:weightednorm}). Indeed, we may use Theorem  \ref{th:weightednorm} to approximate the potential $V$ by a nicer one $V_n$, and then implement the usual semiclassical approximation for $V_n$. More precisely, by the min-max principle (see Section \ref{ss:minmax}), we have
%\begin{align}\label{eq:weylvvncomp-intro}
%    N\left(-\Delta^N_\Omega + \lambda V \right) \leq N\left((1 - \delta) \left(-\Delta^N_\Omega\right) + \lambda V_n\right) + N \left(\delta \left(-\Delta^N_\Omega\right) + \lambda \left(V - V_n\right)\right) ,
%   \end{align}
%where we can take $n\to \infty$ first and $\delta\to 0$ later. See Section \ref{ss:strweyl} for details.

\subsection{Main ingredients of  Theorem \ref{th:weightednorm}}\label{s:ingredient} Now let us explain the proof strategy of our main result  Theorem \ref{th:weightednorm}. Let us focus on the case $d \geq 3$. 

\bigskip
Our general approach is inspired by the method of Rozenblum \cite{rozenblum1972} where the number of bound states is bounded using techniques from microlocal analysis. The main idea is to first localize the Schr\"odinger operator $-\Delta+V$ in small domains such that it has  at most one bound state in each domain, and then put these local bounds together by a covering lemma. In the present paper, since we have to deal with H\"older domains with Neumann boundary conditions, we need to deal with ``oscillatory domains'' when working close to the boundary, and in particular we need to introduce a new Poincaré-Sobolev inequality and a new Besicovitch-type covering lemma for those domains.

\bigskip

To be precise, while Rozenblum \cite{rozenblum1972} works with cubes $Q \subset \R^d$, we will work with ``oscillatory domains'' $D \subset \Omega$ which are given either by cubes if they are far enough away from $\partial \Omega$, or rectangles intersected with $\Omega$ if they are close to the boundary. Since $\Omega$ is a $\gamma$-Hölder domain, the classical Poincaré-Sobolev inequality fails in general for oscillatory domains $D$ with $\overline{D}\cap \partial \Omega \neq \emptyset $. Therefore, we will develop a Poincaré-Sobolev inequality for those domains (see Corollary \ref{co:pssmallM}), which involves the ratio of the two side-lengths of the rectangle and a $L^{p^*}$-norm for some $p^* = p^*(d,\gamma)>\frac{d}{2}$.  % than in the classical Poincaré-Sobolev inequality. \\
Consequently, 
%Using our new Poincaré-Sobolev inequality, 
we show that for oscillatory domains $D$ that are small enough in a suitable sense, we have 
\begin{equation}\label{eq:evvkv1leq1}
    \evv{D}{K V 1_D} \leq 1 . 
\end{equation}
Here, $K = K(d,\gamma) \in \N$ is the constant from our Besicovitch-type covering theorem for oscillatory domains (see Lemma \ref{le:cover}). Then the total number of bound states is controlled by a counting argument, eventually leading to the weighted norm
%
%%As in the proof by Rozenblum, we estimate the number of negative eigenvalues of $-\Delta_\Omega^N + V$ by the number of oscillatory domains we obtain after applying the Besicovitch-type covering lemma (Lemma \ref{le:cover}). \\
%
%\bigskip
%Now let us turn to the technical details. For Rozenblum's proof, a suitable quantity to measure the ``size'' of a cube $Q$ was $\no{V}{\frac{d}{2}, Q}$. For Hölder-domains with Neumann boundary conditions, it turns out that a convenient choice is
\begin{equation}\label{eq:triplepfstrdef}
    \normiii{V}:=\no{V}{\tilde{p}, \beta, D} := \no{V}{\frac{d}{2}, D} + |V|_{\tilde{p}, \beta, D},
\end{equation}
where $\tilde{p} = \tilde{p}(d, \gamma) > \frac{d}{2}$, $\beta = \beta(d, \gamma)> 0$ and $|V|_{\tilde{p}, \beta}$ is a weighted $L^{\tilde{p}}(\Omega)$-seminorm with a weight supported near the boundary of $\Omega$ that grows at a rate determined by $\beta$ as one approaches the boundary, see \eqref{eq:ptdef} and Definition \ref{de:hxlseminorm} for the precise definitions.

%For oscillatory domains $D$ which are far enough away from the boundary, we can simply follow the proof strategy of Rozenblum using $\nor{V}{\frac{d}{2}, D}{}$ as before. \\

\bigskip
In \eqref{eq:triplepfstrdef} we need a $\tilde{p} > p^*$ for the following reason.  Due to the Hölder-regularity of $\partial \Omega$, our oscillatory domains close to the boundary will in many cases look like very narrow rectangles intersected with $\Omega$, so the ratio of the two side-lengths of these rectangles is very far away from one. This influences the constant in the Poincaré-Sobolev inequality (Corollary \ref{co:pssmallM}). Using Hölder's inequality, we get
\begin{equation}
    \no{V}{p^*, D} \leq \no{V}{\tilde{p}, D} \no{1}{\tilde{r}, D}=\no{V}{\pt, D}| D |^{\frac{1}{\rt}}
\end{equation}
for $\tfrac{1}{p^*} = \tfrac{1}{\tilde{p}} + \tfrac{1}{\tilde{r}}$. The quantity $| D |^{\frac{1}{\rt}}$, which is relatively small for these narrow rectangles, cancels the effect of a growing constant in the Poincaré-Sobolev inequality as the rectangle gets more narrow. \\

At this point, it seems natural to measure the ``size'' of oscillatory domains close to the boundary by $ \no{V}{\tilde{p}, D}$.
% \begin{equation}
%     \no{V}{\tilde{p}, D} .
% \end{equation}
Following \cite{rozenblum1972},  we would get \eqref{eq:evltsim1plusintsomega-intro} but 
%
%
%\begin{equation}\label{eq:evltsim1plusintsomega}
%    \ev{\Omega} \lesssim 1 + \int_\Omega |V|^\frac{d}{2} + \int_{\mathrm{close\ to\ } \partial \Omega} |V|^{\tilde{p}} ,
%\end{equation}
%compare with \cite{frank2010equivalence} combined with the Sobolev embedding theorem $H^1(\om)\hookrightarrow L^\qs(\om)$ for $\gamma$-Hölder domains $\om$ with $\qs=\qs(d,\ga)\in \left(2,\frac{2d}{d-2}\right)$ \cite{labutin}.
%\\ \\
%However, s
it does not capture the correct semiclassical behaviour. % we would expect (at least under suitable assumptions), namely \eqref{eq:clrnicesemib}, see \cite{netrusov2005weyl} for $V = 1_\Omega$.
% \begin{equation}
%     \evv{\Omega}{\lambda V} \lesssim 1 + \lambda^{\frac{d}{2}}
% \end{equation}
% for a constant in $\lesssim$ that may depend on $V$, see \cite{netrusov2005weyl} for $V = 1_\Omega$. \\
In order to get the desired semiclassical behaviour for the parts close to the boundary as well, the key idea is to count the number of oscillatory domains $\left\{D_j\right\}_{j \in J_3}$, with some index set  $J_3$, which are narrow rectangles intersected with $\Omega$ close to the boundary, by using a convexity argument. More precisely, we introduce coefficients $A_j > 0$ depending only on the larger side-length of the corresponding rectangle and the distance of the centre of the oscillatory domain $D_j$ to the boundary measured in a suitable sense. Now for suitably chosen $s, s' \in\left(1, \infty\right)$ with $\tfrac{1}{s} + \tfrac{1}{s'} = 1$, we apply Hölder's inequality for sums of products of real numbers to get
\begin{equation}
    |J_3| = \sum_{j \in J_3} A_j^{-1} A_j \le \Big(\sum_{j \in J_3} A_j^{-s'}\Big)^{\frac{1}{s'}} \Big(\sum_{j \in J_3} A_j^{s}\Big)^{\frac{1}{s}} .
\end{equation}
The coefficients $A_j$ will be chosen in such a way that
\begin{equation}\label{eq:ajs}
    A_j^s \lesssim |V|_{\tilde{p}, \beta, D_j}^{\tilde{p}} 
\end{equation}
for all $j\in J_3$. Therefore,
\begin{equation}\label{eq:ajssum}
    \sum_{j \in J_3} A_j^s \lesssim \sum_{j \in J_3} |V|^{\tilde{p}}_{\tilde{p}, \beta, D_j} \lesssim |V|^{\tilde{p}}_{\tilde{p}, \beta} \lesssim \nor{V}{\tilde{p}, \beta}{\tilde{p}} .
\end{equation}
In the proof of Rozenblum, the cubes $Q$ were chosen such that $\|V\|_{\frac d 2, Q} \le 1$.
% \begin{equation}
%     \no{V}{\frac{d}{2}, Q} \sim 1
% \end{equation}
% for some small constant in $\sim$. \\
By contrast, here the oscillatory domains $\left\{D_j\right\}_{j \in J_3}$ are chosen such that $\no{V}{\tilde{p}, D_j}$ is significantly larger than one if the distance of the centre of the oscillatory domain $D_j$ to the boundary is significantly larger than the largest side-length of the corresponding rectangle. This is possible, even though we need to ensure $\evv{D_j}{K V 1_{D_j}} \leq 1$ since in fact, we gain something in the Poincaré-Sobolev inequality for oscillatory domains because we consider an $L^{\tilde{p}}$-norm instead of an $L^{p^*}$-norm. The fact that $\no{V}{\tilde{p}, D_j}$ can be significantly larger than one in certain circumstances combined with the weight that grows near the boundary in the definition of $|V|_{\tilde{p}, \beta}$ allow us to choose $A_j$ so large that 
% \begin{equation}
%     A_j^s \lesssim |V|^{\tilde{p}}_{\tilde{p}, \beta, D_j}
% \end{equation}
% and
\begin{equation}\label{eq:ajmssum}
    \sum_{j \in J_3} A_j^{-s'} \lesssim \nor{V}{\tilde{p}, \beta}{-\frac{1}{2}} 
\end{equation}
while \eqref{eq:ajs} holds.
We choose all parameters in such a way that
\begin{equation}
    \frac{\tilde{p}}{s} - \frac{1}{2s'} = \frac{d}{2} ,
\end{equation}
so by \eqref{eq:ajssum} and \eqref{eq:ajmssum},
\begin{equation}
    |J_3| \leq \left(\sum_{j \in J_3} A_j^{-s'}\right)^\frac{1}{s'} \left(\sum_{j \in J_3} A_j^{s}\right)^\frac{1}{s} \lesssim \nor{V}{\tilde{p}, \beta}{-\frac{1}{2s'}} \nor{V}{\tilde{p}, \beta}{\frac{\tilde{p}}{s}} = \nor{V}{\tilde{p}, \beta}{\frac{d}{2}} ,
\end{equation}
which has the desired semiclassical behaviour. \\

Combining these computations with the computations for oscillatory domains far enough away from the boundary, we obtain the following.  

\begin{theorem}[Precise version of Theorem \ref{th:weightednorm}]\label{th:main}
Let $d\in\N$ with $d\ge 2$, $\ga\in\left[\frac{d-1}{d},1\right)$, $c>0$, $0<\ho<1$, $L\in\N$ and let $\emptyset\neq\om\subset(0,1)^d$ be a $\ga$-Hölder domain with  parameters $c$,  $\ho,L$. Define
\begin{equation} \label{eq:ptdef}
\be:=\frac{1}{d+1}\left(\dgo\right)\left[\frac{1}{d}\left(\dgo\right)^2-d\right]>0,\quad \pt:=\frac{1}{2d}\left(\dgo\right)^2>\frac{d}{2}.
\end{equation}
Let $V:\om\to(-\infty,0]$ be measurable and such that $\no{V}{\pt ,\be} < \infty$. Define
\begin{equation}\label{eq:dzdef}
\dz:=\dz(V):=\min\left(\frac{\ho}{\sqrt{d}},\no{V}{\pt,\be}^{-\frac{1}{2}}\right) \le 1.
\end{equation}
Then
\begin{equation}\label{eq:thpreciseconclusion}
\ev{\om}\ls\dz^{-d}\, ,
\end{equation}
where the constant in the inequality may depend on $\dep$.
\end{theorem}

The definition of a ``$\ga$-Hölder domain with  parameters $c$,  $\ho,L$'' will be given in Definition  \ref{de:gahodomain}. 
%, every $\ga$-Hölder domain with $\ga\in\left(0,1\right]$ is a $\ga$-Hölder domain with  parameters $c$,  $\ho,L$ for some  $c>0$, $\ho>0$, $L\in\N$. 
%To this end, recall that we always assumed that $\ga$-Hölder domains are bounded. 
The condition $\om\subset(0,1)^d$ is not a restriction since we can recover the same result for all (bounded)  $\ga$-Hölder domains by scaling.

\bigskip
Note that from \eqref{eq:thpreciseconclusion} in Theorem \ref{th:main}, we can deduce the first claim \eqref{eq:clrthm} of Theorem \ref{th:weightednorm}.
%        \begin{equation}
%            \ev{\Omega} \leq C\left(1 + \no{V}{\pt,\be}^\frac{d}{2}\right) 
%        \end{equation}
%        for a constant $C = C(d, \gamma, \Omega)> 0$, which was claimed in Theorem \ref{th:weightednorm}. % with $\normiii{\cdot}=\no{\cdot}{\pt,\be}$.
Moreover, the second claim of Theorem \ref{th:weightednorm} follows from the  following corollary.
%Theorem \ref{th:weightednorm} also states that if $\gamma \in \left[\tfrac{2(d-1)}{2d - 1} , 1 \right)$, then $\no{V}{\pt,\be}$ can be replaced by $\no{V}{p}$ for any $p$ large enough depending only on $d$ and $\gamma$. This is the statement of the following corollary.
\begin{corollary}\label{co:main}
If $\gamma \in \left[\tfrac{2(d-1)}{2d - 1} , 1 \right)$ and $p>{\frac{\tilde{p}}{1-\beta}}$, where $\pt$ and $\be$ were defined in Theorem  \ref{th:main},  then
\begin{equation}
\ev{\om}\ls1+\no{V}{p}^{\frac{d}{2}}\, ,
\end{equation}
where the constant in the inequality may depend on $\dep,\, p$.
\end{corollary}
%Theorem \ref{th:main} combined with Corollary \ref{co:main} imply Theorem \ref{th:weightednorm}. 

The details of the proofs of Theorem \ref{th:main} and  Corollary \ref{co:main} can be found in Section \ref{s:conclclr}. Let us give the key ingredients of the proof of Theorem \ref{th:main}  below. First, we have

%Theorem \ref{th:weightednorm}.
%
%\subsubsection{Precise statement of Theorem \ref{th:weightednorm}}\label{ss:precise} Now let us go to details. To begin with, let us state a precise version of Theorem \ref{th:weightednorm}. In this subsection, we explain why it suffices to prove Theorem \ref{th:main} and  Corollary \ref{co:main} in order to show Theorem \ref{th:weightednorm}.

%\subsubsection{Key ingredient for the proof of Theorem \ref{th:main}}\label{ss:key}

%Now we state the key ingredient of the proof of Theorem \ref{th:main} as well as two further results which can be deduced from the key ingredient. Moreover, we prove Theorem \ref{th:main} assuming these results.
%\\ \\

%\bigskip

%The following lemma is the key ingredient for the proof of Theorem \ref{th:main}. 
\begin{lemma}\label{le:familiesabc}
There exists $K=K(\dwc)\in\N$ such that for any measurable $V:\om\to(-\infty,0]$ with $\no{V}{\pt ,\be} < \infty$ there are families $\cF_1,\ldots,\cF_K$ of oscillatory domains $D\subset\om$ such that the following properties are satisfied:
\begin{enumerate}
[(a)]
\item The oscillatory domains $D$ in every $\cF_k$ with $k\in\{1,\ldots,K\}$ are disjoint and
\begin{equation}
\om=\bigcup_{k=1}^K\dot{\bigcup_{D\in\cF_k}}D\, .
\end{equation}

\item For every $k\in\{1,\ldots,K\}$ and every  $D\in\cF_k$ we have
\begin{equation}
\ev{D}\le1\, .
\end{equation}

\item Moreover, for an implicit constant depending on $\dep$, 
\begin{equation}
\sum_{k=1}^K|\cF_k|\ls\dz^{-d}.
\end{equation}
%where the constant in the inequality may depend on $\dep$.
\end{enumerate}
\end{lemma}

To prove Lemma \ref{le:familiesabc}, we need a new covering lemma (see Lemma \ref{le:cover} for details). From Lemma \ref{le:familiesabc}, we can deduce the following two lemmata. 

\begin{lemma}[Selfadjointness of the operator $-\Delta_{\om}^{N}+V$]\label{le:selfadj}
    Let $V:\om\to(-\infty,0]$ be measurable and such that $\no{V}{\pt ,\be} < \infty$. Then the operator $-\Delta_{\om}^{N}+V$ is a selfadjoint operator, which is bounded from below and has the $H^1(\om)$ norm as quadratic form norm.
\end{lemma}

\begin{lemma}\label{le:evkv}
Let $V:\om\to(-\infty,0]$ be measurable and such that $\no{V}{\pt ,\be} < \infty$. Let $K=K(\dwc)\in\N$ and the families $\cF_1,\ldots,\cF_K$ of oscillatory domains be chosen as in Lemma \ref{le:familiesabc}. Then
\begin{equation}
\evv{\om}{\frac{1}{K}V}\ls\dz^{-d}\, .
\end{equation}
\end{lemma}

The proofs of Lemmas \ref{le:familiesabc}, \ref{le:selfadj} and \ref{le:evkv} can be found in Section \ref{ss:pfthmain}. We are now ready to prove Theorem \ref{th:main} assuming those lemmata. %s \ref{le:familiesabc}, \ref{le:selfadj} and \ref{le:evkv}.
\begin{proof}[Proof of Theorem \ref{th:main}]
Let $K=K(\dwc)\in\N$ be chosen as in Lemma \ref{le:familiesabc}. Now apply Lemma \ref{le:evkv} to $KV$ to get
%\begin{align*}
%\ev{\om}&=\evv{\om}{\frac{1}{K}KV}\ls \dz(KV)^{-d}=\left[\min\left(\frac{\ho}{\sqrt{d}},\no{KV}{\pt,\be}^{-\frac{1}{2}}\right)\right]^{-d}\\
%&\ls \left[\min\left(\frac{\ho}{\sqrt{d}},\no{V}{\pt,\be}^{-\frac{1}{2}}\right)\right]^{-d}=\dz(V)^{-d}\, ,
%\end{align*}
% \begin{align*}
% \ev{\om}&=\evv{\om}{\frac{1}{K}KV}\ls \dz(KV)^{-d}=\left[\min\left(\frac{\ho}{\sqrt{d}},\no{KV}{\pt,\be}^{-\frac{1}{2}}\right)\right]^{-d}\\
% &\ls \left[\min\left(\frac{\ho}{\sqrt{d}},\no{V}{\pt,\be}^{-\frac{1}{2}}\right)\right]^{-d}=\dz(V)^{-d}\, ,
% \end{align*}
\begin{equation}
\ev{\om}=\evv{\om}{\frac{1}{K}KV}\ls \dz(KV)^{-d}
% =\left[\min\left(\frac{\ho}{\sqrt{d}},\no{KV}{\pt,\be}^{-\frac{1}{2}}\right)\right]^{-d}\\
% &\ls \left[\min\left(\frac{\ho}{\sqrt{d}},\no{V}{\pt,\be}^{-\frac{1}{2}}\right)\right]^{-d}=
\ls\dz(V)^{-d}\, ,
\end{equation}
where we used that $K$ only depends on $\dwc$ in the second last step.
\end{proof}
% \\ \\
% \charlotte{maybe this can be used elsewhere}
% In this subsection, we give an overview of the structure of the rest of this section. In Section  \ref{ss:holderdo}, we introduce $\ga$-Hölder domains with  parameters $c$,  $\ho,L$. In Section \ref{ss:oscdo}, we define oscillatory domains. In Section \ref{ss:psoscdo}, we prove a Poincaré-Sobolev inequality for oscillatory domains. We use this inequality in Section \ref{ss:choiceoscdo} to prove that $\ev{D_x}\le1$ for our choice of oscillatory domains. In Section \ref{ss:covoscdo}, we prove a  Besicovitch type covering lemma for oscillatory domains. Using this covering lemma, we obtain that the number of oscillatory domains we choose close to the boundary of $\om$ is bounded by a constant times $\dz^{-d}$. In Section \ref{ss:covintcub}, we cover the part of $\om$ far enough away from the boundary with cubes of a suitable size similar to the proof of Rozenblum. Again, we show that the number of cubes we choose is bounded by a constant times $\dz^{-d}$. Combining the results of the previous two subsections, we prove Lemma \ref{le:familiesabc} in Section \ref{ss:pflefamabc}.

%\bigskip

{\bf Structure of the paper.} 
% {\bf rewrite ...} In Section \ref{s:ingredient} we explain the main ingredients of the proofs of Theorems \ref{th:example}, \ref{th:weightednorm} and \ref{th:weyllawforapotential}. 
Sections \ref{s:preliminaries} to \ref{s:conclclr} are devoted to proving the Cwickel-Lieb-Rozenblum type bound (Theorem \ref{th:weightednorm}). In Section \ref{s:preliminaries}, we introduce some preliminaries, including the definition and some basic properties of oscillatory domains. We prove several estimates for oscillatory domains in Section \ref{s:estosc}. In Section \ref{s:cov}, we prove a covering lemma for oscillatory domains. In Section \ref{s:conclclr}, we combine the results from Sections \ref{s:preliminaries} to \ref{s:cov} and prove Theorem \ref{th:weightednorm}. In Section \ref{se:weyln}, we prove Weyl's law for Schrödinger operators on Hölder domains (Theorem \ref{th:weyllawforapotential}). In Section \ref{se:example}, we construct an example with non-semiclassical behaviour and therby prove Theorem \ref{th:example}.

\medskip

{\bf Acknowledgements.}  The author would like to express her deepest gratitude to Phan Thành Nam for his continued support and very helpful advice. She would also like to thank Laure Saint-Raymond for her support and hospitality at Institut des Hautes Études Scientifiques and for inspiring discussions. She would like to thank L\'aszl\'o Erd\H{o}s for useful remarks. She would like to thank Luca Maio for helping prepare the tex file. The author acknowledges the support from the Deutsche Forschungsgemeinschaft (DFG project Nr.~426365943), from the Jean-Paul Gimon Fund and from the Erasmus+ programme.

\bigskip
{\bf Notation.} For two numbers $A,B \ge 0$, we write
$
A\ls  B$
if  $A\le CB$ for some constant $C>0$ which may depend on $\dep$ (see Definition \ref{de:gahodomain}).  Similarly we write $A\sim  B$ if $A\ls  B$ and  $A\gs B$.
%
%
%
%
% and $A\gs B$ is defined similarly. In this paper, the constant $C>0$ may depend on $\dep$ (see Definition \ref{de:gahodomain}) unless the precise dependence is specified in the corresponding claim. 
%%We  also write
%%\begin{equation}
%%A\gs  B \quad \textrm{if }A\ge CB\quad \textrm{for some constant } C>0\,,
%%\end{equation}
%%which may depend on $\dep$. 
%Similarly, we write
%\begin{equation}
%A\sim  B \quad \textrm{if }C_1B\le A\le C_2 B\quad \textrm{for some constants } C_1,C_2>0\,,
%\end{equation}
%where the constants $C_1,C_2>0$ may depend on $\dep$. 

\section{Preliminaries}\label{s:preliminaries}

In this section we collect some technical definitions and preliminary results.

%\subsection{Min-max principle}\label{ss:minmax}
%
%Let us quickly recall the min-max principle, which will be used in several places. For $V:\om\to\R$ measurable, we denote by $\evv{\om}{V}$ the number of negative eigenvalues of the operator $-\Delta_{\om}^{N}+V$. We will show that all the operators $-\Delta_{\om}^{N}+V$  we work with in this paper are self-adjoint operators on $L^2(\om)$ with quadratic form domain $u\in H^1(\om)$. 
%
%\bigskip
%
%By the min-max principle \cite[Theorem  12.1, version 3]{liebloss}, $\evv{\om}{V}$ is the largest number $N\in\N$ such that there exists an $N$-dimensional subspace $M\subset H^1(\om)$ with
%% orthogonal family $\{u_n\}_{n=1}^N$ in $L^2(\om)$ with $u_n\in H^1(\om)$ for all $n\in\{1,\ldots,N\}$ and 
%% such that
%\begin{equation}
%\int_\om|\nabla u|^2+\int_\om  V |u|^2<0\quad \textrm{for all } u\in M\, .
%\end{equation}
%
%%\bigskip
%
%By the min-max principle \cite[Theorem 12.1, version 2]{liebloss}\footnote{As the proof \cite[Theorem 12.1, version 2]{liebloss} shows, in fact, the subspace $M$ need not be a subset of $H^1(\Omega)$ but it suffices to take $M\subset L^2(\Omega)$.}, $\evv{\om}{V}$ is also given by the smallest number $N\in\N$ such that there exists a subspace $M\subset L^2(\om)$ of dimension $N$ such that 
%\begin{equation}
%\int_\om|\nabla u|^2+\int_\om  V |u|^2\ge0
%\end{equation}
%for all $u\in H^1(\om)$ that are in the orthogonal complement of $M$ with respect to $L^2(\om)$.
%

\subsection{Hölder domains}\label{ss:holderdo}
%In this subsection  we define $\ga$-Hölder domains with  parameters $c$,  $\ho,L$ and we introduce notation related to these, which we will use in the rest of the paper.
%
%\bigskip

We often write elements in $\R^d$ as $x=(x', x_d)$, where $x'\in\R^{d-1}$ and $x_d\in\R$. We denote the infinity norm on $\R^{d-1}$ by $\noi{\cdot}$. We have the following technical definition which agrees with the definition of H\"older domains given around \eqref{eq:Holder-intro}.

%That is, for all $x'=(x_1,\ldots,x_{d-1})\in\R^{d-1}$, we define
%\begin{equation}
%\noi{x'}:=\max_{i\in\{1,\ldots,d-1\}}|x_i|.
%\end{equation}

\begin{definition}[$\ga$-Hölder domain with parameters $c,\ho,L$]\label{de:gahodomain}
Let $d\in\N$ with $d\ge 2$, $\ga\in\left(0,1\right]$, $c>0$, $\ho>0$, $L\in\N$ and let $\emptyset\neq\om\subset\R^d$ and be a bounded open set. We call $\om$ a \emph{$\ga$-Hölder domain with  parameters $c$,  $\ho,L$} if there exists a collection $\{O_l\}_{l=1}^L\subset\R^d$ of open sets covering $\partial\om$ with the following properties:
\begin{enumerate}[(i)]
\item For every $l\in\{1,\ldots,L\}$ there exists an orthogonal map $R_l$ and a translation map $T_l$ such that for $\ol:=O_l\cap\om$ we have 
\begin{equation}\label{eq:ol}
\ol=T_lR_l\left\{(x',x_d)\in\R^{d-1}\times\R\bigm| x\in\old,\,0<x_d<f_l(x')\right\rbrace
\end{equation}
and 
\begin{equation}
O_l\setminus\ol\subset T_lR_l\left\{(x',x_d)\in\R^{d-1}\times\R\bigm| x\in\old,\,f_l(x')\le x_d\right\rbrace
\end{equation}
for some open sets $\old\subset\R^{d-1}$ and for a function $f_l:\old\to(3\ho,\infty)$ with
\begin{equation}
|f_l(x')-f_l(y')|\le c\noi{x'-y'}^\ga\ \textrm{for all } x',y'\in\old\, .
\end{equation}
\item For every $l\in\{1,\ldots,L\}$ define
\begin{equation}\label{eq:olhddef}
\olhd :=\left\lbrace x'\in\old\Bigm| \dist{x',\partial\old}>2\ho\right\rbrace, 
\end{equation}
where $ \dist{x',\partial\old}$ denotes the distance from $x'$ to $\partial\old$ with respect to the Euclidean norm in $\R^{d-1}$, and
\begin{equation}
\olh: =T_lR_l\left\{(x',x_d)\in\R^{d-1}\times\R\bigm| x'\in\olhd,\,2\ho<x_d<f_l(x')\right\rbrace\, .
\end{equation}
Then
\begin{equation}\label{eq:coverbryol}
\partial\om\subset\bigcup_{l=1}^L\overline{\olh}\, .
\end{equation}
\end{enumerate}
\end{definition}

%\begin{remark}\label{re:gahodomain}
%Every $\ga$-Hölder domain with $\ga\in\left(0,1\right]$ is a $\ga$-Hölder domain with  parameters $c$,  $\ho,L$ for some  $c>0$, $\ho>0$, $L\in\N$. This can be seen as follows: Recall that all $\ga$-Hölder domains are bonded. By compactness of the boundary, we can cover it by finitely many open sets $\ol$ as in Definition \ref{de:gahodomain}(i), where we can choose $c>0$ and $\ho>0$ independently of $l$ since there are only finitely many sets $\ol$. By possibly decreasing $\ho$, we can also make sure that condition (ii) is satisfied.
%\end{remark}

\begin{definition}[$\oltd$ and $\olt$]\label{de:olt}
For every $l\in\{1,\ldots,L\}$ define
\begin{align}
\oltd &:=\left\lbrace x'\in\old\Bigm| \dist{x',\partial\old}>\ho\right\rbrace,\\
\olt &: =T_lR_l\left\{(x',x_d)\in\R^{d-1}\times\R\bigm| x'\in\oltd,\,\ho<x_d<f_l(x')\right\rbrace,\\
\om^b_\epsilon &:=\left\lbrace x\in\om\bigm| \dist{x,\partial\om}<\eps\right\rbrace,
\end{align}
where $\dist{x,\partial\om}$ denotes the distance with respect to the Euclidean norm in $\R^d$.
\end{definition}

\begin{lemma}\label{oltobho}
Let $\olt$ and $\om^b_\epsilon$  be defined as in Definition \ref{de:olt}. Then
\begin{equation}
\bigcup_{l=1}^L\olt\supset\ob{\ho}.
\end{equation}
\end{lemma}
\begin{proof}
Let $x\in\om$ with $\dist{x,\partial\om}<\ho$. Then there exists $y\in\partial\om$ with $|x-y|=\dist{x,\partial\om}$. By in Definition \ref{de:gahodomain}(ii), we know \eqref{eq:coverbryol},
% that $\partial\om\subset\bigcup_{l=1}^L\olh$, 
so there exists $l\in\{1,\ldots,L\}$ with $y\in\partial\olh$. Without loss of generality, we may assume that $T_l$, $R_l$ are the identity map. Hence, we have $y=\left(y',f_l(y')\right) $ with $y'\in\olhd$. It follows that $\dist{y',\partial\old}>2\ho$ by  \eqref{eq:olhddef}. By the triangle inequality, we obtain $\dist{x',\partial\old}>\ho$, so $x'\in\old$. Furthermore, by the definition of $f_l$, we have $f_l(y')>3\ho$. Again, by the triangle inequality, this implies $x_d>2\ho>\ho$. It follows that $x\in\olt$.
\end{proof}

\begin{definition}[$h_{x,l}$, $\nos{\cdot}{p,\be}{}$ and $\no{\cdot}{p,\be}$]\label{de:hxlseminorm}
Let $d\ge2$ and let $\om\subset\R^d$ be a $\ga$-Hölder domain with  constant $c>0$ and parameters $\ho>0$ and $L\in\N$ for some $\ga\in (0,1]$. In the following, we use the notation from Definition \ref{de:gahodomain}.
\begin{enumerate}[(i)]
\item For any $l\in\{1,\ldots,L\}$ and $x=T_lR_l(x',x_d)\in\ol$, define
\begin{equation}\label{eq:hdef}
h_x: =h_{x,l}: =f_l(x')-x_d>0\,.
\end{equation}
Moreover, for any $x\in\bigcup_{l=1}^L\ol$, we let
\begin{equation}
\hxm: =\min_{l\in\{1,\ldots,L\} \textrm{with } x\in\ol}h_{x,l}\,.
\end{equation}
\item For $\be>0$ and $p\in[1,\infty)$, define the seminorm $\nos{\cdot}{p,\be}{}$ by 
\begin{equation}
\nos{f}{p,\be}{p}: =\inx{\bigcup_{l=1}^L\ol}{x	}\hxm^{-\be}|f(x)|^p
\end{equation}
for all measurable functions $f:\om\to\bC$.
\item Define the norm $\no{\cdot}{p,\be}$ by
\begin{equation}
    \no{f}{p,\be}: =\no{f}{\frac{d}{2},\om}+\nos{f}{p,\be}\qquad \textrm{if }d\ge3\,, \quad \no{f}{p,\be}: =\no{f}{\B,\om}+\nos{f}{p,\be}\qquad \textrm{if }d=2
\end{equation}
% \begin{equation}
% \no{f}{p,\be}: =\no{f}{\frac{d}{2},\om}+\nos{f}{p,\be}\qquad \textrm{if }d\ge3
% \end{equation}
% and 
% \begin{equation}
% \no{f}{p,\be}: =\no{f}{\B,\om}+\nos{f}{p,\be}\qquad \textrm{if }d=2
% \end{equation}
for all measurable functions $f:\om\to\bC$. See \cite[p.~1]{frank2019bound} for the definition of $\no{f}{\B,\om}$.
\end{enumerate}
\end{definition}

\subsection{Oscillatory domains}\label{ss:oscdo}
In this subsection we define oscillatory domains and prove several properties of oscillatory domains.
\begin{definition}[$c_0$, $c_1$ and $c_2$]\label{de:c0c1} Let $c>0$ be as in Definition \ref{de:hxlseminorm}. We define
\begin{equation}
c_0: =\left[\min\left(\frac{1}{c_1},\frac{2^\ga}{64c},\frac{1}{2^{\ga+3}c}\right)\right]^\gai, \quad c_1: = 16, \quad c_2: = c_0c_1^\gai.
\end{equation}
%and 
%\begin{equation}
%c_1: = 16
%\end{equation}
%and
%\begin{equation}
%c_2: = c_0c_1^\gai
%\end{equation}
\end{definition}

% \begin{lemma}[Properties of $c_0$ and $c_1$]\label{le:c0c1}
% We have
% \begin{enumerate}[(i)]
% \item \begin{equation}
% c_0^\ga c_1\le1 \, .
% \end{equation}
% \item \begin{equation}
% \frac{4}{c_1}=\frac{1}{4}\, .
% \end{equation}
% \item \begin{equation}
% \frac{cc_0^\ga c_1}{2^\ga}\le\frac{1}{4}\, .
% \end{equation}
% \item \begin{equation}
% 2^{\ga+1}cc_0^\ga\le\frac{1}{4}\, .
% \end{equation}
% \end{enumerate}
% \end{lemma}
% \begin{proof}
% {\bf Proof of (i). } 
% Since $c_1= 16$, we know by the definition of $c_0$ that 
% \begin{equation}
% c_0\le\left(\frac{1}{c_1}\right)^\gai\, ,
% \end{equation}
% which shows $c_0^\ga c_1\le1$.
% \\
% \\
% {\bf Proof of (ii). } 
% Since $c_1= 16$, we have
% \begin{equation}
% \frac{4}{c_1}=\frac{4}{16}=\frac{1}{4}\, .
% \end{equation}
% \\
% \\
% {\bf Proof of (iii). } 
% By the definition of $c_0$, we have
% \begin{equation}
% c_0^\ga\le\frac{2^\ga}{64c}=\frac{2^\ga}{16\cdot4c}=\frac{2^\ga}{c_1\cdot4c}\, .
% \end{equation}
% \\
% \\
% {\bf Proof of (iv). } 
% By the definition of $c_0$, we have
% \begin{equation}
% c_0^\ga\le\frac{1}{2^{\ga+3}c}=\frac{1}{4c\cdot 2^{\ga+1}}\, .
% \end{equation}
% \end{proof}

\begin{definition}[Oscillatory domain $D$]\label{de:smalldomain}
For every $l\in\{1,\ldots,L\}$, $x\in\olt$ and $\de\in(0,\dz] $ define
\begin{equation}
a: =a_{x}: =a_{x}(\de): =a_{x,l}(\de): =\min\left(\de,c_0\max\left(h_{x,l},c_1\de\right)^{\frac{1}{\ga}}\right)\, ,
\end{equation}
where $c_0 $, $c_1 $ were defined in Definition \ref{de:c0c1}. Define $D: =D_x: =D_x(\de): =D_{x,l}(\de) $ by
\begin{equation}\label{eq:ddef}
D: =T_lR_l\left\{(y',y_d)\in\R^{d-1}\times\R\Bigm|  |y'-x'|<\frac{1}{2}a,\ |y_d-x_d|<\frac{1}{2}\de,\ f_l(y')>y_d\right\rbrace\, .
\end{equation}
\end{definition}

\begin{lemma}\label{le:Dwelldef}
Let $l\in\{1,\ldots,L\}$, $x\in\olt$ and $\de\in(0,\dz] $. Then
\begin{enumerate}[(i)]
\item $D=D_x$ is well-defined and $D\subset\ol$.
\item For all $y=T_lR_l(y',y_d)\in D $, we have
\begin{equation}\label{eq:flgede14}
f_l(y')\ge x_d-\frac{\de}{4}\, .
\end{equation}
\end{enumerate}
\end{lemma}

\begin{proof}
{\bf Proof of (i). } Let $y\in D $. Then $y=T_lR_l(y',y_d) $ for some $(y',y_d)\in\R^{d-1}\times\R$ with $|y'-x'|<\frac{1}{2}a\le\frac{1}{2}\ho,\ |y_d-x_d|<\frac{1}{2}\de\le\frac{1}{2}\ho$ and $ f_l(y')>y_d$. Here we used that $a\le\de\le\dz\le\frac{\ho}{\sqrt{d}}\le\ho$. Since $x'\in\oltd$ and $|y'-x'|<\frac{1}{2}\ho$, we get $y'\in\old$. Thus, $f_l(y')$ is well-defined and therefore, $D $ is well-defined. Since $x\in\olt$, we have $x_d>\ho$, so using $|y_d-x_d|<\frac{1}{2}\ho$, we get $y_d>0$. To sum up, we have shown that $y'\in\old$ and $ f_l(y')>y_d>0$, so $y\in\ol$ by \eqref{eq:ol}.
\\
\\
{\bf Proof of (ii). } Let $y=T_lR_l(y',y_d)\in D $. Then 
\begin{equation}
|f_l(x')-f_l(y')|\le c\noi{x'-y'}^\ga<c\left(\frac{1}{2}a\right)^\ga=\frac{c}{2^\ga}a^\ga\, ,
\end{equation}
so
\begin{align}\label{eq:flyge}
f_l(y')\ge f_l(x')-|f_l(x')-f_l(y')|\ge f_l(x')-\frac{c}{2^\ga}a^\ga=h+x_d-\frac{c}{2^\ga}a^\ga\, ,
\end{align}
% \begin{align}
% \begin{split}\label{eq:flyge}
% f_l(y')&\ge f_l(x')-|f_l(x')-f_l(y')|\ge f_l(x')-x_d+x_d-\frac{c}{2^\ga}a^\ga=h+x_d-\frac{c}{2^\ga}a^\ga\\
% &=h-\frac{\de}{4}+\frac{\de}{4}+x_d-\frac{c}{2^\ga}a^\ga\, ,
% \end{split}
% \end{align}
where we used \eqref{eq:hdef}. Thus, in order to show \eqref{eq:flgede14}, it suffices to show that
\begin{equation}\label{eq:dhaggez}
\frac{\de}{4}+h-\frac{c}{2^\ga}a^\ga\ge0\, ,
\end{equation}
If $a=c_0h^{\frac{1}{\ga}}$ or $a=c_0(c_1\de)^{\frac{1}{\ga}}$, then \eqref{eq:dhaggez} holds by
\begin{equation}\label{eq:el4}
\max \Big( 2^{\ga+1}cc_0^\ga, \frac{cc_0^\ga c_1}{2^\ga} \Big) \le\frac{1}{4}\, .
\end{equation}
If $a=\de$, then $\de\le c_0h^{\frac{1}{\ga}}$ by the definition of $a$, and it reduces to the case $a=c_0h^{\frac{1}{\ga}}$. \end{proof}
% {\bf Case 1: $a=c_0h^{\frac{1}{\ga}}$. } 
% We have
% \begin{equation}
% \frac{\de}{4}+h-\frac{c}{2^\ga}a^\ga=\frac{\de}{4}+h-\frac{c}{2^\ga}\left( c_0h^{\frac{1}{\ga}}\right)^\ga=\frac{\de}{4}+h-\frac{cc_0^\ga}{2^\ga}h\ge0
% \end{equation}
% since $\frac{cc_0^\ga}{2^\ga}\le 1$ by Lemma \ref{le:c0c1}(iv).

% {\bf Case 2: $a=c_0(c_1\de)^{\frac{1}{\ga}}$. } 
% We have
% \begin{equation}
% \frac{\de}{4}+h-\frac{c}{2^\ga}a^\ga=\frac{\de}{4}+h-\frac{c}{2^\ga}\left( c_0(c_1\de)^{\frac{1}{\ga}}\right)^\ga=\frac{\de}{4}+h-\frac{cc_0^\ga c_1}{2^\ga}\de\ge0
% \end{equation}
% since $\frac{cc_0^\ga c_1}{2^\ga}\le \frac{1}{4}$ by Lemma \ref{le:c0c1}(iii).

% {\bf Case 3: $a=\de$. } 
% By the definition of $a$, we have $\de\le c_0h^{\frac{1}{\ga}}$, so 
% \begin{equation}
% \frac{\de}{4}+h-\frac{c}{2^\ga}a^\ga=\frac{\de}{4}+h-\frac{c}{2^\ga}\de^\ga\ge\frac{\de}{4}+h-\frac{c}{2^\ga}\left( c_0h^{\frac{1}{\ga}}\right)^\ga\ge0
% \end{equation}
% by case 1.

\begin{lemma}\label{le:propD}
Let $l\in\{1,\ldots,L\}$, $x\in\olt$ and $\de_x\in(0,\dz] $. Then the following holds true:
\begin{enumerate}[(i)]
\item If $a_x=\de_x$ or $a_x=c_0h_x^{\frac{1}{\ga}}$, then $D_x$ is a cuboid, namely
\begin{equation}
D_x: =T_lR_l\left\{(w',w_d)\in\R^{d-1}\times\R\bigm|  |w'-x'|<\frac{1}{2}a_x,\ |w_d-x_d|<\frac{1}{2}\de_x\right\rbrace\, .
\end{equation}
\item In addition, let $y,z\in\olt$ with $\de_y,\de_z\in(0,\dz] $ and assume that
\begin{equation}
a_w=c_0h_w^{\frac{1}{\ga}}\ \textrm{for all } w\in\left\lbrace x,y,z\right\rbrace \, .
\end{equation}
Furthermore, assume that $\de_x\le2\de_y$, $\de_z\le2\de_y$, $D_x\cap D_y\neq\emptyset$ and $D_x\cap D_z\neq\emptyset$. 
Then
\begin{equation}\label{eq:hcompforacohga}
\frac{1}{2}h_y\le h_w\le2h_y\ \textrm{for all } w\in\left\lbrace x,z\right\rbrace \, .
\end{equation}
\item If $a_x=c_0h_x^{\frac{1}{\ga}}$, then
\begin{equation}
\frac{1}{2}h_x\le h_w\le2h_x\ \textrm{for all } w\in D_x \, .
\end{equation}
\item If $a_x=c_2\de_x^{\frac{1}{\ga}}$, then
\begin{equation}
|h_w-h_x|\le\de_x\ \textrm{for all } w\in D_x \, .
\end{equation}
\item If $a_x=c_0\max\left( h_x,c_1\de\right)^{\frac{1}{\ga}}$ and  $f:\om\to\bC$ is measurable, then
\begin{equation}
\nos{f}{\pt,\be,D_x}{\pt}\gs\max\left( h_x,c_1\de_x\right)^{-\be}\nor{f}{\pt,D_x}{\pt} \, .
\end{equation}
\end{enumerate}
\end{lemma}

\begin{proof}
For simplicity of notation, we write $a:=a_x$, $\de:=\de_x$, $h:=h_x$,  $D:=D_x$ in the proof. 
\\ 
{\bf Proof of (i). } Since $c_0c_1^\gai\le1$ and $\dz\le1$, we know by the definition of $a$ and by $a=\de$ or $a=c_0h^{\frac{1}{\ga}}$ that $h\ge c_1\de$. For all $w=T_lR_l(w',w_d)$ with $|w'-x'|<\frac{1}{2}a$ and $ |w_d-x_d|<\frac{1}{2}\de$, we have
\begin{align}
\begin{split}\label{eq:flwge}
f_l(w')&\ge h+x_d-\frac{c}{2^\ga}a^\ga\ge h+x_d-\frac{c}{2^\ga}\left(c_0h^{\frac{1}{\ga}}\right)^\ga\\
&=x_d+h\left(1-\frac{cc_0^\ga}{2^\ga}\right)\ge x_d+\frac{1}{2}h\ge x_d+\frac{c_1}{2}\de= x_d+8\de\ge w_d+7\de \, ,
\end{split}
\end{align}
% \begin{align}
% \begin{split}\label{eq:flwge}
% f_l(w')&\ge h+x_d-\frac{c}{2^\ga}a^\ga\ge h+x_d-\frac{c}{2^\ga}\left(c_0h^{\frac{1}{\ga}}\right)^\ga=h+x_d-\frac{cc_0^\ga}{2^\ga}h\\
% &=x_d+h\left(1-\frac{cc_0^\ga}{2^\ga}\right)\ge x_d+\frac{1}{2}h\ge x_d+\frac{c_1}{2}\de\ge x_d+2\de\, ,
% \end{split}
% \end{align}
where we used \eqref{eq:flyge} in the first step, $a=\de$ or $a=c_0h^{\frac{1}{\ga}}$ in the second step, \eqref{eq:el4} in the fourth step, and $c_1=16$ in the second last step. By \eqref{eq:ddef}, we get $w\in D$.
% $1-\frac{cc_0^\ga}{2^\ga}\ge\frac{1}{2}$ by Lemma \ref{le:c0c1}(iv) in the fifth step, and $c_1=16$ by Definition the last step. 
% By Lemma \ref{le:Dwelldef}(ii), every $w=T_lR_l(w',w_d)$ with $|w'-x'|<\frac{1}{2}a$ and $ w_d=x_d-\frac{1}{8}\de$ is contained in $D$. Using \eqref{eq:flwge} and the definition of $D$, we deduce that every $w=T_lR_l(w',w_d)$ with $|w'-x'|<\frac{1}{2}a$ and $ |w_d-x_d|<\frac{1}{2}\de$ is contained in $D$.
\\
{\bf Proof of (ii). } First note that by $a_w=c_0h_w^{\frac{1}{\ga}}$, we have $h_w\ge c_1\de_w$ for all $w\in\left\lbrace x,y,z\right\rbrace$. Also note that in order to show \eqref{eq:hcompforacohga}, it suffices to show that
\begin{equation}\label{eq:hwhymaxineq}
|h_w-h_y|\le\frac{1}{2}\max\left(h_w,h_y\right)\, .
\end{equation}
% This can be seen as follows:
% 
% {\bf Case 1: $h_w\ge h_y$. } 
% Then
% \begin{equation}
% h_w\le h_y+|h_w-h_y|\le h_y+\frac{1}{2}\max\left(h_w,h_y\right)=h_y+\frac{1}{2}h_w\, ,
% \end{equation}
% so $\frac{1}{2}h_w\le h_y$, that is, $h_w\le 2h_y$. We get
% \begin{equation}
% \frac{1}{2}h_y\le h_y\le h_w\le2h_y\, .
% \end{equation}

% {\bf Case 2: $h_w\le h_y$. } 
% Then
% \begin{equation}
% h_w\ge h_y-|h_w-h_y|\le h_y-\frac{1}{2}\max\left(h_w,h_y\right)=h_y-\frac{1}{2}h_y=\frac{1}{2}h_y\, ,
% \end{equation}
% so we get
% \begin{equation}
% \frac{1}{2}h_y\le h_w\le h_y\le2h_y\, .
% \end{equation}
% \\
% \\
% Thus, it suffices to show \eqref{eq:hwhymaxineq}. 
Let us begin by showing \eqref{eq:hwhymaxineq} for $w=x$. Since $D_x\cap D_y\neq\emptyset$, we have  
\begin{equation}
|x_d-y_d|<\ha\de_x+\ha\de_y\le\ha2\de_y+\ha\de_y\le2\de_y
\end{equation}
and
\begin{equation}
|x'-y'|<\ha a_x+\ha a_y=\frac{c_0}{2}\left(h_x^\gai+h_y^\gai\right)\le c_0\left[\max\left(h_x,h_y\right)\right]^\gai\, .
\end{equation}
We obtain
\begin{align*}
|h_x-h_y|&=|f_l(x')-x_d-\left(f_l(y')-y_d\right)|\le |f_l(x')-f_l(y')|+|x_d-y_d|\le c|x'-y'|^\ga+2\de_y\\
&\le cc_0^\ga \max\left(h_x,h_y\right)+\frac{2}{c_1}h_y\le\left(cc_0^\ga +\frac{2}{c_1}\right)\max\left(h_x,h_y\right)\le\ha\max\left(h_x,h_y\right)\, .
\end{align*}
% \begin{align*}
% |h_x-h_y|&=|f_l(x')-x_d-\left(f_l(y')-y_d\right)|\le |f_l(x')-f_l(y')|+|x_d-y_d|\\
% &\le c|x'-y'|^\ga+2\de_y\le c\left(c_0\left[\max\left(h_x,h_y\right)\right]^\gai\right)^\ga+2\de_y\\
% &\le cc_0^\ga \max\left(h_x,h_y\right)+\frac{2}{c_1}h_y\le\left(cc_0^\ga +\frac{2}{c_1}\right)\max\left(h_x,h_y\right)\\
% &\le\ha\max\left(h_x,h_y\right)\, .
% \end{align*}
Here we used $h_y\ge c_1\de_y$ in the second last step and we used \eqref{eq:el4} and $c_1=16$ in the last step. 
% $cc_0^\ga\le\frac{1}{4}$ by Lemma \ref{le:c0c1}(iv) and $c_1=16$ by Definition in the last step. 
This shows \eqref{eq:hwhymaxineq} for $w=x$.
\\
\\
Next, let us show \eqref{eq:hwhymaxineq} for $w=z$. Since $D_x\cap D_y\neq\emptyset$ and $D_x\cap D_z\neq\emptyset$, we have 
\begin{equation}
|z_d-y_d|\le |z_d-x_d|+|x_d-y_d|< \ha\de_z+\de_x+\ha\de_y\le\ha2\de_y+2\de_y+\ha\de_y\le4\de_y
\end{equation}
% \begin{align*}
% |z_d-y_d|&\le |z_d-x_d|+|x_d-y_d|<\ha\de_z+\ha\de_x+\ha\de_x+\ha\de_y\le \ha\de_z+\de_x+\ha\de_y\\
% &\le\ha2\de_y+2\de_y+\ha\de_y\le4\de_y
% \end{align*}
and
\begin{align*}
|z'-y'|&\le |z'-x'|+|x'-y'|<\ha a_z+ a_x+\ha a_y\le2\max\left(a_x,a_y,a_z\right)\\
&=2c_0\max\left(h_x^\gai,h_y^\gai,h_z^\gai\right)\le 2c_0\max\left(\left(2h_y\right)^\gai,h_y^\gai,h_z^\gai\right)\le 2^{1+\gai}c_0\max\left(h_y^\gai,h_z^\gai\right)\, ,
\end{align*}
% \begin{align*}
% |z'-y'|&\le |z'-x'|+|x'-y'|<\ha a_z+\ha a_x+\ha a_x+\ha a_y=\ha a_z+ a_x+\ha a_y\\
% &\le2\max\left(a_x,a_y,a_z\right)=2c_0\max\left(h_x^\gai,h_y^\gai,h_z^\gai\right)\le 2c_0\max\left(\left(2h_y\right)^\gai,h_y^\gai,h_z^\gai\right)\\
% &\le 2^{1+\gai}c_0\max\left(h_y^\gai,h_z^\gai\right)\, ,
% \end{align*}
where we used $h_x\le2h_y$ in the second last step. We obtain
\begin{align*}
&\qquad |h_z-h_y|=|f_l(z')-z_d-\left(f_l(y')-y_d\right)|\le |f_l(z')-f_l(y')|+|z_d-y_d|\\
&\le c|z'-y'|^\ga+4\de_y\le c\left(2^{1+\gai}c_0\left[\max\left(h_x,h_y\right)\right]^\gai\right)^\ga+4\de_y\\
&\le 2^{\ga+1}cc_0^\ga \max\left(h_x,h_y\right)+\frac{4}{c_1}h_y\le\left(2^{\ga+1}cc_0^\ga +\frac{4}{c_1}\right)\max\left(h_x,h_y\right)\le\ha\max\left(h_x,h_y\right)\, .
\end{align*}
where we used \eqref{eq:el4} and $c_1=16$ in the last step.
% $2^{\ga+1}cc_0^\ga\le\frac{1}{4}$ by Lemma \ref{le:c0c1}(iv) and $c_1=16$ by Definition in the last step. 
This shows \eqref{eq:hwhymaxineq} for $w=z$.
\\
{\bf Proof of (iii). } Let $w\in D_x $. By $a=c_0h^{\frac{1}{\ga}}$, we have $h\ge c_1\de$. As in the proof of (ii), we get
\begin{align*}
&\qquad |h_w-h|=|f_l(w')-w_d-\left(f_l(x')-x_d\right)|\le |f_l(w')-f_l(x')|+|w_d-x_d|\\
&\le c|w'-x'|^\ga+\ha\de\le c\left(\ha a\right)^\ga+\frac{1}{2c_1}h=c\left(\ha c_0h^{\frac{1}{\ga}}\right)^\ga+\frac{1}{2c_1}h=\frac{cc_0^\ga}{2^\ga}h+\frac{1}{2c_1}h\le\ha h\, ,
\end{align*}
where we used \eqref{eq:el4} and $c_1=16$ in the last step.
% where we used $\frac{cc_0^\ga}{2^\ga}\le\frac{1}{4}$ by Lemma \ref{le:c0c1}(iv) and $c_1=16$ by Definition in the last step. 
This shows \eqref{eq:hcompforacohga}.
% \begin{equation}
% h_w\in\left[\ha h,\frac{3}{2}h\right]\subset\left[\ha h,2h\right]\, .
% \end{equation}
\\
\\
{\bf Proof of (iv). } Let $w\in D_x $. By $a=c_2\de^{\frac{1}{\ga}}$, we have $h\le c_1\de$. As in the proof of (iii), we get
\begin{align*}
|h_w-h|&=|f_l(w')-w_d-\left(f_l(x')-x_d\right)|\le |f_l(w')-f_l(x')|+|w_d-x_d|\\
&\le c|w'-x'|^\ga+\ha\de\le c\left(\ha a\right)^\ga+\ha\de=\frac{cc_0^\ga c_1}{2^\ga}\de+\ha\de\le\de\, ,
\end{align*}
where we used \eqref{eq:el4} and $c_1=16$ in the last step.
% \begin{align*}
% |h_w-h|&=|f_l(w')-w_d-\left(f_l(x')-x_d\right)|\le |f_l(w')-f_l(x')|+|w_d-x_d|\\
% &\le c|w'-x'|^\ga+\ha\de\le c\left(\ha a\right)^\ga+\ha\de=c\left(\ha c_2\de^{\frac{1}{\ga}}\right)^\ga+\ha\de\\
% &=\frac{cc_0^\ga c_1}{2^\ga}\de+\ha\de\le\de\, ,
% \end{align*}
% where we used $c_2=c_0c_1^\gai $ by Definition \ref{de:c0c1} in the second last step, $\frac{cc_0^\ga c_1}{2^\ga}\le\frac{1}{4}$ by Lemma \ref{le:c0c1}(iii) and $c_1=16$ by Definition \ref{de:c0c1} in the last step. 
\\
\\
{\bf Proof of (v). } Let $w\in D_x $. Then
\begin{align*}
h_w&=|f_l(w')-w_d|\le h+|f_l(w')-w_d-\left(f_l(x')-x_d\right)|\le h+|f_l(w')-f_l(x')|+|w_d-x_d|\\
&\le h+c|w'-x'|^\ga+\ha\de\le h+c\left(\ha c_0\max\left( h_x,c_1\de\right)^{\frac{1}{\ga}}\right)^\ga+\ha\de\ls\max\left( h,c_1\de\right)\, ,
\end{align*}
By the definition of $\vert f\vert_{\pt,\be,D}^{\pt}$, we have
\begin{equation}
\vert f\vert_{\pt,\be,D}^{\pt}=\inx{D}{y}\hym^{-\be}|f(y)|^{\pt}\gtrsim \max\left( h,c_1\de\right)^{-\be}\nor{f}{\pt,D}{\pt}\, .
\end{equation}
\end{proof}

\section{Estimates for oscillatory domains}\label{s:estosc}

\subsection{Poincaré-Sobolev inequality for oscillatory domains}\label{ss:psoscdo}
In this subsection we prove a Poincaré-Sobolev inequality for oscillatory domains. The following Lemma and its proof is a version of \cite{labutin} for oscillatory domains (Corollary \ref{co:pssmallM}).
\begin{lemma}[Sobolev inequality for oscillatory domains]\label{le:sobolevdirichlet}
Let $c>0$, $\ga\in(0,1)$ and $\de\in(0,1)$. Let $\at>0$ and let $f:\left[-\at/2, \at/2\right]^{d-1}\to\left(0,\de\right)
$
% \begin{equation}
% f:\left[-\frac{\at}{2},\frac{\at}{2}\right]^{d-1}\to\left(0,\de\right)
% \end{equation}
be such that
\begin{equation}
|f(x')-f(y')|\le c\noi{x'-y'}^\ga\ \textrm{for all } x',y'\in\left[-\tfrac{\at}{2},\tfrac{\at}{2}\right]^{d-1}\,.
\end{equation}
Define $\qs\in(2,\infty)$ by
\begin{equation}
\frac{1}{\qs}:=\ha-\frac{1}{\dgo}
\end{equation}
and define
\begin{equation}
\td: =\left\{(x',x_d)\in\left[-\tfrac{\at}{2},\tfrac{\at}{2}\right]^{d-1}\times(0,\de)\,\middle\vert\, f(x')>x_d\right\rbrace\, .
\end{equation}
Define the ``straight part I'm nots of the boundary of $\td$'' by
\begin{equation}
\tb: =\partial \td\setminus\left\{(x',x_d)\in\left[-\tfrac{\at}{2},\tfrac{\at}{2}\right]^{d-1}\times(0,\de)\,\middle\vert\, f(x')=x_d\right\rbrace\, .
\end{equation} 
Then there exists a constant $\cs=\cs(d,c,\ga)>0$ such that for all $u\in H^1(\td)$ with $u\restriction_\tb\equiv0$ 
in the trace sense, we have
\begin{equation}\label{eq:sobolevdirichlet}
\nor{u}{\qs,\td}{2}\le\cs\nor{\nabla u}{2,\td}{2}\,.
\end{equation}
\end{lemma}
\begin{proof}
We first consider the case of smooth functions  $u$ and later deduce \eqref{eq:sobolevdirichlet} for all $u\in H^1(\td)$ with $u\restriction_\tb\equiv0$ in the trace sense. Let $u\in C^\infty(\td)\cap H^1(\td)$ with $u\restriction_\tb\equiv0$. Let
\begin{equation} \label{eq:psi-def}
\psi:\left(0,\infty\right)\to\R\,,\quad \psi(s)=c^{-\gai}s^{\gai}\,.
\end{equation}
Fix $\tilde x\in \td$. For simplicity of notation, we shift the coordinate system and reflect the last coordinate by replacing every $x\in\R^d$ by  $(x_1-\tilde x_1,\ldots, x_{d-1}-\tilde x_{d-1},-x_d+\tilde x_d)\in\R^d$, so without loss of generality, we may assume  $\tilde x=0$. Define 
\begin{equation} \label{eq:K-def}
K: =\left\lbrace(x',x_d)\in\R^{d-1}\times\R\,\middle\vert\, 0<x_d\,,\  0\le |x'|<\psi(x_d) \right\rbrace\, .
\end{equation}
and note that $u$ it is well-defined on $K$ since we can extend it by zero on $K\setminus \td$ due to $u\restriction_\tb\equiv0$ 
in the trace sense. Let $y'\in\R^{d-1}$ with $|y'|<1$ and note that $(y'\psi(x_d),x_d)\in K$ for all $x_d\in(0,\de)$. Since $u$ is smooth, we can apply Newton's theorem and  $u((y'\psi(\de),\de))=0$ to obtain
\begin{align*}
-u(0) &=\inu{0}{\de}{x_d}{\frac{\partial}{\partial x_d}u(y'\psi(x_d),x_d)}\\
&=\inu{0}{\de}{x_d}\left(\sum_{j=1}^{d-1}(\partial_ju)(y'\psi(x_d),x_d)y_j'\psi'(x_d)+(\partial_du)(y'\psi(x_d),x_d)\right)\,.
\end{align*}
Here $\psi'$ denotes the derivative of $\psi$ while $y'\in\R^{d-1}$. Integrating over $y'\in B_1^{(d-1)}(0)$ and using the notation $\odo=|B_1^{(d-1)}(0)|$, we get
\begin{align*}
&\qquad-\odo u(0)\\
 &=\inx{B_1^{(d-1)}(0)}{y'}\inu{0}{\de}{x_d}\left(\sum_{j=1}^{d-1}(\partial_ju)(y'\psi(x_d),x_d)\frac{y_j'\psi(x_d)}{\psi(x_d)} \psi'(x_d)+(\partial_du)(y'\psi(x_d),x_d)\right)\\
&=\inx{K}{x}\frac{1}{\psi(x_d)^{d-1}}\left(\sum_{j=1}^{d-1}(\partial_ju)(x)\frac{x_j'}{\psi(x_d)} \psi'(x_d)+(\partial_du)(x)\right)\,,
\end{align*}
where we used the change of variables $x=(y'\psi(x_d),x_d)$. Note that
\begin{equation}
\psi'(s)=\frac{1}{c^\gai}\frac{1}{\ga}s^{\gai-1}\,.
\end{equation}
Thus, if $x=(x',x_d)\in K$, then using $x_d<\de\in(0,1)$, we get
\begin{equation}\label{eq:psdfacbd}
\left\vert \frac{x_j'}{\psi(x_d)} \psi'(x_d)\right\vert\le\left\vert  \psi'(x_d)\right\vert=\frac{1}{c^\gai}\frac{1}{\ga}{x_d}^{\gai-1}\le \frac{1}{c^\gai}\frac{1}{\ga}{\de}^{\gai-1}\ls 1\,.
\end{equation}
Therefore, by \eqref{eq:psdfacbd} and the Cauchy-Schwarz inequality on $\R^d$, 
\begin{align*}
\odo |u(0)|&\le\inx{K}{x}\frac{1}{\psi(x_d)^{d-1}}\sum_{j=1}^{d}\left\vert(\partial_ju)(x)\right\vert\ls \inx{K}{x}\frac{1}{\psi(x_d)^{d-1}}\sqrt{\sum_{j=1}^{d}\left\vert(\partial_ju)(x)\right\vert^2}\sqrt{d}\\
&\ls\inx{K}{x}\frac{1}{\psi(x_d)^{d-1}}|\nabla u(x)|\,.
\end{align*}
We get
\begin{equation}\label{eq:estuzerosimcoo}
|u(0)|\ls\inx{K}{x}\frac{1}{\psi(x_d)^{d-1}}|\nabla u(x)|\,,
\end{equation}
where the constant in the inequality only depends on $\deps,\, c$. Let us now undo the change of variables, which was convenient for the above computation. In the old coordinate system \eqref{eq:estuzerosimcoo} reads
\begin{equation}
|u(\tilde x)|\ls\inx{\td}{y}1_K(\tilde x-y)\frac{1}{\psi(\tilde x_d-y_d)^{d-1}}|\nabla u(y)|=\left(\left(1_K\frac{1}{\psi(\cdot_d)^{d-1}}\right)*|\nabla u|\right)(\tilde x)
\end{equation}
for all $\tilde x\in \td$. Here we used the rotational symmetry of $K$ with respect to the first $d-1$ variables. We define $r\in (1,\infty)$ by
\begin{equation}
1+\frac{1}{\qs}=1+\ha-\frac{1}{\dgo}=:\ha+\frac{1}{r}\,.
\end{equation}
% \begin{equation}
% 1+\frac{1}{\qs}=1+\ha-\frac{1}{\dgo}=\ha+1-\frac{1}{\dgo}=\ha+\frac{\dg}{\dgo}=:\ha+\frac{1}{r}\,.
% \end{equation}
By the weak Young inequality and \eqref{eq:estuzerosimcoo}, we get
\begin{equation}
\no{u}{\qs}\ls\no{\nabla u}{2}\no{1_K\frac{1}{\psi(\cdot_d)^{d-1}}}{r,w}\,.
\end{equation}
Here $\no{\cdot}{r,w}$ denotes the weak $L^r$-norm. % $\no{g}{r,w}: =\sup_{\tau>0}\left(\tau\left\vert\{|g|>\tau\}\right\vert^{\frac{1}{r}}\right)$. 
From the definitions of $\psi$ and $K$ in \eqref{eq:psi-def} and  \eqref{eq:K-def}, we have
\begin{equation}
\no{1_K\frac{1}{\psi(\cdot_d)^{d-1}}}{r,w}\ls1,
\end{equation}
which completes the proof for smooth $u$. The claim for general $u$ follows a standard density argument. \end{proof}

\begin{lemma}[Poincaré inequality for oscillatory domains with Neumann boundary conditions]\label{le:poincare}
Let $\de>0$, let $\ah\in\left(\frac{\de}{3},\de\right)$ and let $f:\left[-\ah/2,\ah/2\right]^{d-1}\to\left[\de/4,\de\right)$
% \begin{equation}
% f:\left[-\frac{\ah}{2},\frac{\ah}{2}\right]^{d-1}\to\left[\frac{\de}{4},\de\right)
% \end{equation}
be continuous. Define
\begin{equation}
\hd: =\left\{(x',x_d)\in\left[-\tfrac{\ah}{2},\tfrac{\ah}{2}\right]^{d-1}\times\left(-\tfrac{\de}{4},\de\right)\,\middle\vert\, f(x')>x_d\right\rbrace\, .
\end{equation}
Then there exists a constant $\cp=\cp(d)>0$ such that for all $u\in H^1(\hd)$ with $\int_{\hd} u=0$, we have
\begin{equation}
\cp\frac{1}{\de^2}\nor{u}{2, \hd}{2}\le\nor{\nabla u}{2,\hd}{2}\,.
\end{equation}
\end{lemma}
\begin{proof}
The proof can be found in \cite[Lemma 2.6(2)]{netrusov2005weyl} for slightly different side lengths of the domain.
\end{proof}

\begin{corollary}[Poincaré-Sobolev inequality for oscillatory domains]\label{co:pssmallM}
Let $M\ge1$, $\de\in(0,1)$ and let $f:\left[-\de/(2M),\de/(2M)\right]^{d-1}\to\left[\de/4,\de\right)$
% \begin{equation}
% f:\left[-\frac{\de}{2M},\frac{\de}{2M}\right]^{d-1}\to\left[\frac{\de}{4},\de\right)
% \end{equation}
be such that 
\begin{equation}
|f(x')-f(y')|\le c\noi{x'-y'}^\ga\ \textrm{for all } x',y'\in\left[-\tfrac{\de}{2M},\tfrac{\de}{2M}\right]^{d-1}
\end{equation}
for some $c>0$, $\ga\in(0,1)$. Define $\qs\in(2,\infty)$ by
\begin{equation}
\frac{1}{\qs}=\ha-\frac{1}{\dgo}\,.
\end{equation}
and define
\begin{equation}
D: =\left\{(x',x_d)\in\left(-\tfrac{\de}{2M},\tfrac{\de}{2M}\right)^{d-1}\times(0,\de)\,\middle\vert\, f(x')>x_d\right\rbrace\, .
\end{equation}
Then there exists a constant $\cps=\cps(d,c,\ga)>0$ such that for all $u\in H^1(D)$ with $\int_{D} u=0$, we have
\begin{equation}
\nor{u}{\qs,D}{2}\le\cps M^{(d-1)\left(1-\frac{2}{\qs}\right)}\nor{\nabla u}{2,D}{2}\,.
\end{equation}
\end{corollary}

\begin{proof}
We combine Lemma \ref{le:sobolevdirichlet} and Lemma \ref{le:poincare}. Define $\tm$ as the largest odd number such that $\tm\le M$ and note that $\tm\le M\le3\tm$, so
\begin{equation}
\ah:= \tfrac{\tm\de}{M}\in\left(\tfrac{\de}{3},\de\right)\,.
\end{equation}
We define
\begin{equation}
\tf:\left[-\tfrac{3\tm\de}{2M} ,\tfrac{3\tm\de}{2M}\right]^{d-1}\to\left[\tfrac{\de}{4},\de\right)
\end{equation}
by reflecting $f$: We can write every $x'\in\left[-\tfrac{3\tm\de}{2M} ,\tfrac{3\tm\de}{2M}\right]^{d-1}$ as
\begin{equation}
x'=\tfrac{\de}{M}z'+w'
\end{equation}
with $z'=(z_1,\ldots ,z_{d-1})\in\Z^{d-1}$ and $w'=(w_1,\ldots ,w_{d-1})\in\left[-\frac{\de}{2M} ,\frac{\de}{2M}\right]^{d-1}$. Define $\tf$ by
\begin{equation}
\tf(x'):=f\left(\left( (-1)^{z_1}w_1,\ldots, (-1)^{z_{d-1}}w_{d-1}\right)\right)\,.
\end{equation}
Note that $\tf$ is well defined, continuous and
\begin{equation}
|\tf(x')-\tf(y')|\le c\noi{x'-y'}^\ga\ \textrm{for all } x',y'\in\left[-\tfrac{3\tm\de}{2M} ,\tfrac{3\tm\de}{2M}\right]^{d-1}\,.
\end{equation}
Define
\begin{equation}
\td: =\left\{(x',x_d)\in\left(-\tfrac{3\tm\de}{2M} ,\tfrac{3\tm\de}{2M}\right)^{d-1}\times\left(-\tfrac{\de}{4},\de\right)\,\middle\vert\, \tf(x')>x_d\right\rbrace
\end{equation}
and
\begin{equation}
\hd: =\left\{(x',x_d)\in\left(-\tfrac{\tm\de}{2M} ,\tfrac{\tm\de}{2M}\right)^{d-1}\times\left(0,\de\right)\,\middle\vert\, \tf(x')>x_d\right\rbrace\, .
\end{equation}
Note that $\hd$ consists of $\tm^{d-1}$ reflected copies of $D$. Moreover, $\td$ consists of less than $2(3\tm)^{d-1}$ reflected copies of $D$ in the sense that $\td\cap\{x_d\ge0\}$ consists of $(3\tm)^{d-1}$ reflected copies of $D$ but $\td\cap\{x_d<0\}$ is only contained in $(3\tm)^{d-1}$ reflected copies of $D$. Let
\begin{align*}
&\phi\in C_c^\infty\left(\left(-\tfrac{3\tm\de}{2M} ,\tfrac{3\tm\de}{2M}\right)^{d-1}\times\left(-\tfrac{\de}{4},\tfrac{5\de}{4}\right)\right)\\
& \textrm{with }  0\le\phi\le1\, ,\ \phi\restriction_{\left[-\tfrac{\tm\de}{2M} ,\tfrac{\tm\de}{2M}\right]^{d-1}\times\left[0,\de\right]}\equiv 1\,\ \textrm{and }   \no{\nabla\phi}{\infty}\ls\frac{1}{\de}  \,.
\end{align*}
It is possible to choose such a $\phi$ by scaling. Note that $\phi\restriction_{\hd}\equiv1$ by the definition of $\hd$.  Let $u\in H^1(D)$ with $\int_{D} u=0$ and define the corresponding reflected version $\tu\in H^1(\td)$ by
\begin{equation}
\tu(x):=u\left(\left( (-1)^{z_1}w_1,\ldots, (-1)^{z_{d-1}}w_{d-1}, |w_d|\right)\right)\,.
\end{equation}
for every $x=\tfrac{\de}{M}z+w\in\td$
% \begin{equation}
% x=\frac{\de}{M}z+w\in\td
% \end{equation}
with $z=(z_1,\ldots , z_{d-1},0)\in\Z^{d}$ and $w=(w_1,\ldots ,w_{d})\in\left[-\frac{\de}{2M} ,\frac{\de}{2M}\right]^{d-1}\times\left(-\frac{\de}{4},\infty\right)$. Note that $\tu$ is well-defined because $H^1$ functions are defined up to almost everywhere equality and since reflections of $H^1$ functions are again $H^1$ functions. Moreover, by $\tu\restriction_{D}\equiv u$, we know that $\tu\restriction_{\td}$ consists of less than $2(3\tm)^{d-1}$ reflected copies of $u$, and $\tu\restriction_{\hd}$ consists of $\tm^{d-1}$ reflected copies of $u$. Also note that $\int_{\hd}\tu=0$. Furthermore, $\phi\tu\in H^1(\td)$, $\phi\tu\restriction_{\hd}=\tu$ and $\phi\tu\restriction_{\tb}\equiv0$, where
\begin{equation}
\tb: =\partial \td\setminus\left\{(x',x_d)\in\left[-\tfrac{3\tm\de}{2M} ,\tfrac{3\tm\de}{2M}\right]^{d-1}\times\left(-\tfrac{\de}{4},\de\right]\,\middle\vert\, \tf(x')=x_d\right\rbrace\, .
\end{equation} 
Since $\int_{\td}\tu=0$, $\ah\in\left(\frac{\de}{3},\de\right)$ and $\td$ satisfies the assumptions of Lemma \ref{le:poincare}, we can apply Lemma \ref{le:poincare} to get
\begin{align*}
\nor{\nabla (\phi\tu)}{2,\td}{2}&\ls\int_\td |\phi|^2|\nabla\tu|^2+\int_\td|\tu|^2|\nabla\phi|^2 \ls\int_\td |\nabla\tu|^2+\frac{1}{\de^2}\int_\td|\tu|^2 \\
&\le2\cdot3^{d-1}\left(\int_\hd |\nabla\tu|^2+\frac{1}{\de^2}\int_\hd|\tu|^2\right)\ls \int_\hd |\nabla\tu|^2=\tm^{d-1}\int_D |\nabla u|^2\,.
\end{align*}
In the third step we used that $0\le\phi\le1$ and $\no{\nabla\phi}{\infty}\ls\frac{1}{\de}$, and in the second last step we used that $\tu\restriction_{\hd}$ consists of $\tm^{d-1}$ reflected copies of $u$. On the other hand, $\phi\tu\restriction_{\tb}\equiv0$ and $\td$ satisfies the assumptions of Lemma \ref{le:sobolevdirichlet}, so by Lemma \ref{le:sobolevdirichlet}, we obtain
\begin{align*}
\nor{\nabla (\phi\tu)}{2,\td}{2}&\gs\nor{\phi\tu}{\qs,\td}{2}\ge\nor{\tu}{\qs,\hd}{2}=\left(\int_\hd |\tu|^\qs\right)^{\frac{2}{\qs}}=\left(\tm^{d-1}\int_D |\tu|^\qs\right)^{\frac{2}{\qs}}\\
&=\tm^{(d-1)\frac{2}{\qs}}\left(\int_D |\tu|^\qs\right)^{\frac{2}{\qs}}=\tm^{(d-1)\frac{2}{\qs}}\nor{u}{\qs,D}{2}\,.
\end{align*}
To sum up, we get
\begin{equation}
\tm^{d-1}\int_D |\nabla u|^2\gs\nor{\nabla (\phi\tu)}{2,\td}{2}\gs\tm^{(d-1)\frac{2}{\qs}}\nor{u}{\qs,D}{2}\,,
\end{equation}
so
\begin{equation}
\tm^{(d-1)\left(1-\frac{2}{\qs}\right)}\int_D |\nabla u|^2\gs\nor{u}{\qs,D}{2}\,.
\end{equation}
Now recall that $M\sim3\tm$, so we obtain the desired result.
\end{proof}

\subsection{Choice of the oscillatory domains}\label{ss:choiceoscdo}
In this subsection, we choose depending on $V$ for every $x\in\om$ close to $\partial\om$ an oscillatory domain $D_x$ with centre $x$ such that
\begin{equation}
\ev{D_x}\le1\, .
\end{equation}
For the proof, we use the Poincaré-Sobolev inequality for oscillatory domains (Corollary \ref{co:pssmallM}). 
At the same time, we choose the oscillatory domains $D_x$ such that a certain norm of $V$ on $D_x$ is not too small. This will be needed in the following subsection.
% For the proof, we use the Poincaré-Sobolev inequality for oscillatory domains Section \ref{ss:psoscdo}  (Corollary \ref{co:pssmallM}). 
% At the same time, we choose the oscillatory domains $D_x$ such that a certain norm of $V$ on $D_x$ is not too small, that is, the corresponding oscillatory domain is not chosen too small. This will be needed in the following subsection.
\begin{lemma}\label{le:bdryev}
Let $l\in\{1,\ldots,L\}$ and let $x\in\olt\cap\obd$. Let $\de\in(0,\dz] $ and let $D:=D_x(\de)$.
\begin{enumerate}[(i)]
\item Suppose $a_x(\de)=\de$ and $\nor{V}{\frac{d}{2}, D}{}\ls1$ if $d\ge3$ and $\no{V}{\cB, D}\ls1$ if $d\ge2$.
% \begin{equation}
% \nor{V}{\frac{d}{2}, D}{\frac{d}{2}}\ls1\ \textrm{if } d\ge3\quad \textrm{and }\quad 
% % \end{equation}
% % and
% % \begin{equation}
% \no{V}{\cB, D}\ls1 \ \textrm{if } d=2\,.
% \end{equation}
Then
\begin{equation}
\ev{D}\le1\, .
\end{equation}
\item Suppose $a_x(\de)=c_0\max\left(h_{x},c_1\de\right)^{\frac{1}{\ga}}$ and
\begin{equation}\label{eq:ptnVassgs}
\nor{V}{\pt, D}{\pt}\ls\max\left(\frac{h_x}{c_1\de} ,1\right)^{\frac{d-1}{\ga}}\, .
\end{equation}
Then
\begin{equation}\label{eq:evoscle1}
\ev{D}\le1\, .
\end{equation}
\end{enumerate}
All the constants in $\ls$ in this Lemma only depend on $\depc$.
\end{lemma}

\begin{proof}
{\bf Proof of (i). } 
Let $0\not\equiv u\in H^1(D)$ with $\int_D u=0$. Then using Hölder's inequality and the Poincaré-Sobolev inequality for cubes, see \cite[Theorem 8.12]{liebloss}, we get for $d\ge3$
\begin{align*}
\int_D|\nabla u|^2+\int_DV|u|^2\ge\int_D|\nabla u|^2-\no{V}{\frac{d}{2}, D}\no{u}{\frac{2d}{d-2}, D}^2 %\ge\nor{\nabla u}{2, D}{2}-\no{V}{\frac{d}{2}, D}\cps\nor{\nabla u}{2, D}{2}\\
\ge \nor{\nabla u}{2, D}{2}\left(1-\cps\no{V}{\frac{d}{2}, D}\right)\, .
\end{align*}
Hence, since $\nabla u\not\equiv0$, we get $\int_D|\nabla u|^2+\int_DV|u|^2>0$ if $\no{V}{\frac{d}{2}, D}<\frac{1}{\cps}$. If $d=2$, we use the Poincaré-Sobolev inequality for Orlicz norms, see \cite[Proposition 2.1]{frank2019bound}.

%\begin{align*}
%\int_D|\nabla u|^2+\int_DV|u|^2&=\nor{\nabla u}{2, D}{2}\left(1+\cps\int_D V\frac{|u|^2}{\cps \nor{\nabla u}{2, D}{2}}\right)\\
%&\ge\nor{\nabla u}{2, D}{2}\left(1-\cps\no{V}{\cB, D}\right)\, .
%\end{align*}
%Hence, since $\nabla u\not\equiv0$, we get $\int_D|\nabla u|^2+\int_DV|u|^2>0$ if $\no{V}{\cB, D}<\frac{1}{\cps}$. 
%\\\\

\bigskip
{\bf Proof of (ii). } 
% Note that if $M$ is chosen as $M:=\de /a$
% % \begin{equation}\label{eq:defM}
% % c_0\max\left(h,c_1\de\right)^{\frac{1}{\ga}}=a_x(\de)=a=:\frac{\de}{M}\,,
% % \end{equation}
% % where we write $h:=h_{x,l}$
% , then 
% \begin{equation}\label{eq:M}
% M=c_0^{-1}\de\max\left(h,c_1\de\right)^{-\frac{1}{\ga}}\sim\de^{1-\gai}\max\left(\frac{h}{c_1\de},1\right)^{-\frac{1}{\ga}}
% \end{equation}
% % \begin{align*}\label{eq:M}
% % M&=c_0^{-1}\de\max\left(h,c_1\de\right)^{-\frac{1}{\ga}}=c_0^{-1}c_1^{-\gai}\de^{1-\gai}\max\left(\frac{h}{c_1\de},1\right)^{-\frac{1}{\ga}}\\
% % &\sim\de^{1-\gai}\max\left(\frac{h}{c_1\de},1\right)^{-\frac{1}{\ga}}
% % \end{align*}
% and by the choice of $a$, we also know that $a\le\de$, that is, $M\ge1$. 
Choose $M:=\de /a$ and note $M\ge1$.
By Lemma \ref{le:Dwelldef}(ii), we know that $D$ is an oscillatory domain as in Corollary \ref{co:pssmallM}.  In order to show $\ev{D}\le1$, it suffices to show that for all $0\not\equiv u\in H^1(D)$ with $\int_D u=0$, we have
\begin{equation}
\int_D|\nabla u|^2+\int_DV|u|^2>0\,.
\end{equation}
To this end, let $0\not\equiv u\in H^1(D)$ with $\int_D u=0$. Define $\ps$ by 
\begin{equation}\label{eq:psdef}
1=\frac{1}{\ps}+\frac{1}{\frac{\qs}{2}}, \quad \text{namely} \quad \frac{1}{\ps}:=1-\frac{2}{\qs}=\frac{2}{\dgo}.
\end{equation}
By Hölder's inequality and the Poincaré-Sobolev inequality for oscillatory domains, see Corollary \ref{co:pssmallM}, we get
$$
\int_D|\nabla u|^2+\int_DV|u|^2\ge\int_D|\nabla u|^2-\no{V}{\ps, D}\no{u}{\qs, D}^2 \ge \nor{\nabla u}{2, D}{2}\left(1-\cps M^{\frac{d-1}{\ps}}\no{V}{\ps, D}\right). 
$$
Hence, since $\nabla u\not\equiv0$, the left-hand side is strictly positive if
\begin{equation}\label{eq:evVintcondMV}
M^{\frac{d-1}{\ps}}\no{V}{\ps, D}<\frac{1}{\cps}\,.
\end{equation}
We define $\rt\in (1,\infty)$ by
\begin{equation}\label{eq:rtdef}
    \frac{1}{\rt}:=\frac{1}{\ps}-\frac{1}{\pt}=\frac{2(d-1)}{\left(\dgo\right)^2}\left(\gai-1\right)\,,
\end{equation}
where $\ps$ and $\pt$ are given in \eqref{eq:psdef} and \eqref{eq:ptdef}. 
% We have
% \begin{align*}
% \frac{1}{\ps}& =\frac{2}{\dgo}=\frac{2d}{\left(\dgo\right)^2}\frac{1}{d}\left(\dgo\right)\\
% &=\frac{2d}{\left(\dgo\right)^2}+\frac{2d}{\left(\dgo\right)^2}\left[\frac{1}{d}\left(\dgo\right)-1\right]\\
% &=\frac{1}{\pt}+\frac{2}{\left(\dgo\right)^2}\left[\left(\dgo\right)-d\right]=\frac{1}{\pt}+\frac{2(d-1)}{\left(\dgo\right)^2}\left(\gai-1\right)\\
% &=:\frac{1}{\pt}+\frac{1}{\rt}\,.
% \end{align*}
% Recall that
% \begin{equation}
% | D |\sim\left(\frac{\de}{M}\right)^{d-1}\de=\de^dM^{-\left(d-1\right)}\,,
% \end{equation}
% so by  Hölder's inequality, we get
% \begin{equation}
% \no{V}{\ps, D}\le\no{V}{\pt, D}| D |^{\frac{1}{\rt}}\sim\no{V}{\pt, D}\left(\de^dM^{-\left(d-1\right)}\right)^{\frac{1}{\rt}} \,.
% \end{equation}
Using  Hölder's inequality with \eqref{eq:rtdef}, $| D |\sim\de^dM^{-\left(d-1\right)}$ and $M=\de /a$, it follows that
\begin{align*}
&\quad M^{\frac{d-1}{\ps}}\no{V}{\ps, D}\ls M^{\frac{d-1}{\ps}}\no{V}{\pt, D}| D |^{\frac{1}{\rt}}\ls M^{\frac{d-1}{\ps}}\no{V}{\pt, D}\left(\de^dM^{-\left(d-1\right)}\right)^{\frac{1}{\rt}}\\
&\sim \no{V}{\pt, D}\de^{\frac{d}{\rt}}\left(\de^{1-\gai}\max\left(\frac{h}{c_1\de},1\right)^{-\frac{1}{\ga}}\right)^{\frac{d-1}{\pt}}=\left(\no{V}{\pt, D}^{\pt}\max\left(\frac{h}{c_1\de},1\right)^{-\dg}\right)^{\frac{1}{\pt}}\ls1\,,
\end{align*}
% Using $\frac{1}{\ps}-\frac{1}{\rt}=\frac{1}{\pt}$ in the second step, it follows that
% \begin{align*}
% M^{\frac{d-1}{\ps}}\no{V}{\ps, D}&\ls M^{\frac{d-1}{\ps}}\no{V}{\pt, D}\left(\de^dM^{-\left(d-1\right)}\right)^{\frac{1}{\rt}}=\no{V}{\pt, D}\de^{\frac{d}{\rt}}M^{\frac{d-1}{\pt}}\\
% &\sim \no{V}{\pt, D}\de^{\frac{d}{\rt}}\left(\de^{1-\gai}\max\left(\frac{h}{c_1\de},1\right)^{-\frac{1}{\ga}}\right)^{\frac{d-1}{\pt}}\\
% &=\no{V}{\pt, D}\de^{\frac{d}{\rt}+\left(1-\gai\right)\frac{d-1}{\pt}}\max\left(\frac{h}{c_1\de},1\right)^{-\dg\frac{1}{\pt}}\\
% &=\no{V}{\pt, D}\max\left(\frac{h}{c_1\de},1\right)^{-\dg\frac{1}{\pt}}\\
% &=\left(\no{V}{\pt, D}^{\pt}\max\left(\frac{h}{c_1\de},1\right)^{-\dg}\right)^{\frac{1}{\pt}}\ls1\,,
% \end{align*}
where we used that 
\begin{equation}
\frac{d}{\rt}+\left(1-\gai\right)\frac{d-1}{\pt}=0
\end{equation}
% \begin{equation}
% \frac{d}{\rt}+\left(1-\gai\right)\frac{d-1}{\pt}=d\cdot\frac{2(d-1)}{\left(\dgo\right)^2}\left(\gai-1\right)+\left(1-\gai\right)(d-1)\frac{2d}{\left(\dgo\right)^2}=0
% \end{equation}
in the third step and the assumption \eqref{eq:ptnVassgs}
% \begin{equation}\label{eq:assVptnorpf}
% \nor{V}{\pt, D}{\pt}\ls\max\left(\frac{h_x}{c_1\de_x} ,1\right)^{\frac{d-1}{\ga}}
% \end{equation}
in the last step. Hence, if the constant in \eqref{eq:ptnVassgs} 
% \eqref{eq:assVptnorpf} 
is chosen small enough, we can deduce \eqref{eq:evVintcondMV}, which is what we wanted to show.
% that
% \begin{equation}
% M^{\frac{d-1}{\ps}}\no{V}{\ps, D}<\frac{1}{\cps}\,,
% \end{equation}
% which is what we wanted to show, see \eqref{eq:evVintcondMV}.
\end{proof}

\begin{lemma}[Choice of the oscillatory domains]\label{le:D123}
Let $l\in\{1,\ldots,L\}$ and let $x\in\olt\cap\obd$. Then there exists $\de_x\in(0,\dz] $ such that for $D:=D_x:=D_x(\de_x)$ we have
\begin{equation}
\ev{D}\le1
\end{equation}
and at least one of the following properties is satisfied:
\begin{enumerate}[(1)]
\item $\de_x=\dz$.
\item 
$a_x=\de_x$ and
\begin{equation}\label{eq:osccond2}
\nor{V}{\frac{d}{2}, D}{\frac{d}{2}}\gs1\ \textrm{if } d\ge3 \quad \textrm{and }\quad
% \end{equation}
% and
% \begin{equation}
\no{V}{\B, D}\gs1\ \textrm{if } d=2 \, .
\end{equation}
\item $a_x=c_0\max\left( h_x,c_1\de_x\right)^{\frac{1}{\ga}}$ and
\begin{equation}\label{eq:osccond3}
\nor{V}{\pt, D}{\pt}\gs\max\left(\frac{h_x}{c_1\de_x} ,1\right)^{\frac{d-1}{\ga}}\, .
\end{equation}
\end{enumerate}
All the constants in $\gs$ in this Lemma only depend on $\depc$.
\end{lemma}

\begin{proof}
% Our goal is to find $\de=\de_x\in(0,\dz]$ such that
% \begin{equation}
% \ev{D}\le1
% \end{equation}
% and such that at least one of the properties $(1)$, $(2)$, $(3)$ are satisfied, where we denote $D=D_x(\de_x)$. 
% \\
% \\
% Let us first explain how to choose $\de$ if $c_0\max\left(h,c_1\dz\right)^{\frac{1}{\ga}}>\dz$. 
% By $c_0c_1^\gai\le1$, 
% % By $c_0^\ga c_1\le1$, see Lemma \ref{le:c0c1}(i), 
% we get $c_0h^{\frac{1}{\ga}}>\dz$ and $h>c_1\dz$. By Lemma \ref{le:propD}(i), $D_x(\dz)$ is a cube of side-length $\dz$ around $x$, which is completely contained in $\om$. If $d\ge3$ and $\nor{V}{\frac{d}{2}, D_x(\dz)}{}\ls1$ or if $d= 2$ and $\no{V}{\cB, D_x(\dz)}\ls1$, pick $\de:=\dz$. Then $(1)$ is satisfied and by Lemma \ref{le:bdryev}(i), we have $\ev{D}\le1$. If $d\ge3$ and $\nor{V}{\frac{d}{2}, D_x(\dz )}{}\gs 1$ or if $d= 2$ and $\no{V}{\cB, D_x(\dz)}\gs1$, where the constant in $\gs$ only depends on $\depc$, pick $\de<\dz$ such that $\nor{V}{\frac{d}{2}, D_x(\de )}{}\sim 1$ if $d\ge3$ and $\no{V}{\cB, D_x(\de )}\sim1$ if $d= 2$. This is possible by the continuity of the map $\de\mapsto \nor{V}{\frac{d}{2}, D_x(\de )}{}$ and $\de\mapsto\no{V}{\cB, D_x(\de )}$, respectively. Here we used  \cite[Lemma A.5]{frank2019bound} for $d= 2$. By Lemma \ref{le:bdryev}(i), we get $\ev{D}\le1$ and it follows directly from the definition of $\de$ that $(2)$ is satisfied.

Let us first explain how to choose $\de$ if $c_0\max\left(h,c_1\dz\right)^{\frac{1}{\ga}}>\dz$. It follows that $a(\de)=\de$ for all $\de\le\dz$. Pick $\de\le\dz$ such that
\begin{equation}\label{eq:leosc123sim}
\nor{V}{\frac{d}{2}, D(\de)}{\frac{d}{2}}\sim1\ \textrm{if } d\ge3 \quad \textrm{and }\quad
\no{V}{\B, D(\de)}\sim1\ \textrm{if } d=2 \, .
\end{equation}
holds, so $(2)$ is satisfied. By \eqref{eq:leosc123sim} and Lemma \ref{le:bdryev}(i), we get $\ev{D}\le1$.

\bigskip

Let us now assume that $c_0\max\left(h,c_1\dz\right)^{\frac{1}{\ga}}\le\dz$. Define $\dc:=c_0h^\gai\le\dz$, so $c_0\max\left(h,c_1\dc\right)^{\frac{1}{\ga}}=\dc$. We have $a(\de)=c_0\max\left(h,c_1\de\right)^{\frac{1}{\ga}}$ for all $\de\in[\dc,\dz]$ and $a(\de)=\de$ is for all $\de\le\dc$. If \begin{equation}
\nor{V}{\pt, D_x(\dz)}{\pt}\ls\max\left(\frac{h}{c_1\dz} ,1\right)^{\frac{d-1}{\ga}}\, ,
\end{equation}
pick $\de:=\dz$ and note that $(1)$ is satisfied. Furthermore, by Lemma \ref{le:bdryev}(ii), we have $\ev{D}\le1$. Else, if there exists $\de\in[\dc,\dz]$ such that
\begin{equation}\label{eq:simdcdz}
\nor{V}{\pt, D(\de)}{\pt}\sim\max\left(\frac{h}{c_1\de} ,1\right)^{\frac{d-1}{\ga}},
\end{equation}
where the constant in $\sim$ only depends on $\depc$, pick this $\de$. By \eqref{eq:simdcdz}, (3) is satisfied and moreover, we have $\ev{D}\le1$ by Lemma \ref{le:bdryev}(ii). Else, since the left hand side of \eqref{eq:simdcdz} is increasing in $\de$ and the right hand side of \eqref{eq:simdcdz} is decreasing in $\de$, we have \eqref{eq:osccond3} for $\de=\dc$. Now if $d\ge3$ and $\nor{V}{\frac{d}{2}, D_x(\dc)}{}\ls1$ or if $d= 2$ and $\no{V}{\cB, D_x(\dc)}\ls1$, pick $\de:=\dc$ and note that (3) is satisfied and $\ev{D}\le1$ by Lemma \ref{le:bdryev}(i). Else, pick $\de<\dc$ such that \eqref{eq:leosc123sim} holds. Thus, (2) is satisfied and $\ev{D}\le1$ by Lemma \ref{le:bdryev}(i).

\end{proof}

\section{Covering lemmas}\label{s:cov}

\subsection{Covering of the part close to the boundary by oscillatory domains}\label{ss:covoscdo}
In this subsection, we prove a Besicovitch type covering lemma for oscillatory domains. It is one of the key ingredients of the proof of Theorem \ref{th:weightednorm} since it allows us to choose a family of oscillatory domains as in the previous subsection, which cover the part of $\om$ close to $\partial\om$, but which do not overlap too much. Using this result, we show that the number of oscillatory domains we choose is bounded by a constant times $\dz^{-d}$. 
\begin{lemma}[Covering lemma]\label{le:cover}
Let $l\in\{1,\ldots,L\}$ and suppose that for every $x\in\olt\cap\obd$ we are given a $\de_x\in(0,\dz] $. We define the oscillatory domains $D_x:=D_x(\de_x)$ as in Definition \ref{de:smalldomain}. Then there exists $K_l=K_l(d,\ga)\in\N$ and subfamilies $\cF_1,\ldots,\cF_{K_l}$ of oscillatory domains $D_x:=D_x(\de_x)\subset\ol\cap\obz$ such that
\begin{enumerate}[(i)]
\item For every $k\in\{1,\ldots,K_l\}$ all oscillatory domains in $\cF_k$ are disjoint.
\item
\begin{equation}
\bigcup_{k=1}^{K_l}\dot{\bigcup_{D\in\cF_k}}D\supset\olt\cap\obd\, .
\end{equation}
\end{enumerate}
\end{lemma}

\begin{proof}
Recall that for every $x\in\olt\cap\obd$ we are given a $\de_x\in(0,\dz]$ and the corresponding $a_x$ is given by
\begin{equation}
a_{x} =\min\left(\de_x,c_0\max\left(h_{x,l},c_1\de_x\right)^{\frac{1}{\ga}}\right)\, .
\end{equation}
Furthermore, the oscillatory domain $D_x=D_x(\de_x)$ is given by
\begin{equation}
D_x: =T_lR_l\left\{(y',y_d)\in\R^{d-1}\times\R\bigm|  |y'-x'|<\frac{1}{2}a,\ |y_d-x_d|<\frac{1}{2}\de_x,\ f_l(y')>y_d\right\rbrace\, .
\end{equation}
Without loss of generality, let us assume that $T_lR_l$ is the identity map. To begin with, we decompose $\olt\cap\obd$ into three parts
\begin{equation}
\olt\cap\obd=A_1\cup A_2\cup A_3
\end{equation}
such that
\begin{align*}
A_1: =\left\lbrace a_x=c_0h_x^\gai\right\rbrace, \quad A_2: =\left\lbrace  a_x=c_2\de_x^\gai\right\rbrace, \quad A_3: =\left\lbrace    a_x=\de_x\right\rbrace\,.
\end{align*}
For each of these sets, we will prove a corresponding covering theorem with a constant $K^l_1$, $K^l_2$, $K^l_3$ depending on $\deps$. Combining these results, we will get the desired result with $K^l:=K^l_1+K^l_2+K^l_3\in\N$.
\\ \\
The beginning of the proof for $A_b$ with $b\in\{1,2\}$ will be similar to the proof of Besicovich's covering theorem for cubes, see for example. Let $b\in\{1,2\}$. We will denote families of oscillatory domains $D_x$ with $x\in A_b$ by  $\left(\cF_k\right)_{k\in\N}$. At the beginning of our construction, the $\cF_k$ are all assumed to be empty. We put oscillatory domains inside those $\cF_k$ according to the following procedure:
\\ \\
First, choose $\tilde x_1\in A_b$ with 
\begin{equation}
\tilde\de_1: =\de_{\tilde x_1}\ge\ha\sup_{x\in A_b}\de_x
\end{equation}
and denote the corresponding domain by $\tilde D_1: =D_{\tilde x_1}$. Put $\tilde D_1$ in $\cF_1$. Then, if possible, choose $\tilde x_2\in A_b\setminus\tilde D_1$ such that 
\begin{equation}
\tilde\de_2: =\de_{\tilde x_2}\ge\ha\sup_{x\in A_b\setminus \tilde D_1 }\de_x
\end{equation}
and denote the corresponding domain by $\tilde D_2: =D_{\tilde x_2}$. If $\tilde D_2\cap\tilde D_1=\emptyset$, put $\tilde D_2$ in $\cF_1$. Otherwise, put $\tilde D_2$ in $\cF_2$. More generally, if $\tilde x_1,\ldots,\tilde x_{n- 1}$ have already been chosen for some $n\in \N$, we proceed as follows: If
\begin{equation}
\bigcup_{m=1}^{n-1}\tilde D_m \supset A_b\,,
\end{equation}
then stop. Else, we can choose $\tilde x_n\in A_b\setminus\bigcup_{m=1}^{n-1}\tilde D_m$ such that 
\begin{equation}\label{eq:denxnco}
\tilde\de_n: =\de_{\tilde x_n}\ge\ha\sup_{x\in A_b\setminus\bigcup_{m=1}^{n-1}\tilde D_m}\de_x
\end{equation}
and denote the corresponding domain by $\tilde D_n: =D_{\tilde x_n}$. Put $\tilde D_n$ in $\cF_k$, where $k$ is the lowest natural number such that $\tilde D_n\cap D=\emptyset$ for all $D\in\cF_k$. Note that by construction, we know that all oscillatory domains in each $\cF_k$ are disjoint.
\\ \\
We are done if we can show that there exists $K_b^l=K_b^l(\deps)\in\N$ such that $\cF_k=\emptyset$ for all $k\ge K_b^l+1$. For the moment, let us assume we had already shown that.
\\ \\
The oscillatory domains, which we have chosen in the above construction, cover $A_b$, namely,
\begin{equation}\label{eq:DntcoverAb}
\bigcup_{n\in \N}\tilde D_n\supset A_b\,,
\end{equation}
where we use the convention that $\tilde D_n=\emptyset$ if we have to stop before the $n^{\text{\tiny th}}$ step in the procedure above. In order to see this, note that \eqref{eq:DntcoverAb} is clear from the construction above if we have to stop after a finite number of steps. If the number of steps is infinite, we claim 
\begin{equation}\label{eq:de0}
\lim_{n\to\infty}\dt_n=0\,.
\end{equation}
By construction, we know that $\dt_n\le2\dt_m\ \textrm{for all } n\ge m$. Thus, in order to show \eqref{eq:de0}, it suffices to show that for every $\epsilon>0$ there exists $n\in\N$ with $\dt_n<\epsilon$. Suppose this was wrong, that is, there exists $\epsilon>0$ such that for all $n\in\N$, we have $\dt_n\ge\epsilon$. It follows that
\begin{equation}\label{eq:volDninfinite}
\sum_{n\in\N}|\tilde D_n|=\infty\,.
\end{equation}
On the other hand, we have $\tilde D_n\subset\ol\cap\obz$ for all $n\in\N$, $|\ol\cap\obz|\le|\om|\le1$ and by assumption, all $\cF_k$ with $k\ge K_b^l+1$ are empty. Since the oscillatory domains in each $\cF_k$ are disjoint, we get
\begin{equation}
\sum_{n\in\N}|\tilde D_n|\le\sum_{k=1}^{K_b^l}|\ol\cap\obz|\le K_b^l\,,
\end{equation}
which contradicts \eqref{eq:volDninfinite}. This proves \eqref{eq:de0}. Now suppose \eqref{eq:DntcoverAb} was wrong, that is, there exists $\tilde x\in A_b\setminus\bigcup_{n\in N}\tilde D_n$. By \eqref{eq:de0}, there exists $n\in\N$ with $\dt_n\le\frac{1}{4}\de_{\tilde x}$, which contradicts \eqref{eq:denxnco}, and thereby
% . Recall that $\dt_n$  was chosen such that 
% \begin{equation}
% \tilde\de_n\ge\ha\sup_{x\in A_b\setminus\bigcup_{k=1}^{n-1}\tilde D_k}\de_x\ge\ha\de_{\tilde x}\,,
% \end{equation}
% which contradicts $\dt_n\le\frac{1}{4}\de_{\tilde x}$. This 
proves \eqref{eq:DntcoverAb}.
\\ \\
It remains to show that there exists $K_b^l=K_b^l(\deps)\in\N$ such that $\cF_k=\emptyset$ for all $k\ge K_b^l+1$. Let $m\in\N$ with $\cF_m=\emptyset$. Let $D_m\in\cF_m$. By the above construction, for every $k\in\{1,\ldots,m-1\}$ there exists $
D_k\in\cF_k\ \textrm{with }  D_k\cap D_m\neq\emptyset \,
$ 
and such that $D_1,\ldots,D_{m-1}$ were all chosen before $D_m$ in the construction. By relabelling the families $\cF_k$ and the corresponding domains $D_k$, we may without loss of generality assume that $D_k$ was chosen before $D_n$ for all $k,n\in\{1,\ldots,m\}$ with $k\le n$. We denote the centres of $D_1,\ldots,D_m$ by $x_1,\ldots,x_m\in A_b$, the corresponding $\de$ by $\de_1,\ldots,\de_m$ and the corresponding $h$ by $h_1,\ldots,h_m$. Note that by construction, we have 
\begin{equation}
\de_n\le2\de_k\ \textrm{for all }  k,n\in\{1,\ldots,m\}\ \textrm{with } k\le n\,.
\end{equation}
Note that the $D_k$ do \emph{not} agree with the $\tilde D_k$ from the construction above but we have $D_k\in\bigcup_{n\in\N}\tilde D_n$ for all $k\in\{1,\ldots,m\}$. 
\\ \\
Our goal is to show that $m\le K_b^l$ for some $K_b^l=K_b^l(\deps)\in\N$, and we will show this for $b= 1$ and $b= 2$ separately.
\\ \\
\textbf{Proof of $m\le K_1^l<\infty$ for $A_1$. }
By Lemma \ref{le:propD}(i), we note that the $D_k$, $k\in\{1,\ldots,m\}$ are all cuboids contained in $\om$. Applying Lemma \ref{le:propD}(ii) with $x=x_m$, $y=x_1$ and $z=x_k$, 
% for $k\in\{1,\ldots,m-1\}$, 
we get for all $k\in\{1,\ldots,m-1\}$
%\begin{equation}
%\frac{1}{2}h_1\le h_m\le2h_1\ \textrm{and }  \frac{1}{2}h_1\le h_k\le2h_1 \, .
%\end{equation}
%We obtain
% \begin{equation}\label{eq:hm5hk}
% h_m\le2h_1\le4h_k\ \textrm{and }  h_k\le2h_1\le4h_m \, .
% \end{equation}
% % \begin{equation}\label{eq:hm5hk}
% % h_m\le2h_1=4\cdot\ha h_1\le4h_k\ \textrm{and }  h_k\le2h_1=4\cdot\ha h_1\le4h_m \, .
% % \end{equation}
% By $D_k\cap D_m\neq\emptyset$, it follows that for all $k\in\{1,\ldots,m-1\}$, we have
\begin{equation}\label{eq:xkxminfA1}
    \noi{x_k'-x_m'}<\ha c_0h_k^\gai+\ha c_0h_m^\gai\le\ha c_0\left(4h_m\right)^\gai+\ha c_0h_m^\gai\le \frac{\alpha}{2} c_0h_m^\gai
\end{equation}
% \begin{align}\label{eq:xkxminfA1}
% \begin{split}
% \noi{x_k'-x_m'}&<\ha c_0h_k^\gai+\ha c_0h_m^\gai\le\ha c_0\left(4h_m\right)^\gai+\ha c_0h_m^\gai=\left(4^\gai+1\right)\ha c_0h_m^\gai\\
% &\le \alpha\ha c_0h_m^\gai
% \end{split}
% \end{align}
with $\alpha:=\lceil4^\gai+1\rceil\in\N$. Now define the lattice $\cG$ by
\begin{equation*}
\cG := \left\{ (g',g_d)\in\R^{d-1}\times\R\ \middle\vert \begin{array}{l}
    g_d-\left(x_m\right)_d\in\left\lbrace -\frac{\de_m}{2},-\frac{\de_m}{4},0,\frac{\de_m}{4},\frac{\de_m}{2}\right\rbrace\ \textrm{and }   \\
    g'-x_m'=\frac{j}{4\alpha}c_0h_m^\gai\ \textrm{for some } j\in\Z^{d-1}\ \textrm{with } \noi{j}\le2\alpha^2
  \end{array}\right\}\,.
\end{equation*}
Here $\left(x_m\right)_d$ is the $d^{\text{\tiny th}}$ coordinate of $x_m$. Note that
\begin{equation}
|\cG|=5\cdot\left(2\cdot2\alpha^2+1\right)^{d-1}\,,
\end{equation}
which only depends on $\deps$. For every $k\in\{1,\ldots,m-1\}$, we associate a $g\in\cG$ to $x_k$, which is chosen such that 
\begin{equation}
\noi{g'-x_k'}+|g_d-\left(x_k\right)_d |
\end{equation}
is minimal among all $g\in\cG$. Note that by the choice of $\cG$, if $g\in\cG$ is associated to $x_k$, then 
by \eqref{eq:xkxminfA1}, $h_m\le4h_k$ and the definition of $\alpha$, we obtain 
\begin{equation}\label{eq:xkgihk}
\noi{x_k'-g'}\le\frac{1}{8\alpha}c_0\left(4 h_k\right)^\gai\le\frac{1}{8}c_0 h_k^\gai\,.
\end{equation}
Furthermore, we claim that
\begin{equation}\label{eq:xkdgdcovpf}
|g_d-\left(x_k\right)_d |<\ha\de_k\,.
\end{equation}
If $|\left(x_m\right)_d -\left(x_k\right)_d |\ge\ha\de_m$, then by the triangle inequality and $D_m\cap D_k\neq\emptyset$, we get \eqref{eq:xkdgdcovpf}. If $|\left(x_m\right)_d -\left(x_k\right)_d |<\ha\de_m$, we use the definition of $\cG$ and $2\de_k\ge\de_m$ to get \eqref{eq:xkdgdcovpf}. Combining \eqref{eq:xkgihk} and \eqref{eq:xkdgdcovpf}, we obtain that if $g\in\cG$ is associated to $x_k$, then $g\in D_k$.
\\ \\
Now, if $g\in\cG$ is associated to both $x_k$ and $x_n$ for $k,n\in\{1,\ldots,m-1\}$ with $k<n$, then by \eqref{eq:xkgihk} and $h_m\le4h_k$, 
\begin{align}\label{eq:xkxni}
\begin{split}
\noi{x_k'-x_n'}\le\noi{x_k'-g'}+\noi{g'-x_n'}\le\frac{1}{8\alpha}c_0h_m^\gai+\frac{1}{8\alpha}c_0h_m^\gai \le \frac{1}{4}c_0h_k^\gai
%&\le\frac{1}{4\alpha}c_0\left(4 h_k\right)^\gai\le\frac{1}{4}c_0h_k^\gai\,,
\end{split}
\end{align}
% \begin{align}\label{eq:xkxni}
% \begin{split}
% \noi{x_k'-x_n'}&\le\noi{x_k'-g'}+\noi{g'-x_n'}\le\frac{1}{8\alpha}c_0h_m^\gai+\frac{1}{8\alpha}c_0h_m^\gai=\frac{1}{4\alpha}c_0h_m^\gai\\
% &\le\frac{1}{4\alpha}c_0\left(4 h_k\right)^\gai=4^\gai\frac{1}{4\alpha}c_0h_k^\gai\le\alpha\frac{1}{4\alpha}c_0h_k^\gai=\frac{1}{4}c_0h_k^\gai\,,
% \end{split}
% \end{align}
where we used $\alpha:=\lceil4^\gai+1\rceil$ in the last step. Since $k<n$, we have $x_n\notin D_k$ by construction. Using \eqref{eq:xkxni} and $x_n\notin D_k$, we deduce that
\begin{equation}\label{eq:xkdxnddek}
|\left(x_n\right)_d-\left(x_k\right)_d |\ge\ha\de_k\,.
\end{equation}
\textbf{Claim. }
Fix $g\in\cG$. Then there are at most two indices $k,n\in\{1,\ldots,m-1\}$ with $k\neq n$ such that $\left(x_k\right)_d \ge g_d$ and $\left(x_n\right)_d \ge g_d$ and such that $g$ is associated to both $x_k$ and $x_n$.
\\ 
\\
\textbf{Proof of the claim. }
Suppose there were $k,n,j\in\{1,\ldots,m-1\}$ with $k<n<j$ such that $g$ is associated to  $x_k$, $x_n$ and $x_j$ and such that $\left(x_k\right)_d \ge g_d$, $\left(x_n\right)_d \ge g_d$ and $\left(x_j\right)_d \ge g_d$. Without loss of generality, let us assume that $g_d=0$. By $g\in D_i$ for all $i\in\{k,n,j\}$, we have
\begin{equation}\label{eq:knj1}
0\le\left(x_i\right)_d<\ha\de_i
\end{equation}
By \eqref{eq:xkdxnddek} and $k<n$, we have
\begin{equation}\label{eq:knj2}
\left(x_n\right)_d \ge\left(x_k\right)_d +\ha\de_k\,.
\end{equation}
Note that $0\le\left(x_n\right)_d \le\left(x_k\right)_d $ is not possible since this would imply $x_n\in D_k$ by \eqref{eq:xkxni}, $g\in D_k$ and the fact that $D_k$ is a cuboid. Similarly, we find that
\begin{equation}\label{eq:knj3}
\left(x_j\right)_d \ge\left(x_n\right)_d +\ha\de_n\,.
\end{equation}
Combining \eqref{eq:knj1}, \eqref{eq:knj2} and \eqref{eq:knj3}, we obtain
\begin{equation}
   \ha\de_j>\left(x_j\right)_d\ge\left(x_n\right)_d +\ha\de_n\ge \left(x_k\right)_d +\ha\de_k+\ha\de_n \ge\ha\de_j\,,
\end{equation}
% \begin{align*}
% \ha\de_j&>\left(x_j\right)_d\ge\left(x_n\right)_d +\ha\de_n\ge \left(x_k\right)_d +\ha\de_k+\ha\de_n\\
% &= \left(x_k\right)_d +\frac{1}{4}\left(2\de_k\right)+\frac{1}{4}\left(2\de_n\right)\ge0+\frac{1}{4}\de_j+\frac{1}{4}\de_j=\ha\de_j\,,
% \end{align*}
which is a contradiction. In the last step we used that $2\de_k\ge\de_j$ and $2\de_n\ge\de_j$ since $k,n<j$. This finishes the proof of the claim.
\\ \\
Similarly, we can show that there are at most two indices $k,n\in\{1,\ldots,m-1\}$ with $k\neq n$ such that $\left(x_k\right)_d \le g_d$ and $\left(x_n\right)_d \le g_d$ and such that $g$ is associated to both $x_k$ and $x_n$. Hence, for every $g\in\cG$ there exist at most four different indices $j_1,j_2,j_3,j_4\in\{1,\ldots,m-1\}$ such that $g$ is associated to $x_{j_i}$ for all $i=1,2,3,4$. This shows that $
m-1\le4|\cG|$, 
% or put differently,
%\begin{equation}
%m\le4|\cG|+1. 
%\end{equation}
where the right hand side only depends on $\deps$.  It follows that a possible choice is
\begin{equation}
K^l_1:=4|\cG|+1.
\end{equation}

\textbf{Proof of $m\le K_2^l<\infty$ for $A_2$. }
Define the lattice $\cG$ by
\begin{equation*}
\cG := \left\{ (g',g_d)\in\R^{d-1}\times\R\ \middle\vert \begin{array}{l}
    g_d-\left(x_m\right)_d\in\left\lbrace -\frac{\de_m}{2},-\frac{\de_m}{4},0,\frac{\de_m}{4},\frac{\de_m}{2}\right\rbrace\ \textrm{and }   \\
    g'-x_m'=\frac{j}{2\alpha}c_2\de_m^\gai\ \textrm{for some } j\in\Z^{d-1}\ \textrm{with } \noi{j}\le\alpha
  \end{array}\right\}\,,
\end{equation*}
where $\alpha:=\lceil2^\gai\rceil\in\N$.  Note that $\cG\subset \overline{D_m}$ and
\begin{equation}\label{eq:2sizeG}
|\cG|=5\cdot\left(2\alpha+1\right)^{d-1}\,,
\end{equation}
which only depends on $\deps$. For every $k\in\{1,\ldots,m-1\}$, we associate a $g\in\cG$ to $x_k$, which is chosen such that 
\begin{equation}
\noi{g'-x_k'}+|g_d-\left(x_k\right)_d |
\end{equation}
is minimal among all $g\in\cG$. We may also say that $x_k$ is associated to $g\in\cG$. Note that if $g\in\cG$ is associated to $x_k$, then
\begin{equation}\label{eq:2xkdgdcovpf}
|g_d-\left(x_k\right)_d |<\ha\de_k\,.
\end{equation}
If $x_k\notin D_m$, this follows from $D_m\cap D_k\neq\emptyset$ and $\cG\subset \overline{D_m}$. If $x_k\in D_m$, this follows from $\delta_m\le 2 \delta_k$
%\begin{equation}
%\frac{\de_m}{8}\le\frac{\de_k}{4}<\frac{\de_k}{2}
%\end{equation}
and the choice of $\cG$. Furthermore, we have
\begin{equation}\label{eq:2xkgicovpf}
\noi{x_k'-g'}\le\ha c_2\de_k^\gai\,.
\end{equation}
If $x_k\notin D_m$, this follows from $D_m\cap D_k\neq\emptyset$ and $\cG\subset \overline{D_m}$. If $x_k\in D_m$, this follows from 
\begin{equation}
\frac{1}{4\alpha}c_2\de_m^\gai\le\frac{1}{4}c_2\left(\frac{\de_m}{2}\right)^\gai\le\frac{1}{4}c_2\de_k^\gai\,.
\end{equation}
Combining \eqref{eq:2xkdgdcovpf} and \eqref{eq:2xkgicovpf}, we get $g\in D_k$. 
\\ \\ 
Fix $g\in\cG$ and consider all $k\in\{1,\ldots,m-1\}$ such that $g$ is associated to $x_k$. Without loss of generality, we may assume that we chose our coordinate system in such a way that $g=(0,\ldots,0)$. For each coordinate $i\in\{1,\ldots,d\}$, we distinguish between two cases: 
\begin{align*}
(0)\,\,\left(x_k\right)_i\ge0\quad \text{and}\quad (1) \,\,\left(x_k\right)_i<0,
\end{align*}
and associate a $\sigma=(\sigma_1,\ldots\sigma_d)\in\{0,1\}^d$ to $x_k$ such that condition $(\sigma_i)$ is satisfied for $\left(x_k\right)_i$ for all $i\in\{1,\ldots,d\}$. Fix $\sigma\in\{0,1\}^d$ and assume without loss of generality that $\sigma=(0,\ldots,0)$. Otherwise, rotate and reflect the coordinate system accordingly. In the following, we would like to count the number of $x_k$ such that $x_k$ is associated to $g$ and to $\sigma=(0,\ldots,0)$.
\\ \\
Let $k_1\in\{1,\ldots,m-1\}$ be the smallest number such that $x_{k_1}$ is associated to $g$ and to $\sigma=(0,\ldots,0)$. For $n\in\N$, if it exists, let $k_n\in\{1,\ldots,m-1\}$ be the  $n^{\text{\tiny th}}$ smallest number such that $x_{k_n}$ is associated to $g$ and to $\sigma=(0,\ldots,0)$. By construction, we have $\de_{k_n}\le2\de_{k_1}$ for all $n\ge1$, and moreover, we also have $g\in D_{k_n}$ as we have noticed before. Hence, with the notation $S_\sigma:=[0,\infty)^d$, we have by $\de_{k_n}\le2\de_{k_1}$
\begin{equation}\label{eq:dknssigmaint}
\left| D_{k_n}\cap S_\sigma\right|\le\left|\left[0,c_2\de_{k_n}^\gai\right]^{d-1}\times\left[0,\de_{k_n}\right]\right|\le 2 ^{\dgo}\left(c_2\de_{k_1}^\gai\right)^{d-1}\de_{k_1}
% \subset\left[0,c_2\left(2\de_{k_1}\right)^\gai\right]^{d-1}\times\left[0,2\de_{k_1}\right]
\end{equation}
% \begin{equation}\label{eq:dknssigmaint}
% D_{k_n}\cap S_\sigma\subset\left[0,c_2\de_{k_n}^\gai\right]^{d-1}\times\left[0,\de_{k_n}\right]\subset\left[0,c_2\left(2\de_{k_1}\right)^\gai\right]^{d-1}\times\left[0,2\de_{k_1}\right]
% \end{equation}
for all $n\ge1$. 
% Also note that 
% \begin{equation}
% \lvert D_{k_1}\cap S_\sigma\rvert\ge\frac{1}{2^d}\left(c_2\de_{k_1}^\gai\right)^{d-1}\de_{k_1}
% \end{equation}
% since all $y\in D_{k_1}$ with $y_i\ge\left(x_k\right)_i$ for all $i\in\{1,\ldots,d\}$ are contained in $D_{k_1}\cap S_\sigma$. 
Let $n\ge1$. We claim that $\de_{k_n}\ge\de_{k_1}$. To see this, recall that $x_{k_n}\in S_\sigma$ but $x_{k_n}\notin D_{k_1}$ since $k_n>k_1$. Therefore, there exists $i_0\in\{1,\ldots,d\}$ with
\begin{equation}
\left(x_{k_n}\right)_{i_0} \ge\left(x_{k_1}\right)_{i_0} +\ha c_2\de_{k_1}^\gai\ \textrm{if } i_0\in\{1,\ldots,d-1\}
\end{equation}
and 
\begin{equation}
\left(x_{k_n}\right)_d \ge\left(x_{k_1}\right)_d +\ha\de_{k_1}\ \textrm{if } i_0=d\,.
\end{equation}
At the same time, since $g\in D_{k_n}$, we have
\begin{equation}
\left(x_{k_n}\right)_{i} \le\ha c_2\de_{k_n}^\gai\ \textrm{for all } i\in\{1,\ldots,d-1\}
\end{equation}
and  
$\left(x_{k_n}\right)_{d} \le\ha \de_{k_n}$. 
% Thus, if $i_0\in\{1,\ldots,d\}$, we get 
% \begin{equation}
% \ha c_2\de_{k_n}^\gai\ge \left(x_{k_n}\right)_{i_0}\ge\left(x_{k_1}\right)_{i_0} +\ha c_2\de_{k_1}^\gai\ge \ha c_2\de_{k_1}^\gai,
% \end{equation}
% so $\de_{k_n}\ge\de_{k_1}$. If $i_0=d$, we get 
% \begin{equation}
% \ha \de_{k_n}\ge \left(x_{k_n}\right)_{d}\ge\left(x_{k_1}\right)_{d} +\ha \de_{k_1}\ge \ha \de_{k_1},
% \end{equation}
% so $\de_{k_n}\ge\de_{k_1}$. Hence, in any case, 
We deduce that $\de_{k_n}\ge\de_{k_1}$. 
For all $n\ge1$, we have by $\de_{k_n}\ge\de_{k_1}$ and $x_n\notin D_j$ for $j<n$,
\begin{equation}\label{eq:DknSsig}
\left\lvert\left( D_{k_n}\cap S_\sigma\right)\setminus\bigcup_{j=1}^{n-1}D_{k_j}\right\rvert\ge \frac{1}{2^d}\left(c_2\de_{k_1}^\gai\right)^{d-1}\de_{k_1}\,.
\end{equation}
% \\ \\
% \textbf{Claim. }
% For all $n\ge1$, we have
% \begin{equation}
% \left\lvert\left( D_{k_n}\cap S_\sigma\right)\setminus\bigcup_{j=1}^{n-1}D_{k_j}\right\rvert\ge \frac{1}{2^d}\left(c_2\de_{k_1}^\gai\right)^{d-1}\de_{k_1}\,.
% \end{equation}
% \textbf{Proof of the claim. }
% We have already noticed that this holds true for $n=1$. Now let $n>1$. By $\de_{k_n}\ge\de_{k_1}$, we are done if we can show that
% \begin{equation}
% y\notin\bigcup_{j=1}^{n-1}D_{k_j}
% \end{equation}
% for all $y\in D_{k_n}$ with $y_i\ge\left(x_{k_n}\right)_{i}$ for all $i\in\{1,\ldots,d\}$. So let $y\in D_{k_n}$ with $y_i\ge\left(x_{k_n}\right)_{i}$ for all $i\in\{1,\ldots,d\}$ and assume that $y\in D_{k_j}$ for some $j\in\{1,\ldots,n-1\}$. Since $D_{k_j}$ is a cuboid and both $y\in D_{k_j}$ and $g=(0,\ldots,0)\in D_{k_j}$, we have
% \begin{equation}
% \left[0,y_1\right]\times\ldots\left[0,y_d\right]\subset D_{k_j}\,.
% \end{equation}
% By $y_i\ge\left(x_{k_n}\right)_{i}\ge0$ for all $i\in\{1,\ldots,d\}$, we have $x_{k_n}\in D_{k_j}$, which cannot be true since $k_n>k_m$. This proves the claim.
% \\ \\
Now assume that $x_{k_1},\ldots,x_{k_1}$ are all associated to $g$ and to $\sigma=(0,\ldots,0)$. By \eqref{eq:DknSsig},
\begin{equation}\label{eq:volest}
    \left\lvert S_\sigma\cap\bigcup_{n=1}^{N}D_{k_n}\right\rvert=\sum_{n=1}^N\left\lvert\left( D_{k_n}\cap S_\sigma\right)\setminus\bigcup_{j=1}^{n-1}D_{k_j}\right\rvert\ge N\frac{1}{2^d}\left(c_2\de_{k_1}^\gai\right)^{d-1}\de_{k_1}\,.
\end{equation}
% \begin{align*}
% \left\lvert S_\sigma\cap\bigcup_{n=1}^{N}D_{k_n}\right\rvert&=\left\lvert \dot{\bigcup_{1\le n\le N} }\left(\left( D_{k_n}\cap S_\sigma\right)\setminus\bigcup_{j=1}^{n-1}D_{k_j}\right)\right\rvert=\sum_{n=1}^N\left\lvert\left( D_{k_n}\cap S_\sigma\right)\setminus\bigcup_{j=1}^{n-1}D_{k_j}\right\rvert\\
% &\ge\sum_{n=1}^N\frac{1}{2^d}\left(c_2\de_{k_1}^\gai\right)^{d-1}\de_{k_1}=N\frac{1}{2^d}\left(c_2\de_{k_1}^\gai\right)^{d-1}\de_{k_1}\,.
% \end{align*}
On the other hand, we have \eqref{eq:dknssigmaint}
% \begin{equation}
% D_{k_n}\cap S_\sigma\subset\left[0,c_2\de_{k_n}^\gai\right]^{d-1}\times\left[0,\de_{k_n}\right]\subset\left[0,c_2\left(2\de_{k_1}\right)^\gai\right]^{d-1}\times\left[0,2\de_{k_1}\right]
% \end{equation}
% for all $n\ge1$, and therefore,
% \begin{equation}
% \left\lvert S_\sigma\cap\bigcup_{n=1}^{N}D_{k_n}\right\rvert\le\left(c_2\left(2\de_{k_1}\right)^\gai\right)^{d-1}2\de_{k_1}=2 ^{\dgo}\left(c_2\de_{k_1}^\gai\right)^{d-1}\de_{k_1}\,.
% \end{equation}
% Hence, 
% \begin{equation}
% N\frac{1}{2^d}\left(c_2\de_{k_1}^\gai\right)^{d-1}\de_{k_1}\le\left\lvert S_\sigma\cap\bigcup_{n=1}^{N}D_{k_n}\right\rvert\le2 ^{\dgo}\left(c_2\de_{k_1}^\gai\right)^{d-1}\de_{k_1},
% \end{equation}
and thus,
\begin{equation}
N\le 2 ^{\dgo+d}\,.
\end{equation}
Repeating the same argument for every $g\in\cG$ and every $\sigma\in\{0,1\}^d$, we find by \eqref{eq:2sizeG} that
% \begin{equation}
% m-1\le|\cG|\cdot2^{d}\cdot2 ^{\dgo+d}=5\cdot\left(2\alpha+1\right)^{d-1}\cdot2 ^{\dgo+2d}
% \end{equation}
% so 
\begin{equation}
K_2^l:=5\cdot\left(2\alpha+1\right)^{d-1}\cdot2 ^{\dgo+2d}+1\in\N
\end{equation}
is a possible choice. Note that the right-hand side only depends on $\deps$.
\\ \\
\textbf{Covering theorem for $A_3$. }
This is simply the Besicovitch covering theorem for cubes and here $K_3^l=K_3^l(d)\in\N$.
\end{proof}

\begin{definition}[Choice of a subfamily of oscillatory domains]\label{de:Dj}
Let $l\in\{1,\ldots,L\}$ and for every $x\in\olt\cap\obd$ let $\de_x\in(0,\dz] $, and hence also $D_x:=D_x(\de_x)$, be chosen as in Lemma \ref{le:D123}. Now apply Lemma \ref{le:cover} to this collection of oscillatory domains to get subfamilies $\cF_1,\ldots,\cF_{K_l}$ of oscillatory domains, where $K_l=K_l(d,\ga)\in\N$. Write 
\begin{equation}
\bigcup_{k=1}^{K_l}\cF_k=:\left\lbrace D_j\right\rbrace_{j\in J}\, .
\end{equation}
where $J=J_1\cup J_2\cup J_3$ is an index set such that for every $m\in\{1,2,3\}$ and all $j\in J_m$ the oscillatory domain $D_j$ satisfies condition $(m)$ in Lemma \ref{le:D123}. In the following, for every $m\in\{1,2,3\}$ and every $j\in J_m$ the point $x_j\in\olt\cap\obd$ is such that $D_j=D_{x_j}$.
\end{definition}

\begin{definition}\label{de:somze}
Let $\ga\in\left[\frac{d-1}{d}, 1\right)$. Define $s,s',\omega,\zeta$ by
\begin{equation}
\frac{1}{s'}:=\frac{\di\left(\dgo\right)^2-d}{\di\left(\dgo\right)^2+1},\quad \quad\frac{1}{s}:=\frac{d+ 1}{\di\left(\dgo\right)^2+1},
\end{equation}
and
\begin{equation}
\omega:=\left(\dgo\right)\frac{\di\left(\dgo\right)^2-d}{\di\left(\dgo\right)^2+1},\quad \zeta:=\frac{1}{s}\left(-\be+\dg\right).
\end{equation}
\end{definition}

We are now ready to prove that the number of oscillatory domains we choose is bounded by a constant times $\dz^{-d}$.
\begin{lemma}\label{le:Jest}
Let $l\in\{1,\ldots,L\}$ and let $\left\lbrace D_j\right\rbrace_{j\in J}$ be chosen as in Definition \ref{de:Dj}. Then
\begin{equation}
|J|\ls\dz^{-d}\, .
\end{equation}
\end{lemma}

\begin{proof}
Recall that $J=J_1\cup J_2\cup J_3$, where each $J_m$, $m\in\{1,2,3\}$ is chosen such that all $D_j$ with $j\in J_m$ satisfy condition $(m)$ in Lemma \ref{le:D123}. We will show for all $m\in\{1,2,3\}$ that $|J_m|\ls\dz^{-d}$.
\\
\\
{\bf Notation. } If $j\in J$ and $D_j=D_{x_j}(\de_{x_j})$ for some $x_j\in\olt\cap\obd$, write
\begin{equation}
\de_j:=\de_{x_j}\, ,\quad h_j:=h_{x_j}\, ,\quad a_j:=a_{x_j}=\min\left(\de_j,c_0\max\left(h_j,c_1\de_j\right)^{\frac{1}{\ga}}\right)\,,\quad K:=K_l\, .
\end{equation}
% \begin{align*}
% \de_j&:=\de_{x_j}\\
% h_j&:=h_{x_j}\\
% a_j&:=a_{x_j}=\min\left(\de_j,c_0\max\left(h_j,c_1\de_j\right)^{\frac{1}{\ga}}\right)\, .
% \end{align*}
{\bf Estimate for $J_1$. } 
Recall that if $j\in J_1$, then $\de_j=\dz$. We have
\begin{align*}
|J_1|&\le \vert\left\lbrace j\in J_1\mid a_j=\dz\right\rbrace\vert+\left\vert\left\lbrace j\in J_1\mid a_j=c_0h_j^\gai\right\rbrace\right\vert+\left\vert\left\lbrace j\in J_1\mid a_j=c_2\dz^\gai\right\rbrace\right\vert\\
&=:|J_{1,1}|+|J_{1,2}|+|J_{1,3}|\, .
\end{align*}
{\bf Estimate for $J_{1,1}$. } 
Note that if $a_j=\dz$, then $|D_j|=\dz^d$ since by Lemma \ref{le:Dwelldef}(i) and Lemma \ref{le:propD}(i), $D_j$ is a cuboid contained in $\ol$. Thus,
\begin{equation}
|J_{1,1}|=\vert\left\lbrace j\in J_1\mid a_j=\dz\right\rbrace\vert=\dz^{-d}\sum_{j\in J_{1,1}}|D_j|\ls \dz^{-d}\left\vert\ol\cap\obz\right\vert\ls\dz^{-d}\, ,
\end{equation}
where we used the covering lemma (Lemma \ref{le:cover}) and $D_j\subset\ol\cap\obz$ for all $j\in J$ in the second last step . In the last step, we used that $\left\vert\ol\cap\obz\right\vert\le |\om |\le1$.
% \begin{align*}
% |J_{1,1}|&=\vert\left\lbrace j\in J_1\mid a_j=\dz\right\rbrace\vert=\sum_{j\in J_{1,1}}\dz^d\dz^{-d}=\dz^{-d}\sum_{j\in J_{1,1}}|D_j|\ls K_l\dz^{-d}\left\vert\ol\cap\obz\right\vert\\
% &\ls\dz^{-d}\, ,
% \end{align*}
% where we used the covering lemma (Lemma \ref{le:cover}) in the second last step and that $D_j\subset\ol\cap\obz$ for all $j\in J$. In the last step, we used that $\left\vert\ol\cap\obz\right\vert\le |\om |\le1$.
\\
\\
{\bf Estimate for $J_{1,2}$. } 
% Recall that if $a_j=c_0h_j^\gai$, then $h_j\ge c_1\dz$, so
% \begin{equation}\label{eq:djsim}
% |D_j|=\dz a_j^{d-1}=\dz \left(c_0h_j^\gai\right)^{d-1}\sim \dz h_j^\dg\, ,
% \end{equation}
% where we used Lemma \ref{le:propD}(i) in the first step. Furthermore, recall that by Lemma \ref{le:propD}(iii), we have 
% \begin{equation}\label{eq:hjhw}
% \frac{1}{2}h_j\le h_w\le2h_j\ \textrm{for all } w\in D_j \, .
% \end{equation}
Recall that if $a_j=c_0h_j^\gai$, then $h_j\ge c_1\dz$ and $|D_j|\sim \dz h_j^\dg$
by Lemma \ref{le:propD}(i). Using these properties, Lemma \ref{le:propD}(iii) and Lemma \ref{le:cover}, we get 
% \begin{equation}\label{eq:hjhw}
% \frac{1}{2}h_j\le h_w\le2h_j\ \textrm{for all } w\in D_j \, .
% \end{equation}
% We get
\begin{align*}
|J_{1,2}|&=\left\vert\left\lbrace j\in J_1\mid a_j=c_0h_j^\gai\right\rbrace\right\vert\sim\dz^{-1}\sum_{j\in J_{1,2}}|D_j| h_j^{-\dg}\sim\dz^{-1}\sum_{j\in J_{1,2}}\inx{D_j}{w} h_w^{-\dg}\\
&\ls\dz^{-1}\inx{\left\lbrace w\in \ol\cap\obz\,\middle\vert\, h_w\ge\ha c_1\dz\right\rbrace}{w} h_w^{-\dg}\, .
\end{align*}
% \begin{align*}
% |J_{1,2}|&=\left\vert\left\lbrace j\in J_1\mid a_j=c_0h_j^\gai\right\rbrace\right\vert\sim\dz^{-1}\sum_{j\in J_{1,2}}|D_j| h_j^{-\dg}\\
% &\sim\dz^{-1}\sum_{j\in J_{1,2}}\inx{D_j}{w} h_w^{-\dg}\ls\dz^{-1}\inx{\left\lbrace w\in \ol\cap\obz\,\middle\vert\, h_w\ge\ha c_1\dz\right\rbrace}{w} h_w^{-\dg}\, ,
% \end{align*}
% where we used \eqref{eq:djsim} in the second step and \eqref{eq:hjhw} in the third step. In the fourth step we used $h_j\ge c_1\dz$, \eqref{eq:hjhw} and Lemma \ref{le:cover}.
% \begin{align*}
% |J_{1,2}|&=\left\vert\left\lbrace j\in J_1\mid a_j=c_0h_j^\gai\right\rbrace\right\vert=\sum_{j\in J_{1,2}}\dz h_j^\dg\dz^{- 1} h_j^{-\dg}\sim\dz^{-1}\sum_{j\in J_{1,2}}|D_j| h_j^{-\dg}\\
% &\sim\dz^{-1}\sum_{j\in J_{1,2}}\inx{D_j}{w} h_w^{-\dg}\ls\dz^{-1}\inx{\left\lbrace w\in \ol\cap\obz\,\middle\vert\, h_w\ge\ha c_1\dz\right\rbrace}{w} h_w^{-\dg}\, ,
% \end{align*}
% where we used $|D_j|\sim \dz h_j^\dg$ in the third step, and $\frac{1}{2}h_j\le h_w\le2h_j$ in the fourth step and in the fifth step. In the fifth step we also used the covering lemma (Lemma \ref{le:cover}).
Note that $\dg>1$ since $d\ge2$ and $\ga<1$. We estimate using $\dg\le d$ and $\dz\le1$
% \begin{equation}
%    |J_{1,2}|&\ls\dz^{-1}\inx{\left\lbrace w\in \ol\cap\obz\,\middle\vert\, h_w\ge\ha c_1\dz\right\rbrace}{w} h_w^{-\dg}\ls\dz^{-1}\int_{\ha c_1\dz}^{\infty} \dr h h^{-\dg}\sim\dz^{-\dg}\ls\dz^{-d} , 
% \end{equation}
\begin{align*}
|J_{1,2}|&\ls\dz^{-1}\inx{\left\lbrace w\in \ol\cap\obz\,\middle\vert\, h_w\ge\ha c_1\dz\right\rbrace}{w} h_w^{-\dg}\ls\dz^{-1}\int_{\ha c_1\dz}^{\infty} \dr h h^{-\dg}\sim\dz^{-\dg}\le\dz^{-d}\, .
\end{align*}
% where we used $\dg\le d$ and $\dz\le1$ in the last step. 
% \begin{align*}
% |J_{1,2}|&\ls\dz^{-1}\inx{\left\lbrace w\in \ol\cap\obz\,\middle\vert\, h_w\ge\ha c_1\dz\right\rbrace}{w} h_w^{-\dg}\ls\dz^{-1}\int_{\ha c_1\dz}^{\infty} \dr h h^{-\dg}\\
% &\sim\dz^{-\dg}\int_{\ha c_1\dz}^{\infty} \dr h\frac{1}{c_1\dz} \left(\frac{h}{c_1\dz}\right)^{-\dg}=\dz^{-\dg}\int_{\ha }^{\infty} \dr tt^{-\dg}\\
% &\sim\dz^{-\dg}\ls\dz^{-d}\, ,
% \end{align*}
% where we used the change of variables $t=\frac{h}{c_1\dz}$ in the fourth step and $\dg\le d$ in the last step. 
{\bf Estimate for $J_{1,3}$. } 
By $a_j=c_2\dz^\gai$, we have $h_j\le c_1\dz$ for all $j\in J_{1,3}$. Furthermore, 
\begin{equation}\label{eq:djj13}
|D_j|\sim\dz a_j^{d-1}\sim\dz^\dgo\, .
\end{equation}
% \begin{equation}
% |D_j|\sim\dz a_j^{d-1}=\dz \dz^\dg=\dz^\dgo\, .
% \end{equation}
By Lemma \ref{le:propD}(iv), we have $|h_w-h_j|\le\de_0$ for all $j\in J_{1,3}$ and all $w\in D_j$. Since $h_j\le c_1\dz$, we get
\begin{equation}\label{eq:hwdzestj13}
h_w\le h_j+|h_w-h_j|\le c_1\dz+\dz=(c_1+1)\dz\ \textrm{for all } w\in D_j \, .
\end{equation}
We obtain
\begin{align*}
|J_{1,3}|&=\left\vert\left\lbrace j\in J_1\mid a_j=c_2\dz^\gai\right\rbrace\right\vert=\dz^{-\left(\dgo\right)}\sum_{j\in J_{1,3}}|D_j|\\
&\ls\dz^{-\left(\dgo\right)}\left\vert\left\lbrace w\in \ol\cap\obz\,\middle\vert\, h_w\le(c_1+1)\dz\right\rbrace\right\vert\sim\dz^{-\dg}\le\dz^{-d}\, ,
\end{align*}
where were used \eqref{eq:djj13} in the second step, the covering lemma (Lemma \ref{le:cover}) and \eqref{eq:hwdzestj13} in the third step, and $\dg\le d$ and $\dz\le1$ in the last step.
% \begin{align*}
% |J_{1,3}|&=\left\vert\left\lbrace j\in J_1\mid a_j=c_2\dz^\gai\right\rbrace\right\vert=\sum_{j\in J_{1,3}}\dz^\dgo\dz^{-\left(\dgo\right)}=\dz^{-\left(\dgo\right)}\sum_{j\in J_{1,3}}|D_j|\\
% &\ls\dz^{-\left(\dgo\right)}\left\vert\left\lbrace w\in \ol\cap\obz\,\middle\vert\, h_w\le(c_1+1)\dz\right\rbrace\right\vert\sim\dz^{-\dg}\le\dz^{-d}\, ,
% \end{align*}
% where were used $|D_j|\sim\dz^\dgo$ in the third step, the covering lemma (Lemma \ref{le:cover}) and $h_w\le h_j+|h_w-h_j|\le c_1\dz+\dz=(c_1+1)\dz$ in the fourth step, and $\dg\le d$ in the last step.
\\
\\
{\bf Estimate for $J_2$. } 
Recall that if $j\in J_2$, then $\nor{V}{\frac{d}{2}, D_j}{\frac{d}{2}}\gs1$ if $d\ge3$ and $\no{V}{\B, D_j}\gs1$ if $d=2$. Using the covering lemma (Lemma \ref{le:cover}), we obtain for $d\ge3$
\begin{equation}
|J_2|=\sum_{j\in J_2}1\ls\sum_{j\in J_2}\nor{V}{\frac{d}{2}, D_j}{\frac{d}{2}}\ls\nor{V}{\frac{d}{2}, \ol\cap\obz}{\frac{d}{2}}\ls\nor{V}{\frac{d}{2}, \om}{\frac{d}{2}}\ls\dz^{-d}.
\end{equation}
For $d=2$, using \cite[Lemma A.1]{frank2019bound} we get
\begin{equation}
|J_2|=\sum_{j\in J_2}1\ls\sum_{j\in J_2}\no{V}{\B, D_j}\ls\no{V}{\B, \ol\cap\obz}\ls\no{V}{\B, \om}\ls\dz^{-d}, 
\end{equation}

{\bf Estimate for $J_3$. } 
Recall that if $j\in J_3$, then $a_j=c_0\max\left( h_j,c_1\de_j\right)^{\frac{1}{\ga}}$ and
\begin{equation}\label{eq:J3est}
\nor{V}{\pt, D_j}{\pt}\gs\max\left(\frac{h_j}{c_1\de_j} ,1\right)^{\frac{d-1}{\ga}}\, .
\end{equation}

% In the following, we will use $s'$, $s$, $\ome$ and $\ze$ as defined in Definition \ref{de:somze}. By $\frac{1}{s}+\frac{1}{s'}=1$, see Lemma \ref{le:somze}(i), and Hölder's inequality, we have
In the following, we will use $s'$, $s$, $\ome$ and $\ze$ as defined in Definition \ref{de:somze}. By 
\begin{equation}\label{eq:l1}
    \frac{1}{s}+\frac{1}{s'}=1
\end{equation}
and Hölder's inequality, we have
\begin{align} \label{eq:J3-pre-final}
|J_3|&=\sum_{j\in J_3}\de_j^\ome\max\left(\frac{h_j}{c_1\de_j} ,1\right)^{-\ze}\de_j^{-\ome}\max\left(\frac{h_j}{c_1\de_j} ,1\right)^{\ze} \nonumber \\
%&\le\left(\sum_{j\in J_3}\de_j^{\ome s'}\max\left(\frac{h_j}{c_1\de_j} ,1\right)^{-\ze s'}\right)^{\frac{1}{s'}}\left(\sum_{j\in J_3}\de_j^{-\ome s}\max\left(\frac{h_j}{c_1\de_j} ,1\right)^{\ze s}\right)^{\frac{1}{s}}\\
&\le\left(\sum_{j\in J_3}\de_j^{\dgo}\max\left(\frac{h_j}{c_1\de_j} ,1\right)^{-\ze s'}\right)^{\frac{1}{s'}}\left(\sum_{j\in J_3}\de_j^{-\be}\max\left(\frac{h_j}{c_1\de_j} ,1\right)^{-\be+\dg}\right)^{\frac{1}{s}}\, ,
\end{align}
where we used 
\begin{equation}\label{eq:l2}
\ome s'=1+\dg \quad \textrm{ and } \quad \ome s=\be
\end{equation}
and the definition of $\ze$, see Definition \ref{de:somze}, in the last step. By \eqref{eq:J3est} and Lemma \ref{le:propD}(v), we have for every $j\in J_3$
\begin{equation}
    \de_j^{-\be}\max\left(\frac{h_j}{c_1\de_j} ,1\right)^{-\be+\dg}\ls\max\left( h_j,c_1\de_j\right)^{-\be}\nor{V}{\pt,D_j}{\pt}\ls\nos{V}{\pt,\be,D_j}{\pt}\, .
\end{equation}
% \begin{align*}
% \de_j^{-\be}\max\left(\frac{h_j}{c_1\de_j} ,1\right)^{-\be+\dg}&\sim\max\left( h_j,c_1\de_j\right)^{-\be}\max\left(\frac{h_j}{c_1\de_j} ,1\right)^{\dg}\\
% &\ls\max\left( h_j,c_1\de_j\right)^{-\be}\nor{V}{\pt,D_j}{\pt}\ls\nos{V}{\pt,\be,D_j}{\pt}\, .
% \end{align*}
% where we used Lemma \ref{le:somze}(ii) and the definition of $\ze$, see Definition \ref{de:somze}, in the last step. By Lemma \ref{le:propD}(v) and \eqref{eq:J3est}, we have for every $j\in J_3$
% \begin{align*}
% \de_j^{-\be}\max\left(\frac{h_j}{c_1\de_j} ,1\right)^{-\be+\dg}&\sim\max\left( h_j,c_1\de_j\right)^{-\be}\max\left(\frac{h_j}{c_1\de_j} ,1\right)^{\dg}\\
% &\ls\max\left( h_j,c_1\de_j\right)^{-\be}\nor{V}{\pt,D_j}{\pt}\ls\nos{V}{\pt,\be,D_j}{\pt}\, .
% \end{align*}
Using the covering lemma (Lemma \ref{le:cover}), we get by \eqref{eq:dzdef}
\begin{equation}\label{eq:secestdz}
\sum_{j\in J_3}\de_j^{-\be}\max\left(\frac{h_j}{c_1\de_j} ,1\right)^{-\be+\dg}\ls\sum_{j\in J_3}\nos{V}{\pt,\be,D_j}{\pt}\ls \no{V}{\pt,\be}^{-\pt}\ls\nos{V}{\pt,\be,\om}{\pt}\ls\dz^{-2\pt}\, .
\end{equation}
If we can show that
\begin{equation}\label{eq:J32s}
\sum_{j\in J_3}\de_j^{\dgo}\max\left(\frac{h_j}{c_1\de_j} ,1\right)^{-\ze s'}\ls\dz\, ,
\end{equation}
then we get by \eqref{eq:J3-pre-final} and \eqref{eq:secestdz} that 
% \begin{align*}
% |J_3|&\le\left(\sum_{j\in J_3}\de_j^{\dgo}\max\left(\frac{h_j}{c_1\de_j} ,1\right)^{-\ze s'}\right)^{\frac{1}{s'}}\left(\sum_{j\in J_3}\de_j^{-\be}\max\left(\frac{h_j}{c_1\de_j} ,1\right)^{-\be+\dg}\right)^{\frac{1}{s}}\\
% &\ls\dz^{\frac{1}{s'}}\dz^{-2\pt\frac{1}{s}}=\dz^{\frac{1}{s'}-2\pt\frac{1}{s}}=\dz^{-d}\, ,
% \end{align*}
% where we used Lemma \ref{le:somze}(iv) in the last step. Hence, it remains to show \eqref{eq:J32s}. First note that by $a_j=c_0\max\left( h_j,c_1\de_j\right)^{\frac{1}{\ga}}$ and by Lemma \ref{le:Dwelldef}(ii), we have
\begin{align*}
|J_3| %&\le\left(\sum_{j\in J_3}\de_j^{\dgo}\max\left(\frac{h_j}{c_1\de_j} ,1\right)^{-\ze s'}\right)^{\frac{1}{s'}}\left(\sum_{j\in J_3}\de_j^{-\be}\max\left(\frac{h_j}{c_1\de_j} ,1\right)^{-\be+\dg}\right)^{\frac{1}{s}}\\
\ls\dz^{\frac{1}{s'}}\dz^{-2\pt\frac{1}{s}}=\dz^{-d}\, ,
\end{align*}
where we used
\begin{equation}\label{eq:l4}
-2\pt\frac{1}{s}+\frac{1}{s'}=-d
\end{equation}
in the last step. Hence, it remains to show \eqref{eq:J32s}. First note that by $a_j=c_0\max\left( h_j,c_1\de_j\right)^{\frac{1}{\ga}}$ and by Lemma \ref{le:Dwelldef}(ii), we have
\begin{equation}
|D_j|\sim\dz a_j^{d-1}=\dz \left( c_0\max\left( h_j,c_1\de_j\right)^{\frac{1}{\ga}}\right)^{d- 1}\sim\dz^\dgo\max\left(\frac{h_j}{c_1\de_j} ,1\right)^\dg\, 
\end{equation}
for all $j\in J_3$. Thus,
\begin{align*}
&\qquad\sum_{j\in J_3}\de_j^{\dgo}\max\left(\frac{h_j}{c_1\de_j} ,1\right)^{-\ze s'}\sim\sum_{j\in J_3}|D_j|\max\left(\frac{h_j}{c_1\de_j} ,1\right)^{-\left(\ze s'+\dg\right)}\\
&=\sum_{j\in J_3,\,h_j\ge c_1\dz}|D_j|\max\left(\frac{h_j}{c_1\de_j} ,1\right)^{-\left(\ze s'+\dg\right)}+\sum_{j\in J_3,\,h_j< c_1\dz}|D_j|\max\left(\frac{h_j}{c_1\de_j} ,1\right)^{-\left(\ze s'+\dg\right)}\\
&=:S_1+S_2\, .
\end{align*}
% \begin{align*}
% &\qquad\sum_{j\in J_3}\de_j^{\dgo}\max\left(\frac{h_j}{c_1\de_j} ,1\right)^{-\ze s'}\\
% &=\sum_{j\in J_3}\de_j^{\dgo}\max\left(\frac{h_j}{c_1\de_j} ,1\right)^\dg\max\left(\frac{h_j}{c_1\de_j} ,1\right)^{-\left(\ze s'+\dg\right)}\\
% &\sim\sum_{j\in J_3}|D_j|\max\left(\frac{h_j}{c_1\de_j} ,1\right)^{-\left(\ze s'+\dg\right)}\\
% &=\sum_{j\in J_3,h_j\ge c_1\dz}|D_j|\max\left(\frac{h_j}{c_1\de_j} ,1\right)^{-\left(\ze s'+\dg\right)}+\sum_{j\in J_3,h_j< c_1\dz}|D_j|\max\left(\frac{h_j}{c_1\de_j} ,1\right)^{-\left(\ze s'+\dg\right)}\\
% &=:S_1+S_2\, .
% \end{align*}
{\bf Estimate for $S_1$. } 
If $j\in J_3$ with $h_j\ge c_1\dz$, then we have $h_j\ge c_1\dz\ge c_1\de_j$ since $\dz\ge \de_j$. Thus,
\begin{equation}
a_j=c_0\max\left( h_j,c_1\de_j\right)^{\frac{1}{\ga}}=c_0h_j^{\frac{1}{\ga}}.
\end{equation}
By Lemma \ref{le:propD}(iii), it follows that
% , we obtain
% \begin{equation}\label{eq:hjhwcomp}
% \frac{1}{2}h_j\le h_w\le2h_j\ \textrm{for all } w\in D_j \, .
% \end{equation}
% It follows that
\begin{align*}
S_1&\le\sum_{j\in J_3,\,h_j\ge c_1\dz}|D_j|\left(\frac{h_j}{c_1\de_0} \right)^{-\left(\ze s'+\dg\right)} \sim \sum_{j\in J_3,\,h_j\ge c_1\dz}\inx{D_j}{w} \left(\frac{h_w}{c_1\de_0} \right)^{-\left(\ze s'+\dg\right)}\\
&\ls\inx{\left\lbrace w\in \ol\cap\obz\,\middle\vert\, h_w\ge\ha c_1\dz\right\rbrace}{w} \left(\frac{h_w}{c_1\de_0} \right)^{-\left(\ze s'+\dg\right)}\ls\int_{\ha c_1\dz}^{\infty} \dr h \left(\frac{h}{c_1\de_0} \right)^{-\left(\ze s'+\dg\right)}\\
&=c_1\dz\int_{\ha }^{\infty} \dr t t^{-\left(\ze s'+\dg\right)}\ls\dz\, .
\end{align*}
% \begin{align*}
% S_1&=\sum_{j\in J_3,h_j\ge c_1\dz}|D_j|\max\left(\frac{h_j}{c_1\de_j} ,1\right)^{-\left(\ze s'+\dg\right)}=\sum_{j\in J_3,h_j\ge c_1\dz}|D_j|\left(\frac{h_j}{c_1\de_j} \right)^{-\left(\ze s'+\dg\right)}\\
% &\le\sum_{j\in J_3,h_j\ge c_1\dz}|D_j|\left(\frac{h_j}{c_1\de_0} \right)^{-\left(\ze s'+\dg\right)}\sim \sum_{j\in J_3,h_j\ge c_1\dz}\inx{|D_j|}{w} \left(\frac{h_w}{c_1\de_0} \right)^{-\left(\ze s'+\dg\right)}\\
% &\ls\inx{\left\lbrace w\in \ol\cap\obz\,\middle\vert\, h_w\ge\ha c_1\dz\right\rbrace}{w} \left(\frac{h_w}{c_1\de_0} \right)^{-\left(\ze s'+\dg\right)}\ls\int_{\ha c_1\dz}^{\infty} \dr h \left(\frac{h_w}{c_1\de_0} \right)^{-\left(\ze s'+\dg\right)}\\
% &=c_1\dz\int_{\ha c_1\dz}^{\infty} \dr h\frac{1}{c_1\dz} \left(\frac{h_w}{c_1\de_0} \right)^{-\left(\ze s'+\dg\right)}=c_1\dz\int_{\ha }^{\infty} \dr t t^{-\left(\ze s'+\dg\right)}\ls\dz\, .
% \end{align*}
% Here we used $h_j\ge c_1\dz$ in the second step and we used $\dz\ge \de_j$ and $\ze s'+\dg>1\ge0$ (Lemma \ref{le:somze}(iii)) in the third step. In the fourth step we used that $\frac{1}{2}h_j\le h_w\le2h_j$ for all $w\in D_j$. We also used this fact in the fifth step to get $h_w\ge\frac{1}{2}h_j\ge\ha c_1\dz$ and we applied the covering lemma (Lemma \ref{le:cover}) in the fifth step, too. In the second last step we used the change of variables $t=\frac{h_w}{c_1\de_0} $ and in the last step we used $\ze s'+\dg>1$ (see Lemma \ref{le:somze}(iii)) to deduce that the integral is finite.
Here we used $h_j\ge c_1\dz$, $\dz\ge \de_j$ and 
\begin{equation}\label{eq:l3}
    \ze s'+\dg>1
\end{equation}
in the second step. We explain \eqref{eq:l3} below. In the fourth step, we used Lemma \ref{le:propD}(iii) to get $h_w\ge\frac{1}{2}h_j\ge\ha c_1\dz$ and moreover we applied the covering lemma (Lemma \ref{le:cover}). In the second last step we used the change of variables $t=\frac{h_w}{c_1\de_0} $ and in the last step we used \eqref{eq:l3} to deduce that the integral is finite.
% in the second step. We explain \eqref{eq:l3} below. We used \eqref{eq:hjhwcomp} in the third and fourth step. In the fourth step we used this fact to get $h_w\ge\frac{1}{2}h_j\ge\ha c_1\dz$ and moreover we applied the covering lemma (Lemma \ref{le:cover}). In the second last step we used the change of variables $t=\frac{h_w}{c_1\de_0} $ and in the last step we used \eqref{eq:l3} to deduce that the integral is finite.
\\
\\
In order to show \eqref{eq:l3}, note that we have
\begin{equation}
    \ze s'+\dg =\frac{1}{s}\left(-\be+\dg\right)s'+\dg=1-2+\dg\frac{d+1}{\di\left(\dgo\right)^2-d}\, .
\end{equation}
% \begin{align*}
% &\qquad\ze s'+\dg =\frac{1}{s}\left(-\be+\dg\right)s'+\dg\\
% & =\frac{d+ 1}{\di\left(\dgo\right)^2+1}\frac{\di\left(\dgo\right)^2+1}{\di\left(\dgo\right)^2-d}\times\\
% &\qquad\left(\frac{-1}{d+1}\left(\dgo\right)\left[\frac{1}{d}\left(\dgo\right)^2-d\right]+\dg\right)+\dg\\
% &=-\left(\dgo\right)+\dg\frac{d+1}{\di\left(\dgo\right)^2-d}+\dg\\
% &=-1+\dg\frac{d+1}{\di\left(\dgo\right)^2-d}\\
% &=1-2+\dg\frac{d+1}{\di\left(\dgo\right)^2-d}\, .
% \end{align*}
Therefore,  \eqref{eq:l3} is equivalent to 
%Thus $\ze s'+\dg>1$, if and only if
%\begin{equation}
%-2+\dg\frac{d+1}{\di\left(\dgo\right)^2-d}>0\, .
%\end{equation}
%This is equivalent to
\begin{equation}\label{eq:Ypre}
\dg>\frac{2}{d+1}\left[\di\left(\dgo\right)^2-d\right]=\frac{2d}{d+1}\left[\frac{1}{d^2}\left(\dgo\right)^2-1\right]\, .
\end{equation}
Define 
\begin{equation}
Y: =\di\left(\dgo\right)
\end{equation}
and note that since $\ga\in\left[\frac{d-1}{d}, 1\right)$, we have $\dg\in\left(d-1,d\right]$ and therefore, $Y\in\left(1,\frac{d+1}{d}\right]$. The inequality \eqref{eq:Ypre} now reads 
\begin{equation}\label{eq:Y}
dY-1>\frac{2d}{d+1}\left[Y^2-1\right]\, .
\end{equation}
% Using $dY-1\ge2Y-1$ since $d\ge2$ and $\frac{d}{d+1}\le\frac{1}{Y}$, we have
% Since  $d\ge2$, we have
% \begin{equation}
% dY-1\ge2Y-1
% \end{equation}
% and since $\frac{d}{d+1}\le\frac{1}{Y}$, we have
Using $\frac{d}{d+1}\le\frac{1}{Y}$, $\frac{2d}{d+1}> 1$ and $d\ge2$, we have
\begin{equation}
\frac{2d}{d+1}\left[Y^2-1\right]\le\frac{2}{Y}Y^2-\frac{2d}{d+1}<2Y-1\le dY-1\,,
\end{equation}
which shows \eqref{eq:Y} and hence, \eqref{eq:l3} holds. 
% \begin{equation}
% \frac{2d}{d+1}\left[Y^2-1\right]\le\frac{2}{Y}\left[Y^2-1\right]=2Y-\frac{2}{Y}<2Y-1\le dY-1\,.
% \end{equation}
% Thus, we only need to show that
% \begin{equation}
% 2Y-1>\frac{2}{Y}\left[Y^2-1\right]\,.
% \end{equation}
% This is equivalent to
% \begin{equation}
% 2Y^2-Y=Y\left[2Y-1\right]>2\left[Y^2-1\right]=2Y^2-2\,.,\end{equation}
% which is equivalent to
% \begin{equation}
% Y<2\,.
% \end{equation}
% By $d\ge2$, we have
% \begin{equation}
% Y\le\frac{d+1}{d}=1+\frac{1}{d}\le1+\frac{1}{2}=\frac{3}{2}<2\,,
% \end{equation}
% which shows \eqref{eq:Y} and hence, $\ze s'+\dg>1$.
\\
\\
{\bf Estimate for $S_2$. } 
By \eqref{eq:l3}, we get
% By $\ze s'+\dg>1\ge0$ (see Lemma \ref{le:somze}(iii)), we get
\begin{equation}\label{eq:s2}
S_2=\sum_{j\in J_3,\,h_j<c_1\dz}|D_j|\max\left(\frac{h_j}{c_1\de_j} ,1\right)^{-\left(\ze s'+\dg\right)}\le\sum_{j\in J_3,\,h_j<c_1\dz}|D_j|\, .
\end{equation}
Let $j\in J_3$ with $h_j<c_1\dz$. By the definition of $J_3$, we have $a_j=c_0h_j^{\frac{1}{\ga}}$ or $a_j=c_2\de_j^{\frac{1}{\ga}}$. If $a_j=c_0h_j^{\frac{1}{\ga}}$, then by assumption $h_j\le c_1\dz$, so we get by Lemma \ref{le:propD}(iii)
\begin{equation}
h_w\le2h_j\le2c_1\dz\quad \textrm{for all }\quad w\in D_j \, .
\end{equation}
If $a_j=c_2\de_j^{\frac{1}{\ga}}$, we get by Lemma \ref{le:propD}(iv)
\begin{equation}
h_w\le h_j+|h_w-h_j|\le c_1\dz+\de_j\le c_1\dz+\de_0\le2c_1\dz\quad \textrm{for all }\quad  w\in D_j \, .
\end{equation}
% \begin{equation}
% h_w\le h_j+|h_w-h_j|\le c_1\dz+\de_j\le c_1\dz+\de_0= (c_1+1)\dz\le2c_1\dz\ \textrm{for all } w\in D_j \, .
% \end{equation}
In both cases we get $h_w\le2c_1\dz$ for all $w\in D_j$. Using this fact, \eqref{eq:s2} and the covering lemma (Lemma \ref{le:cover}), we obtain
\begin{equation}
S_2\le\sum_{j\in J_3,\, h_j< c_1\dz}|D_j|\ls\left\vert\left\lbrace w\in \ol\cap\obz\,\middle\vert\, h_w\le2 c_1\dz\right\rbrace\right\vert\ls\dz\, .
\end{equation}
\end{proof}

\subsection{Covering of the interior by cubes}\label{ss:covintcub}
In this subsection, we consider the part of $\om$ far enough away from $\partial\om$ and show that we can choose a family of cubes $D_x$ with centre $x\in\om$ far enough away from $\partial\om$ such that 
\begin{equation}
\ev{D_x}\le1\, 
\end{equation}
and such that the number of cubes we choose is bounded by a constant times $\dz^{-d}$. This part of the proof mimics the proof strategy of Rozenblum \cite{rozenblum1972}, \cite[Section 4.5.1]{frank2021schrodinger}.
\begin{definition}\label{de:debd}
For all $x\in\om\setminus\obd$, define $\de_x\in(0,\dz] $ by
% \begin{equation}
% \de_x:=\min\left\lbrace\dz,\,\sup\left\lbrace \dt\in(0,\dz]\,\middle\vert\, \nor{V}{\frac{d}{2}, D_x(\dt )}{\frac{d}{2}}\ls1\right\rbrace\right\rbrace\ \textrm{if } d\ge3
% \end{equation}
\begin{equation}
\de_x:=\sup\left\lbrace \dt\in(0,\dz]\,\middle\vert\, \nor{V}{\frac{d}{2}, D_x(\dt )}{\frac{d}{2}}\ls1\right\rbrace\ \textrm{if } d\ge3
\end{equation}
and
% \begin{equation}
% \de_x:=\min\left\lbrace\dz,\,\sup\left\lbrace \dt\in(0,\dz]\,\middle\vert\, \no{V}{\cB, D_x(\dt )}\ls1\right\rbrace\right\rbrace\ \textrm{if } d=2\,,
% \end{equation}
\begin{equation}
\de_x:=\sup\left\lbrace \dt\in(0,\dz]\,\middle\vert\, \no{V}{\cB, D_x(\dt )}\ls1\right\rbrace\ \textrm{if } d=2
\end{equation}
where
\begin{equation}
D_x(\dt ):=\left\lbrace y\in\R^d\,\middle\vert\, \noi{y-x}<\ha\dt\right\rbrace
\end{equation}
and where $\noi{\cdot}$ denotes the $\infty$-norm on $\R^d$. Here the constants in $\ls$ have to be chosen small enough depending on $d$.
\end{definition}

\begin{lemma}\label{le:Dintwelld}
Let $x\in\om\setminus\obd$ and let $\de_x\in(0,\dz] $ be as in Definition \ref{de:debd}. Then $
D:=D_x:=D_x(\de_x )\subset\om$ and 
\begin{equation} \label{eq:Dintwelld}
\ev{D}\le1\, .
\end{equation}
\end{lemma}

\begin{proof} Since $x\in\om\setminus\obd$ and $\de_x\le\dz$, we have $\dist{x,\partial\om}\ge\ha\sqrt{d}\dz$. By the definition of $D$, we obtain $D\subset\om$. The bound \eqref{eq:Dintwelld} can be proved as in  Lemma \ref{le:bdryev}(i). 
\end{proof}

\begin{lemma}[Covering lemma for the interior of $\om$]\label{le:coverint}
For every $x\in\om\setminus\obd$ let $\de_x\in(0,\dz] $ and $D_x:=D_x(\de_x)$ be as in Definition \ref{de:debd}. 
\begin{enumerate}[(i)]
\item Then there exists $K_0=K_0(d,\ga)\in\N$ and subfamilies $\cF_1,\ldots,\cF_{K_0}$ of oscillatory domains $D_x=D_x(\de_x)\subset\ol\cap\obz$ such that
for every $k\in\{1,\ldots,K_0\}$ all oscillatory domains in $\cF_k$ are disjoint, and moreover, 
\begin{equation}
\bigcup_{k=1}^{K_0}\dot{\bigcup_{D\in\cF_k}}D\supset\om\setminus\obd\, .
\end{equation}

\item Let $J_0$ be an index set and denote
\begin{equation}
\bigcup_{k=1}^{K_0}\cF_k=:\left\lbrace D_j\right\rbrace_{j\in J_0}\, .
\end{equation}
Then
\begin{equation}\label{eq:J0est}
|J_0|\ls\dz^{-d}\, .
\end{equation}
\end{enumerate}
\end{lemma}

\begin{proof}
{\bf Proof of (i). } 
In order to get the desired result, it suffices to apply the Besicovitch covering lemma for cubes to the family $\left\lbrace D_x\right\rbrace_{x\in\om\setminus\obd}$.
\\
\\
{\bf Proof of (ii). } 
For every $j\in J_0$, let $\de_j$ be such that $D_j:=D_{x_j}(\de_j)$ for some $x_j\in\om\setminus\obd$.
Write $J_0=J_{0,0}\cup J_{0,1}$, where $J_{0,0}$ and $J_{0,1}$ are chosen such that $\de_j=\dz$ for all $j\in J_{0,0}$, $\nor{V}{\frac{d}{2}, D_j}{}\sim1$ for all $j\in J_{0,1}$ if $d\ge3$, and $\no{V}{\cB, D_j}\sim1$ for all $j\in J_{0,1}$ if $d=2$.
% \begin{equation}
% \de_j=\dz\ \textrm{for all } j\in J_{0,0} \,,
% \end{equation}
% where $\de_j$ is such that $D_j:=D_{x_j}(\de_j)$ for some $x_j\in\om\setminus\obd$ and
% \begin{equation}
% \nor{V}{\frac{d}{2}, D_j}{}\sim1\ \textrm{for all } j\in J_{0,1}\ \textrm{if } d\ge3 \, , \ \textrm{and } \quad \no{V}{\cB, D_j}\sim1\ \textrm{for all } j\in J_{0,1}\ \textrm{if } d= 2 \, .
% \end{equation}
% and 
% \begin{equation}
% \no{V}{\cB, D_j}\sim1\ \textrm{for all } j\in J_{0,1}\ \textrm{if } d= 2 \, .
% \end{equation}
% \begin{equation}
% \de_j=\dz\ \textrm{for all } j\in J_{0,0} \,,
% \end{equation}
% % where $\de_j$ is such that $D_j:=D_{x_j}(\de_j)$ for some $x_j\in\om\setminus\obd$ and
% \begin{equation}
% \nor{V}{\frac{d}{2}, D_j}{\frac{d}{2}}\sim1\ \textrm{for all } j\in J_{0,1}\ \textrm{if } d\ge3 
% \end{equation}
% and 
% \begin{equation}
% \no{V}{\cB, D_j}\sim1\ \textrm{for all } j\in J_{0,1}\ \textrm{if } d= 2 \, .
% \end{equation}
This is possible by the definition of $\de_x$ for $x\in\om\setminus\obd$. By (i), $|D_j|=\dz^d$ and $|\om|\le1$ for all $j\in J_{0,0}$, so we get $|J_{0,0}|\ls\dz^{-d}$.
% \begin{equation}
% |J_{0,0}|\ls\dz^{-d}\, .
% \end{equation}
Using (i), we obtain
\begin{equation}
|J_{0,1}|=\sum_{j\in J_{0,1}}1\ls\sum_{j\in J_{0,1}}\nor{V}{\frac{d}{2}, D_j}{\frac{d}{2}}\ls\nor{V}{\frac{d}{2}, \om\setminus\obd}{\frac{d}{2}}\ls\nor{V}{\frac{d}{2}, \om}{\frac{d}{2}}\ls\dz^{-d}
\end{equation}
if $d\ge3$ and
\begin{equation}
|J_{0,1}|=\sum_{j\in J_{0,1}}1\ls\sum_{j\in J_{0,1}}\no{V}{\B, D_j}\ls\no{V}{\B, \om\setminus\obd}\ls\no{V}{\B, \om}\ls\dz^{-d}
\end{equation}
if $d=2$. Thus, we get \eqref{eq:J0est}.
% \begin{equation}
% |J_0|\le |J_{0,0}|+|J_{0,1}|\ls\dz^{-d}\, .
% \end{equation}
\end{proof}
\section{Conclusion of Theorem \ref{th:weightednorm} and Theorem \ref{th:main}}\label{s:conclclr}
In this section, we conclude Theorem \ref{th:main} and  we also prove Corollary \ref{co:main}. We remarked in Section \ref{s:ingredient} that these two results imply Theorem \ref{th:weightednorm}.

\subsection{Proof of Theorem \ref{th:main}}\label{ss:pfthmain}
In this subsection, we combine the results from the previous subsections to prove Lemma \ref{le:familiesabc}. From this, we deduce Lemma \ref{le:selfadj} and Lemma \ref{le:evkv}. As we have already shown in Section \ref{s:ingredient}, we obtain Theorem \ref{th:main} from these lemmata.

\begin{proof}[Proof of Lemma \ref{le:familiesabc}]
Define $K: =K_0+K_1+\ldots K_L$, where $K_0$ was defined in Lemma \ref{le:coverint}(i) and $K_l$ for $l\in\{1,\ldots,L\}$ was defined in Lemma \ref{le:cover}. Note that $K$ only depends on $\dwc$. Denote by $\cF_k^l$ with $l\in\{0,\ldots,L\}$ and $k\in\{1,\ldots,K_l\}$ the corresponding families of oscillatory domains. By Lemma \ref{le:coverint}(i) and Lemma \ref{le:cover}(i), the oscillatory domains in each $\cF_k^l$ are disjoint and moreover,
\begin{equation}
\om\supset\bigcup_{l=0}^{L}\bigcup_{k=1}^{K_l}\dot{\bigcup_{D\in\cF_k^l}}D\supset\left(\om\setminus\obd\right)\cup\left(\bigcup_{l=1}^{L}\left(\olt\cap\obd\right)\right)\supset\om\,,
\end{equation}
% \begin{align*}
% \om&\supset\bigcup_{l=0}^{L}\bigcup_{k=1}^{K_l}\dot{\bigcup_{D\in\cF_k^l}}D=\left(\bigcup_{k=1}^{K_0}\dot{\bigcup_{D\in\cF_k^0}}D\right)\cup\left(\bigcup_{l=1}^{L}\bigcup_{k=1}^{K_l}\dot{\bigcup_{D\in\cF_k^l}}D\right)\\
% &\supset\left(\om\setminus\obd\right)\cup\left(\bigcup_{l=1}^{L}\left(\olt\cap\obd\right)\right)\supset\om\,,
% \end{align*}
where we used in the last step that by $\ho\ge\sqrt{d}\dz$ and Lemma \ref{oltobho}, we have $\bigcup_{l=1}^{L}\olt\supset\obd$.
% \begin{equation}
% \bigcup_{l=1}^{L}\olt\supset\obh\supset\obz\supset\obd . 
% \end{equation}
% in the last step, see Lemma \ref{oltobho}. 
This shows (a). For (b), note that by Lemma \ref{le:Dintwelld} and Lemma \ref{le:bdryev}, we have
\begin{equation}
\ev{D}\le1\ \textrm{for all } l\in\{0,\ldots,L\}\, ,\ k\in\{1,\ldots,K_l\}\ \textrm{and } D\in\cF_k^l\,.
\end{equation}
For (c), we obtain by Lemma \ref{le:Jest} and Lemma \ref{le:coverint}(ii)
\begin{equation}
\sum_{l=0}^{L}\sum_{k=1}^{K_l}| \cF_k|\ls\dz^{-d}\, .
\end{equation}
% For (c), recall that by Lemma \ref{le:Jest}
% \begin{equation}
% \sum_{k=1}^{K_l}|\cF_k|=|J^l|\ls\dz^{-d}\ \textrm{for all } l\in\{1,\ldots,L\}\, ,
% \end{equation}
% where the index set $J^l$ is defined by $\bigcup_{k=1}^{K_l}\cF_k=:\left\lbrace D_j\right\rbrace_{j\in J^l}$. By Lemma \ref{le:coverint}(ii), we have
% \begin{equation}
% |J_0|\ls\dz^{-d}\, .
% \end{equation}
% We obtain
% \begin{equation}
% \sum_{l=0}^{L}\sum_{k=1}^{K_l}| \cF_k|\ls\dz^{-d}\, .
% \end{equation}
\end{proof}

\begin{proof}[Proof of Lemma \ref{le:selfadj}]
%    By Friedrich's theorem, closable quadratic forms that are bounded from below give rise to a corresponding selfadjoint operator. Here we will show that the quadratic form for $-\Delta_{\om}^{N}+V$ is well-defined on $H^1(\om)$ and bounded from below. Furthermore, we will show that the corresponding quadratic form norm is given by the $H^1(\om)$-norm. It follows that the quadratic form for $-\Delta_{\om}^{N}+V$ is closed, so in particular it is closable. Hence, by Friedrich's theorem, $-\Delta_{\om}^{N}+V$ is a well-defined selfadjoint operator with the $H^1(\om)$ norm as quadratic form norm.
 By Friedrich's extension, it suffices to prove that  $-\Delta_{\om}^{N}+V$ is bounded from below with the quadratic form domain $H^1(\om)$. 
Let $K=K(\dwc)\in\N$ be as in Lemma \ref{le:familiesabc}. Denote   $\tilde V:=4KV$ and recall $\no{V}{\pt ,\be} < \infty$. 
%
% $\tilde V:\om\to(-\infty,0]$ by 
%    \begin{equation}
%        \tilde V:=4KV
%    \end{equation}
%    and notice that since $\no{V}{\pt ,\be} < \infty$, we also have 
%    % $\no{\tilde V}{\pt ,\be} < \infty$
%    $\Vert{\tilde V}\Vert_{\pt ,\be} < \infty$. 
    Let $\cF_1,\ldots,\cF_K$ be the families of oscillatory domains $D\subset\om$ which we got from Lemma \ref{le:familiesabc} applied to $\tilde V$. 
    % As at the beginning of the proof of Lemma \ref{le:evkv}, w
    We compute
\begin{align}\label{eq:quadrformest}
\begin{split}
-\Delta_{\om}^{N}+V&\ge -\ha\Delta_{\om}^{N}+\ha\left(\frac{1}{K} \sum_{k=1}^K\left(-\Delta_{\om}^{N}\right)+\sum_{k=1}^K\sum_{D\in\cF_k}2V1_D\right)\\
&\ge -\ha\Delta_{\om}^{N}+\frac{1}{2K}\sum_{k=1}^K\sum_{D\in\cF_k}\left(-\Delta_{D}^{N}+\ha\tilde V1_D\right)\, .
\end{split}
\end{align}
% \begin{align*}
% -\Delta_{\om}^{N}+V&=-\ha\Delta_{\om}^{N}+\ha\left(-\Delta_{\om}^{N}+2V\right)\\
% &\ge -\ha\Delta_{\om}^{N}+\ha\left(\frac{1}{K} \sum_{k=1}^K\left(-\Delta_{\om}^{N}\right)+\sum_{k=1}^K\sum_{D\in\cF_k}1_D2V1_D\right)\\
% &\ge -\ha\Delta_{\om}^{N}+\frac{1}{2K}\left(\sum_{k=1}^K 1_D \left(-\Delta_{\om}^{N}\right) 1_D +\sum_{k=1}^K\sum_{D\in\cF_k}1_D\ha\tilde V1_D\right)\\
% &\ge -\ha\Delta_{\om}^{N}+\frac{1}{2K}\sum_{k=1}^K\sum_{D\in\cF_k}1_D\left(-\Delta_{\om}^{N}+\ha\tilde V\right)1_D\, .
% \end{align*}
Let $k \in \{1, \dots, K\}$, $D \in \cF_k$, $u \in H^1(\Omega)$ and $u_D := \tfrac{1}{|D|} \int_D u \in \R$. Note that for $v := u - u_D$, we have $v \in H^1(\Omega)$, $\int_D v = 0$ and $\int_D |\nabla u|^2 = \int_D |\nabla v|^2$. We get using $\tilde{V} \leq 0$
\begin{align*}
    &\quad\int_D |\nabla u|^2 + \int_D \frac{1}{2} \tilde{V}|u|^2 \geq \int_D |\nabla v|^2 + \int_D \frac{1}{2} \tilde{V} \left(2 |v|^2 + 2 |u_D|^2\right) \\
    & = \int_D |\nabla v|^2 + \int_D \tilde{V} |v|^2 + \left|\frac{1}{|D|} \int_D u\right|^2 \int_D \tilde{V} \ge \frac{1}{|D|} \int_D \tilde{V} \int_D |u|^2 .
\end{align*}
% \begin{align*}
%     \int_D |\nabla u|^2 + \int_D \frac{1}{2} \tilde{V}|u|^2 &= \int_D |\nabla v|^2 + \int_D \frac{1}{2} \tilde{V} |v + u_D|^2 \\
%     &\geq \int_D |\nabla v|^2 + \int_D \frac{1}{2} \tilde{V} \left(2 |v|^2 + 2 |u_D|^2\right) \\
%     &= \int_D |\nabla v|^2 + \int_D \tilde{V} |v|^2 + \left|\frac{1}{|D|} \int_D u\right|^2 \int_D \tilde{V} \\
%     &\geq 0 + \frac{1}{|D|} \int_D |u|^2 \int_D \tilde{V} = \frac{1}{|D|} \int_D \tilde{V} \int_D |u|^2 .
% \end{align*}
In the fourth step we used that $\int_D |\nabla v|^2 + \int_D \tilde{V}|v|^2 \geq 0$ by the choice of the family $\cF_k$ and $\int_D v = 0$ for the first two summands, and we used Jensen's inequality and $\tilde{V} \leq 0$ for the last summand. We deduce that in the sense of quadratic forms,
\begin{equation}\label{eq:quadrsum}
   -\Delta_{D}^{N}+\ha1_D\tilde V1_D \geq -\left(\frac{1}{|D|} \int_D |\tilde{V}|\right) 1_D .
\end{equation}
% \begin{equation}
%     1_D \left(-\Delta^N_\Omega + \frac{1}{2} \tilde{V}\right) 1_D \geq \frac{1}{|D|} \int_D \tilde{V} 1_D = - \frac{1}{|D|} \int_D |\tilde{V}| 1_D .
% \end{equation}
We obtain by \eqref{eq:quadrformest} and \eqref{eq:quadrsum},
\begin{equation}
    -\Delta^N_\Omega + V \geq  \frac{1}{2} \left(-\Delta^N_\Omega\right) - \left(\frac{1}{2K} \sum_{k = 1}^K \sum_{D \in \cF_k} \frac{1}{|D|} \int_D |\tilde{V}|\right) 1_\Omega ,
\end{equation}
% \begin{align*}
%     -\Delta^N_\Omega + V &\geq \frac{1}{2} \left(-\Delta^N_\Omega\right) + \frac{1}{2K} \sum_{k = 1}^K \sum_{D \in \cF_k}  \left(-\Delta_{D}^{N}+\ha\tilde V1_D\right) \\
%   % &\geq \frac{1}{2} \left(-\Delta^N_\Omega\right) - \frac{1}{2K} \sum_{k = 1}^K \sum_{D \in \cF_k} \left(\frac{1}{|D|} \int_D |\tilde{V}|\right) 1_D \\
%     &\geq \frac{1}{2} \left(-\Delta^N_\Omega\right) - \left(\frac{1}{2K} \sum_{k = 1}^K \sum_{D \in \cF_k} \frac{1}{|D|} \int_D |\tilde{V}|\right) 1_\Omega ,
% \end{align*}
% Note that
% \begin{equation}
%     \left(\frac{1}{2K} \sum_{k = 1}^K \sum_{D \in \cF_k} \frac{1}{|D|} \int_D |\tilde{V}| \right) <\infty
% \end{equation}
where  the constant in the last part is finite
since $\sum_{k = 1}^K |\cF_k|\ls\dz(V)^{-d}<\infty$.
% \begin{align*}
%     -\Delta^N_\Omega + V &\geq \frac{1}{2} \left(-\Delta^N_\Omega\right) + \frac{1}{2K} \sum_{k = 1}^K \sum_{D \in \cF_k} 1_D \left(- \Delta^N_\Omega + \frac{1}{2} \tilde{V}\right) 1_D \\
%     &\geq \frac{1}{2} \left(-\Delta^N_\Omega\right) - \frac{1}{2K} \sum_{k = 1}^K \sum_{D \in \cF_k} \frac{1}{|D|} \int_D |\tilde{V}| 1_D \\
%     &\geq \frac{1}{2} \left(-\Delta^N_\Omega\right) - \left(\frac{1}{2K} \sum_{k = 1}^K \sum_{D \in \cF_k} \frac{1}{|D|} \int_D |\tilde{V}|\right) 1_\Omega ,
% \end{align*}
% that is, for every $u \in H^1(\Omega)$, we have
% \begin{equation}
%     \int_\Omega |\nabla u|^2 + \int_\Omega V |u|^2 \geq \frac{1}{2} \int_\Omega |\nabla u|^2 - \left(\frac{1}{2K} \sum_{k = 1}^K \sum_{D \in \cF_k} \frac{1}{|D|} \int_D |\tilde{V}|\right) \int_\Omega |u|^2 .
% \end{equation}
It follows that the quadratic form for $-\Delta^N_\Omega + V$ with domain $H^1(\Omega)$ is well-defined, bounded from below and its form norm is given by the $H^1(\Omega)$-norm. This finishes the proof of the self-adjointness of $-\Delta^N_\Omega + V$.
\end{proof}

\begin{proof}[Proof of Lemma \ref{le:evkv}]
As at the beginning of the proof of Lemma \ref{le:selfadj}, in the sense of quadratic forms, we have
\begin{align}\label{eq:quadrform1Kest}
-\Delta_{\om}^{N}+\frac{1}{K}V %&\ge\frac{1}{K}\sum_{k=1}^K\left(-\Delta_{\om}^{N}\right)+\frac{1}{K}\sum_{k=1}^K\sum_{D\in\cF_k}V1_D\\
\ge \frac{1}{K} \sum_{k = 1}^K \sum_{D \in \cF_k}  \left(-\Delta_{D}^{N}+ V1_D\right),
% \sum_{k=1}^K\sum_{D\in\cF_k}1_D\frac{1}{K}\left(-\Delta_{D}^{N}\right)1_D+\sum_{k=1}^K\sum_{D\in\cF_k}1_D\frac{1}{K}V1_D\\
% &=\sum_{k=1}^K\sum_{D\in\cF_k}1_D\frac{1}{K}\left(-\Delta_{D}^{N}+V\right)1_D\, ,
\end{align}
% \begin{align*}
% -\Delta_{\om}^{N}+\frac{1}{K}V&\ge\sum_{k=1}^K\frac{1}{K}\left(-\Delta_{\om}^{N}\right)+\sum_{k=1}^K\sum_{D\in\cF_k}1_D\frac{1}{K}V1_D\\
% &\ge\sum_{k=1}^K\sum_{D\in\cF_k}1_D\frac{1}{K}\left(-\Delta_{D}^{N}\right)1_D+\sum_{k=1}^K\sum_{D\in\cF_k}1_D\frac{1}{K}V1_D\\
% &=\sum_{k=1}^K\sum_{D\in\cF_k}1_D\frac{1}{K}\left(-\Delta_{D}^{N}+V\right)1_D\, ,
% \end{align*}
where we used $V\le0$ and $\om=\bigcup_{k=1}^K\bigcup_{D\in\cF_k}D$ in the first step. In the second step, we used that for every $k\in\{1,\ldots,K\}$ the oscillatory domains $D\in\cF_k$ are disjoint and therefore, $\int_\om|\nabla u|^2\ge\sum_{D\in\cF_k}\int_D|\nabla u|^2$ for all $u\in H^1(\om)$. By Lemma \ref{le:familiesabc}(b), we have
\begin{equation}
\ev{D}\le1\, 
\end{equation}
for every $k\in\{1,\ldots,K\}$ and every  $D\in\cF_k$. Hence, by the min-max principle \cite[Theorem 12.1, version 2]{liebloss}\footnote{As the proof \cite[Theorem 12.1, version 2]{liebloss} shows, in fact, the subspace $M$ need not be a subset of $H^1(\Omega)$ but it suffices to take $M\subset L^2(\Omega)$.}, for every $k\in\{1,\ldots,K\}$ and  $D\in\cF_k$ there exists a function $u^D\in H^1(D)\subset L^2(\om)$ such that
\begin{equation}
\int_D|\nabla u|^2+\int_DV|u|^2\ge0\ \textrm{for all } u\in H^1(\om)\ \textrm{with } \int_\om \overline{u^D}u=0\, .
\end{equation}
It follows that if $u\in H^1(\om)$ is in the orthogonal complement in the $L^2(\om)$ sense of $\Span\left\{u^D\mid D\in\cF
_k,\, k=1,\ldots K\right\}$, then 
\begin{equation}
\int_\om|\nabla u|^2+\int_\om\frac{1}{K}V|u|^2\ge \frac{1}{K}\sum_{k=1}^K\sum_{D\in\cF_k}\int_D\left(|\nabla u|^2+\int_DV|u|^2\right)\ge0\, .
\end{equation}
Since the orthogonal complement of $\Span\left\{u^D\mid D\in\cF
_k,\, k=1,\ldots K\right\}$ is a subspace of $L^2(\om)$ of dimension at most $\sum_{k=1}^K|\cF_k|$, we obtain by the min-max principle \cite[Theorem 12.1, version 2]{liebloss} and \eqref{eq:quadrform1Kest}
\begin{equation}
\evv{\om}{\frac{1}{K}V}\le\sum_{k=1}^K|\cF_k|\ls\dz^{-d}\, ,
\end{equation}
where we used Lemma \ref{le:familiesabc}(c) in the last step.
\end{proof}

\subsection{Proof of Corollary \ref{co:main}}\label{ss:pfcolpnorm}
For the proof of Corollary \ref{co:main}, we need the following lemma.
\begin{lemma}[A subset of the $\gamma$ with $\beta < 1$]\label{le:bele1}
Let $d \geq 2$.
\begin{enumerate}[(i)]
\item 
If $\gamma \in \left[\tfrac{2(d-1)}{2d - 1} , 1 \right]$, then
\begin{equation}
    \beta = \beta(d, \gamma) = \frac{1}{d + 1} \left(\frac{d-1}{\gamma} + 1\right)\left[\frac{1}{d}\left(\frac{d-1}{\gamma} + 1\right)^2 - d\right] < 1 .
\end{equation}
\item
If $\gamma \in (0,1)$ is such that $\beta < 1$, then $\no{\cdot}{\tilde{p},\beta} \lesssim \no{\cdot}{p}$ for all $p>{\frac{\tilde{p}}{1-\beta}}$ .
\end{enumerate}
\end{lemma}
\begin{proof}
{\bf Proof of (i). }
Let $\gamma \in \left[\tfrac{2(d-1)}{2d-1} , 1\right]$. Note that
\begin{equation}
    Y := \frac{1}{d}\left(\frac{d-1}{\gamma} + 1\right) \leq \frac{1}{d}\left(\frac{2d-1}{2} + 1\right) = 1 + \frac{1}{2d} ,
\end{equation}
so $Y^2-1\le \frac{1}{d} \left(1 + \frac{1}{4d}\right)$. Hence, we have
%
%\begin{equation}
%    Y^2 - 1 \leq \left(1 + \frac{1}{2d}\right)^2 - 1 = \frac{1}{d} \left(1 + \frac{1}{4d}\right) .
%\end{equation}
% \begin{equation}
%     Y = \frac{1}{d}\left(\frac{d-1}{\gamma} + 1\right) \leq \frac{1}{d}\left(\frac{2d-1}{2} + 1\right) = \frac{1}{d}\left(d + \frac{1}{2}\right) = 1 + \frac{1}{2d} ,
% \end{equation}
% so
% \begin{equation}
%     Y^2 - 1 \leq \left(1 + \frac{1}{2d}\right)^2 - 1 = 1 + \frac{1}{d} + \frac{1}{4d^2} - 1 = \frac{1}{d} \left(1 + \frac{1}{4d}\right) .
% \end{equation}
\begin{align*}
    \beta = \frac{d^2}{d + 1}Y \left[Y^2 - 1\right] \leq \frac{d^2}{d + 1} \left(1 + \frac{1}{2d}\right)\frac{1}{d}\left(1 + \frac{1}{4d}\right) < 1 \,.
\end{align*}
% \begin{align*}
%     \beta &= \frac{1}{d + 1}\left(\frac{d-1}{\gamma} + 1\right)\left[\frac{1}{d}\left(\frac{d-1}{\gamma} + 1\right)^2 - d\right] \\
%     &= \frac{d^2}{d + 1}Y \left[Y^2 - 1\right] \leq \frac{d^2}{d + 1} \left(1 + \frac{1}{2d}\right)\frac{1}{d}\left(1 + \frac{1}{4d}\right) \\
%     &= \frac{d}{d + 1}\left(1 + \frac{1}{4d} + \frac{1}{2d} + \frac{1}{8d^2}\right) = \frac{1}{d + 1}\left(d + \frac{3}{4} + \frac{1}{8d}\right) \\
%     &\leq \frac{1}{d + 1} \left(d + \frac{3}{4} + \frac{1}{8}\right) < 1 ,
% \end{align*}
% where we used $d \geq 2 \geq 1$ in the second last step. 
\\ \\
{\bf Proof of (ii). }
Let $q:=\frac{p}{\pt}>\frac{1}{1-\be}>1$ and note that $\frac{1}{q'}=1-\frac{1}{q}>1-(1-\be)=\be$, so $\be q'<1$. By Hölder's inequality, we get
\begin{align*}
\vert V\vert_{\pt,\be}^{\pt}&=\inx{\uoml}{x}\hxm^{-\be}|V(x)|^{\pt}\le\left(\inx{\uoml}{x}\hxm^{-\be q'}\right)^{\frac{1}{q'}}\left(\inx{\uoml}{x}|V(x)|^{\pt q}\right)^{\frac{1}{q}}\\
&\ls\nor{V}{p}{\pt}\, ,
\end{align*}
where we used $\be q'<1$ and $\pt q=p$ in the last step. Furthermore, since $|\om|\le1$ and $p>\pt\ge\frac{d}{2}$, we can apply Jensen's inequality to get
\begin{align*}
\nor{V}{p}{p}&=\int_\om|V|^p=|\om|\frac{1}{|\om|}\int_\om\left(|V|^{\frac{d}{2}}\right)^{\frac{2p}{d}}\ge|\om|\left(\frac{1}{|\om|}\int_\om|V|^{\frac{d}{2}}\right)^{\frac{2p}{d}}=|\om|^{1-\frac{2p}{d}}\nor{V}{\frac{d}{2}}{p}\ge\nor{V}{\frac{d}{2}}{p}\, .
\end{align*}
For $d=2$, we have $\no{V}{p}\gs\no{V}{\cB,\om}$ by \cite[Chapter 5.1, Theorem 3, p.~155]{rao1991orlicz} and $p>\pt\ge\frac{d}{2}=1$.
Using Definition \ref{de:hxlseminorm}(iii), we obtain $\no{V}{\pt,\be}\ls\no{V}{p}$ for any $d\ge2$. 
\end{proof}

\begin{proof}[Proof of Corollary \ref{co:main}]
By Lemma \ref{le:bele1}, we have $\no{\cdot}{\tilde{p},\beta} \lesssim \no{\cdot}{p}$. Therefore, by Theorem \ref{th:main},
\begin{align*}
\ev{\om}&\ls\left[\min\left(\frac{\ho}{\sqrt{d}},\, \no{V}{\pt,\be}^{-\frac{1}{2}}\right)\right]^{-d}=\max\left(\left(\frac{\ho}{\sqrt{d}}\right)^{-d},\, \no{V}{\pt,\be}^{\frac{d}{2}}\right)\\
&\ls1+\no{V}{\pt,\be}^{\frac{d}{2}}\ls1+\no{V}{p}^{\frac{d}{2}}\, .
\end{align*}
\end{proof}

%%%%%%%%%%%%%
%%%%%%%%%%%%%%
%%%%%%%%%%%%%%%

\section{Weyl's law for Schr\"odinger operators (Theorem \ref{th:weyllawforapotential})}\label{se:weyln}

In this section, we deduce Theorem \ref{th:weyllawforapotential} using Theorem \ref{th:weightednorm}. The main idea is to first reduce to Weyl's law for continuous compactly supported potentials, namely 
\begin{equation}\label{eq:schroedleq2pibint}
    N\left(-\Delta^N_\Omega + \lambda W\right) = (2 \pi)^{-d} \left|B_1(0)\right| \lambda^\frac{d}{2} \int_\Omega |W|^\frac{d}{2} + o \left(\lambda^\frac{d}{2}\right) \mathrm{\ as\ } \lambda \rightarrow \infty .
\end{equation}
for all $W \in C_c (\Omega)$ with $W \leq 0$. Using \eqref{eq:schroedleq2pibint} combined with the Cwikel-Lieb-Rozenblum type bound (Theorem \ref{th:weightednorm}), we can then deduce Theorem \ref{th:weyllawforapotential}.
% For the proof of Theorem \ref{th:weyllawforapotential}, we prove an upper and a lower bound separately. The proof of the lower bound, see Lemma \ref{le:lowerboundweylspotentialadd}, does not depend on the regularity of the domain $\om$ and it is a consequence of the Weyl law for Schr\"odinger operators on $\R^d$, see for example \cite[p.~59]{namfa2}. 

%The proof of the Theorem \ref{th:weyllawforapotential} is inspired by Weyl's proof strategy for the Weyl law for constant potentials on bounded domains with Dirichlet boundary conditions \cite{weyl,weyl1912asymptotische} combined with the Cwikel-Lieb-Rozenblum type bound (Theorem \ref{th:weightednorm}). Indeed, we may use Theorem  \ref{th:weightednorm} to approximate the potential $V$ by a nicer one $V_n$, and then implement the usual semiclassical approximation for $V_n$. More precisely, by the min-max principle, we have
%\begin{align}\label{eq:weylvvncomp-intro}
%    N\left(-\Delta^N_\Omega + \lambda V \right) \leq N\left((1 - \delta) \left(-\Delta^N_\Omega\right) + \lambda V_n\right) + N \left(\delta \left(-\Delta^N_\Omega\right) + \lambda \left(V - V_n\right)\right) ,
%   \end{align}
%where we can take $n\to \infty$ first and $\delta\to 0$ later. See Section \ref{ss:strweyl} for details.

\subsection{Reduction to compactly supported potentials}\label{ss:strweyl} Let $V$ be as in Theorem \ref{th:weyllawforapotential}. Assume that we have \eqref{eq:schroedleq2pibint} for $0\ge W \in C_c (\Omega)$. Since $\normiii{\cdot}:=\no{\cdot}{\pt,\be}$ is a weighted $L^{\pt}$-norm on $\Omega$ with $\pt < \infty$, % and a weight that is absolutely continuous with respect to the Lebesgue measure, 
there exists a sequence $(V_n)_{n \in \N} \subset C_c (\Omega)$ with $V_n \leq 0$ such that
\begin{equation}\label{eq:W-density-0}
  \no{V_n - V}{\tfrac{d}{2}} \ls  \normiii{V - V_n} \rightarrow 0 \mathrm{\ as\ } n \rightarrow \infty.
\end{equation}
Let $\delta \in (0,1)$. Recall that from the min-max principle \cite[Theorem 12.1, version 2]{liebloss}, one can deduce that
\begin{equation}
    N (A + B) \leq N(A) + N(B)
\end{equation}
for any two self-adjoint operators $A$ and $B$ defined on the same Hilbert space with the same quadratic form domain. We have for every $n \in \N$ and $\lambda > 0$
\begin{align}\label{eq:weylvvncomp}
\begin{split}
    N\left(-\Delta^N_\Omega + \lambda V \right) &\leq N\left((1 - \delta) \left(-\Delta^N_\Omega\right) + \lambda V_n\right) + N \left(\delta \left(-\Delta^N_\Omega\right) + \lambda \left(V - V_n\right)\right) \\
    &= N\left(-\Delta^N_\Omega + \lambda \frac{V_n}{1 - \delta}\right) + N \left(-\Delta^N_\Omega + \lambda \frac{V - V_n}{\delta}\right) \\
    &\leq N\left(-\Delta^N_\Omega + \lambda \frac{V_n}{1 - \delta} \right) + C_\Omega\left( 1+ \delta^{-\frac{d}{2}} \lambda^{\frac{d}{2}} \normiii{V - V_n}^\frac{d}{2}\right),
    \end{split}
\end{align}
 where we used Theorem \ref{th:weightednorm} in the last step. 
% Suppose we had already shown that for any $W \in C^\infty_c (\Omega)$ with $W \leq 0$
% \begin{equation}\label{eq:schroedleq2pibint}
%     N\left(-\Delta^N_\Omega + \lambda W\right) \leq (2 \pi)^{-d} \left|B_1(0)\right| \lambda^\frac{d}{2} \int_\Omega |W|^\frac{d}{2} + o \left(\lambda^\frac{d}{2}\right) \mathrm{\ as\ } \lambda \rightarrow \infty .
% \end{equation}
%Note that $\no{V_n - V}{\tfrac{d}{2}} \rightarrow 0$ as $n \rightarrow \infty$ by $\normiii{V_n - V} \rightarrow 0$ as $n \rightarrow \infty$. 
Using \eqref{eq:W-density-0}, \eqref{eq:weylvvncomp}, \eqref{eq:schroedleq2pibint} for $W = \tfrac{V_n}{ 1 - \delta}$  and the definition of $\normiii{\cdot}$, we get
\begin{align}\label{eq:weylvlimusingvnclr}
\begin{split}
    &\quad \limsup_{\lambda \rightarrow \infty} \lambda^{- \frac{d}{2}} N\left(-\Delta^N_\Omega + \lambda V\right) \\
    &\leq \limsup_{\delta \rightarrow 0} \limsup_{n \rightarrow \infty} \limsup_{\lambda \rightarrow \infty} \lambda^{-\frac{d}{2}} \left(N\left(-\Delta^N_\Omega + \lambda \frac{V_n}{1 - \delta}\right) + C_\Omega\left( 1+ \delta^{-\frac{d}{2}} \lambda^\frac{d}{2} \normiii{V - V_n}^\frac{d}{2}\right)\right) \\
    &\leq \limsup_{\delta \rightarrow 0} \limsup_{n \rightarrow \infty} \left((2\pi)^{-d} \left|B_1(0)\right| \int_\Omega \left|\frac{V_n}{1 - \delta}\right|^\frac{d}{2} + C_\Omega \delta^{- \frac{d}{2}} \normiii{V - V_n}^\frac{d}{2}\right) \\
    &= (2 \pi)^{-d} \left|B_1(0)\right|\limsup_{\delta \rightarrow 0}  \int_\Omega \left|\frac{V}{1 - \delta}\right|^\frac{d}{2} = (2\pi)^{-d} \left|B_1(0)\right| \int_\Omega |V|^\frac{d}{2} ,
    \end{split}
\end{align}
which is the desired upper bound. For the corresponding lower bound, we replace $V$ by $(1-\de)V_n$, and we replace $V_n$ by $(1-\de)V$ in \eqref{eq:weylvvncomp} to get
\begin{equation}
    N\left(-\Delta^N_\Omega + \lambda V \right)\ge N\left(-\Delta^N_\Omega + \lambda (1-\de)V_n \right) -C_\Omega\left(1+ \delta^{-\frac{d}{2}} (1-\de)^{\frac{d}{2}} \lambda^{\frac{d}{2}} \normiii{V - V_n}^\frac{d}{2}\right),
\end{equation}
and then proceeding as above. 
% in \eqref{eq:weylvlimusingvnclr}, we obtain the desired lower bound
%\begin{equation}\label{eq:weyllowerbd}
%    N \left(-\Delta^N_\Omega + \lambda V\right) \ge (2 \pi)^{-d} \left|B_1(0)\right| \lambda^{\frac{d}{2}} \int_\Omega |V|^{\frac{d}{2}} + o \left(\lambda^{\frac{d}{2}}\right) \mathrm{\ as\ } \lambda \rightarrow \infty .
%\end{equation}
% Therefore, it remains to show \eqref{eq:schroedleq2pibint} for all $W \in C^\infty_c (\Omega)$ with $W \leq 0$.

%In the following, we explain the proof strategy for the upper bound
%\begin{equation}\label{eq:weylupperbd}
%    N \left(-\Delta^N_\Omega + \lambda V\right) \le (2 \pi)^{-d} \left|B_1(0)\right| \lambda^{\frac{d}{2}} \int_\Omega |V|^{\frac{d}{2}} + o \left(\lambda^{\frac{d}{2}}\right) \mathrm{\ as\ } \lambda \rightarrow \infty 
%\end{equation}
%assuming \eqref{eq:schroedleq2pibint}. Afterwards, we explain how to prove the corresponding lower bound, which follows the same proof strategy up to some small modifications. The proof of \eqref{eq:schroedleq2pibint} will be given in Section \ref{ss:weyldet}.
%

\subsection{Weyl's law for compactly supported potentials}\label{ss:weyldet}

Now we prove \eqref{eq:schroedleq2pibint} for $0\ge W \in C_c (\Omega)$. This result can also be found in \cite[Theorem 4.29]{frank2021schrodinger} for the Laplacian on $\R^d$. For the reader's convenience, we explain the proof below since our setting is slightly different. We follow the proof strategy of Weyl, namely we cover the support of $W$ by small cubes of side-length independent of $\lambda$ such that each cube is completely contained in $\Omega$. We then apply Weyl's law for constant potentials on cubes. \\

Let $m_0 \in \N$ be such that
\begin{equation}
    \sqrt{d} 2^{-m_0} < \dist{\supp W , \partial \Omega}.
\end{equation}
Then every cube of side-length at most $2^{- m_0}$ intersecting $\supp W$ is contained in $\Omega$. For every $m \in \N$, $j \in \Z^d$ let
\begin{equation}
    Q^m_j := 2^{-m} \left(j + (0,1)^d\right)
\end{equation}
be the open cube of side-length $2^{-m}$ whose bottom left corner is at $2^{-m} j \in \R^d$. For every $m \in \N$ let
\begin{equation}
    J_m := \left\{j \in \Z^d \bigm| Q^m_j \cap \supp W \neq \emptyset \right\}.
\end{equation}
\textbf{Upper bound. } 
We claim that for every $m \geq m_0$, we have
\begin{equation}\label{eq:schroedleqsumjschroed}
    N\left(-\Delta^N_\Omega + \lambda W\right) \leq \sum_{j \in J_m} N \left(-\Delta^N_{Q^m_j} + \lambda W\right) .
\end{equation}
This can be seen as follows. By the min-max principle \cite[Theorem 12.1, Version 2]{liebloss}, we know that for every $j \in J_m$ there exists an $N (-\Delta^N_{Q^m_j} + \lambda W)$-dimensional subspace of $L^2(Q^m_j)$, which we call $M_j$, such that 
\begin{equation}\label{eq:testge0weyl}
    \int_{Q^m_j} |\nabla u|^2 + \int_{Q^m_j} \lambda W |u|^2 \geq 0
\end{equation}
for all $u \in H^1 (Q^m_j)$ that are in the orthogonal complement of $M_j$ with respect to $L^2 (Q^m_j)$. Let $M \subset L^2(\Omega)$ be the span of all $M_j$ for $j \in J_m$, where we extend functions in $L^2(Q^m_j)$ by zero on $\Omega \setminus Q^m_j$. Note that $M$ has dimension at most
\begin{equation}
    \sum_{j \in J_m} N\left(- \Delta ^N_{Q^m_j} + \lambda W\right) .
\end{equation}
By \eqref{eq:testge0weyl}, $\Omega \supset \cup_{j \in J_m} \overline{Q_j} \supset \supp W$ and since the cubes $Q_j$ are disjoint, we get for every $u \in H^1(\Omega)$ in the orthogonal complement of $M$ with respect to $L^2(\Omega)$
\begin{equation}
    \int_\Omega |\nabla u|^2 + \int_\Omega \lambda W |u|^2 \geq \sum_{j \in J_m} \left(\int_{Q^m_j} |\nabla u|^2 + \int_{Q^m_j} \lambda W |u|^2\right) \geq 0 .
\end{equation}
By the min-max principle \cite[Theorem 12.1, version 2]{liebloss},
% \footnote{As the proof \cite[Theorem 12.1, Version 2]{liebloss} shows, in fact, the subspace $M$ need not be a subset of $H^1(\Omega)$ but it suffices to take $M\subset L^2(\Omega)$.}
we obtain \eqref{eq:schroedleqsumjschroed}. Therefore, for every $m \in \N$, $m \geq m_0$
\begin{equation}
    N\left(- \Delta^N_\Omega + \lambda W\right) \leq \sum_{j \in J_m} N \left(-\Delta^N_{Q^m_j} + \lambda W\right) \leq \sum_{j \in J_m} N\left(- \Delta^N_{Q^m_j} - \lambda \sup_{x \in Q^m_j} |W(x)|\right) .
\end{equation}
By Weyl's law for constant potentials with Neumann boundary conditions on cubes, see for example \cite[Theorem 3.20]{frank2021schrodinger},
\begin{align*}
    \limsup_{\lambda \rightarrow \infty} \lambda^{- \frac{d}{2}} N\left(-\Delta^N_\Omega + \lambda W\right) &\leq \sum_{j \in J_m} \limsup_{\lambda \rightarrow \infty}\lambda^{- \frac{d}{2}} N \left(- \Delta^N_{Q^m_j} - \lambda \sup_{x \in Q^m_j} |W(x)|\right) \\
    &= \sum_{j \in J_m} \left((2\pi)^{-d} \left|B_1(0)\right| \left|Q^m_j\right| \left(\sup_{x \in Q^m_j} |W(x)|\right)^\frac{d}{2}\right) .
\end{align*}
Since $W \in C_c (\Omega)$, the right hand side agrees with  \eqref{eq:schroedleq2pibint} as $m\to \infty$. 
%converges to
%\begin{equation}
%    (2 \pi)^{-d} \left|B_1(0)\right| \int_\Omega |W|^\frac{d}{2}
%\end{equation}
%as $m \rightarrow \infty$.

\bigskip
\textbf{Lower bound. } 
For every $m \in \N$, $m \geq m_0$, we have 
\begin{equation}\label{eq:weyllowerdiri}
    N\left(-\Delta^N_\Omega + \lambda W\right) \geq \sum_{j \in J_m} N \left(-\Delta^D_{Q^m_j} + \lambda W\right) .
\end{equation}
For the proof of \eqref{eq:weyllowerdiri}, note that by the min-max principle \cite[Theorem 12.1, Version 3]{liebloss} for every $j \in J_m$  there exists an $N \left(-\Delta^D_{Q^m_j} + \lambda W\right)$-dimensional subspace  of $H^1_0 (Q^m_j)$, which we call $M_j$, such that 
\begin{equation}\label{eq:lowerle0weyl}
    \int_{Q^m_j} |\nabla u_j|^2 + \int_{Q^m_j} \lambda W |u_j|^2 < 0\ \textrm{for all } 0\not\equiv u_j\in M_j. 
\end{equation}
Since the cubes $Q^m_j$,  $j \in J_m$ are disjoint and each $M_j\subset H^1_0 (Q^m_j)\subset H^1(\om)$, if we denote by $M$ the span of all $M_j$, then $M\subset H^1(\om)$ is a subspace of dimension
\begin{equation}
    \sum_{j \in J_m} N \left(-\Delta^D_{Q^m_j} + \lambda W\right) .
\end{equation}
By $\supp W\subset \cup_{j \in J_m} \overline{Q_j}$ and \eqref{eq:lowerle0weyl}, we obtain for every $0\not\equiv  u=\sum_{j \in J_m} u_j\in M$ with $u_j\in M_j$ for each $j \in J_m$, 
\begin{equation}
    \int_\Omega |\nabla u|^2 + \int_\Omega \lambda W |u|^2 = \sum_{j \in J_m} \left(\int_{Q^m_j} |\nabla u_j|^2 + \int_{Q^m_j} \lambda W |u_j|^2\right) < 0 .
\end{equation}
By the min-max principle \cite[Theorem 12.1, Version 3]{liebloss}, we get \eqref{eq:weyllowerdiri}. Using \eqref{eq:weyllowerdiri} and Weyl's law for constant potentials on cubes \cite{weyl,weyl1912asymptotische}, we get
\begin{align*}
    &\qquad \liminf_{\lambda \rightarrow \infty} \lambda^{- \frac{d}{2}} N\left(-\Delta^N_\Omega + \lambda W\right) \geq \liminf_{\lambda \rightarrow \infty} \lambda^{- \frac{d}{2}} \sum_{j \in J_m} N \left(-\Delta^D_{Q^m_j} + \lambda W\right)\\
    &\geq \sum_{j \in J_m} \liminf_{\lambda \rightarrow \infty}\lambda^{- \frac{d}{2}} N \left(- \Delta^D_{Q^m_j} - \lambda \inf_{x \in Q^m_j} |W(x)|\right) \\
    &= \sum_{j \in J_m} \left((2\pi)^{-d} \left|B_1(0)\right| \left|Q^m_j\right| \left(\inf_{x \in Q^m_j} |W(x)|\right)^\frac{d}{2}\right).
\end{align*}
Taking $m \rightarrow \infty$, we conclude \eqref{eq:schroedleq2pibint}.  The proof of Theorem \ref{th:weyllawforapotential} is complete. 
%
%the right hand side converges to
%\begin{equation}
%    (2 \pi)^{-d} \left|B_1(0)\right| \int_\Omega |W|^\frac{d}{2}.
%\end{equation}

%\begin{proof}[Proof of ]
%    Theorem \ref{th:weyllawforapotential} is a direct consequence of the proof strategy explained in Section \ref{ss:strweyl} combined with Lemma \ref{le:upperboundcompactlysupportedpotentials}.
%\end{proof}
% {\bf Proof of Theorem \ref{th:weyllawforapotential}. }
% The lower bound was proved in Lemma \ref{le:lowerboundweylspotentialadd}. The upper bound is a direct consequence of the proof strategy explained in Section \ref{ss:strweyl} combined with Lemma \ref{le:upperboundcompactlysupportedpotentials}.
\begin{remark}
    In dimension $d\ge3$, we can obtain the lower bound in Theorem \ref{th:weyllawforapotential} for all potentials $V\in L^{\frac{d}{2}}(\om)$, $V\le0$ by comparing with $-\Delta_{\R^d} +\lambda V$ and using Weyl's law for Schr\"odinger operators on $\R^d$ \cite[Theorem 4.46]{frank2021schrodinger}. However, this is not true in dimension $d=2$, see \cite[Remark after Theorem 4.46]{frank2021schrodinger}. 
    % In Lemma ???? we will show that Weyl's law for Schr\"odinger operators on $\R^d$ also holds in dimension $d=2$ if we assume stronger integralbility assumptions on $V$.
\end{remark}
\section{Example with non-semiclassical behaviour (Theorem \ref{th:example})}\label{se:example}

In this section we prove Theorem \ref{th:example}. We explain the proof strategy in Section \ref{ss:strex} and the details are given in Section \ref{ss:exdet}.

\subsection{General  strategy}\label{ss:strex}
In this subsection, we explain for fixed $\gamma\in\left(\frac{d-1}{d},1\right)$ how to construct a  $\gamma$-Hölder domain $\Omega\subset\R^d$ and a potential $V: \Omega \rightarrow \left(-\infty,0\right]$ with $V \in L^{\frac{d}{2}}(\Omega)$ such that \eqref{eq:thexeq} holds. We construct the $\gamma$-Hölder domain $\Omega$ in the same way as Netrusov and Safarov \cite[Theorem 1.10]{netrusov2005weyl}.
The potential $V$ will be chosen such that it grows near the boundary of $\om$. This will allow us for certain values of $\lambda$ going to infinity to find significantly more  than $ \lambda^{ \frac{d}{2}}$ negative eigenvalues of $- \Delta^N_\Omega + \lambda V$.
\bigskip

We start by fixing $M := 2^m$, where $m \in \N$ is chosen large enough depending on $\gamma$. The main part of $\Omega$ will be given by the subgraph
\begin{equation}
    \left\{ \left(x', x_d\right) \in \R^{d - 1} \times \R \bigm| x' \in Q^{(d-1)} , 0 < x_d < f(x')\right\}
\end{equation}
of a $\gamma$-Hölder continuous function $f$ on the $(d-1)$-dimensional unit cube $Q^{(d-1)}$ that vanishes on the boundary of $Q^{(d-1)}$. For every $j \in \N$ we can decompose $Q^{(d-1)}$ into $M^{(d-1)j}$ small cubes of side-length $M^{-j}$. The function $f$ will be chosen in such a way that for any $j \in \N$ it oscillates on the order of magnitude $M^{- j \gamma}$ on each of the small cubes of side-length $M^{-j}$. 
% \luca{Independent of the content, what does the usage of "Morally" mean here? Maybe it is just a way of using this I am unaware of, but it seems weird to me.} 
Intuitively speaking, $f$ looks no better than a $\gamma$-Hölder continuous function on each of the small cubes for every length scale $M^{-j}$, $j \in \N$. \\

The potential $V$ will be chosen such that it is large close to the boundary of $\Omega$. We define 
\begin{equation}
    V \left(x' , x_d\right) := -c \left(f(x') - x_d\right)^{\frac{2}{d} (-1 + \epsilon)} \mathrm{\ for\ } \left(x', x_d\right) \in \Omega \subset \R^{d-1} \times \R
\end{equation}
for a suitably chosen $0 < \epsilon < (d-1)\left(1/\gamma- 1\right)$ and a constant $c = c(d, \gamma, \epsilon) > 0$. Note that $V \in L^{\frac{d}{2}}(\Omega)$ since $\epsilon > 0$.
In the following, for fixed $j \in \N$, we can for simplicity think of $f$ as a $j$-dependent constant $c(j)>0$ plus a small spike of height $M^{-j\gamma}$ on each of the $M^{(d-1)j}$ $(d-1)$-dimensional small cubes of side-length $M^{-j}$. 
We denote these small cubes of side-length $M^{-j}$ by $Q(j,k)$, $k \in \left\{1, \dots, M^{(d-1)j}\right\}$. Moreover, we define
\begin{equation}
    \Omega_{j,k} := \left\{\left(x', x_d\right) \in Q(j,k) \times \R \bigm| c(j) < x_d < f(x')\right\},
\end{equation}
and $u_{j,k} \in H^1(\Omega)$ by 
\begin{equation}
    u_{j,k} \left(x', x_d\right) := \sin\left(M^{j \gamma} \left(x_d - c(j)\right)\right) 1_{\Omega_{j,k}} \left(x', x_d\right) . 
\end{equation}
Note that for fixed $j \in \N$, the interior of the support the $\{u_{j,k}\}^{M^{(d-1)j}}_{k = 1}$ are disjoint. One can show that for 
\begin{equation}
    \lambda(j) := M^{2 \gamma j \left(1 + \frac{1}{d} (- 1 + \epsilon)\right)}
\end{equation}
we have for every $k \in \{1, \dots, M^{(d-1)j}\}$
\begin{equation}
    \int_\Omega \left|\nabla u_{j,k}\right|^2 + \int_\Omega \lambda V \left|u_{j,k}\right|^2 < 0 .
\end{equation}
Hence, for every $j \in \N$
\begin{equation}
    N \left(- \Delta^N_\Omega + \lambda (j) V\right) \geq M^{(d-1) j}.
\end{equation}
A computation shows that since $\epsilon < (d-1)(1/\gamma- 1)$,
\begin{equation}
    \lim_{j \rightarrow \infty} \lambda(j)^{-\frac{d}{2}} M^{(d-1) j} = \infty .
\end{equation}
It follows that
\begin{equation}
    \limsup_{\lambda \rightarrow \infty} \lambda^{- \frac{d}{2}} N \left(- \Delta^N_\Omega + \lambda V\right) = \infty .
\end{equation}

\subsection{Details of the proof of Theorem \ref{th:example}}\label{ss:exdet}
Let us now come to the details of the proof of Theorem \ref{th:example}. For the definition of $\Omega$ and of the orthogonal set of test functions, we closely follow \cite[Theorem 1.10]{netrusov2005weyl}.

\begin{definition}\label{def:exbasic}
Let $d \geq 2$, $\gamma\in\left( \tfrac{d-1}{d},1\right)$ and let $m \in \N$ be large enough such that $m\ga\ge 1$ and $m(1-\ga)\ge 4$.
% $\tfrac{1}{2^{m\left(1-\gamma\right)} - 1} \leq \frac{1}{8}$. 
Define $Q^{(d-1)} := \left(0,1\right)^{d-1}$ and
\begin{align*}
    \psi : \R^{d-1} &\rightarrow \left[0,1/2\right] \, ,\ 
    x' \mapsto \frac{1}{2} - \Big|x' - \left(\frac{1}{2}, \dots, \frac{1}{2}\right)\Big|_\infty 1_{Q^{(d-1)}}(x').
\end{align*}
% \begin{align*}
%     \psi : \R^{d-1} &\rightarrow \left[0,\frac{1}{2}\right] \, ,\ 
%     x' \mapsto \begin{cases} \frac{1}{2} - \bigm|x' - \left(\frac{1}{2}, \dots, \frac{1}{2}\right)\bigm|_\infty &\mathrm{,\ if\ } x' \in Q^{(d-1)} \\ 0 & \mathrm{else.} \end{cases}
% \end{align*}
% \begin{align*}
%     \psi : \R^{d-1} &\rightarrow \left[0,\tfrac{1}{2}\right] \\
%     x' &\mapsto \begin{cases} \frac{1}{2} - \bigm|x' - \left(\frac{1}{2}, \dots, \frac{1}{2}\right)\bigm|_\infty &\mathrm{,\ if\ } x' \in Q^{(d-1)} \\ 0 & \mathrm{else.} \end{cases}
% \end{align*}
For every $j\in \N_0$ let
\begin{equation}
    K_j := \left\{0,1,2,3,\dots,2^{jm}-1\right\}^{d-1}
\end{equation}
and for $k \in K_j$ define
\begin{equation}
    Q(j,k) := \left\{ x'\in \R^{d-1} \Bigm| 2^{jm} x' - k \in Q^{(d-1)}\right\} \subset Q^{(d-1)}.
\end{equation}
Define for $j \in \N_0$ the functions
\begin{equation}
    g_j : Q^{(d-1)} \rightarrow \left[ 0,1/2\right]\,,\ x' \mapsto \sum_{k \in K_j} \psi\left(2^{jm}x'-k\right),
\end{equation}
for $n\in\N_0\cup\{-1\}$
\begin{equation}
    f_n : Q^{(d-1)} \rightarrow \left[0,\infty\right) \,,\
    x' \mapsto \sum_{j=0}^n 2^{-{\gamma}jm}g_j(x'). 
\end{equation}
We also denote $f_{-1} \equiv 0$ and $f=\lim_{n\to \infty}f_n$. 
%\begin{equation}
%    f : Q^{(d-1)} \rightarrow \left[0,\infty\right) \,,\
%    x' \mapsto \sum_{j=0}^\infty 2^{-{\gamma}jm}g_j(x').
%\end{equation}
% \begin{align*}
%     g_j : Q^{(d-1)} &\rightarrow \left[ 0,\tfrac{1}{2}\right]  \\
%     x' &\mapsto \sum_{k \in K_j} \psi\left(2^{jm}x'-k\right),
% \end{align*}
% for $n\in\N_0\cup\{-1\}$
% \begin{align*}
%     f_n : Q^{(d-1)} &\rightarrow \left[0,\infty\right) \\
%     x' &\mapsto \sum_{j=0}^n 2^{-{\gamma}jm}g_j(x') \\
% \end{align*}
% and
% \begin{align*}
%     f : Q^{(d-1)} &\rightarrow \left[0,\infty\right) \\
%     x' &\mapsto \sum_{j=0}^\infty 2^{-{\gamma}jm}g_j(x').
% \end{align*}
Also define for every $n\in \N_0$, $k\in K_n$ 
\begin{equation}\label{eq:ankdef}
a_{n,k} := \underset{x'\in Q\left(n,k\right)}{\sup} f_{n-1}(x').
\end{equation}
\end{definition}
\begin{lemma}\label{le:fcont}
\begin{enumerate}[(i)]
% \item For every $j \in \N_0$ the cubes $\left\{ Q(j,k)\right\}_{k\in K_j}$ are disjoint and
% \begin{equation}
%     \overline{Q^{(d-1)}} = \bigcup_{k \in K_j} \overline{Q(j,k)}.
% \end{equation}
% \item $\psi : \left( \R^{d-1}, \noi{\cdot} \right) \rightarrow \left[0, \tfrac{1}{2} \right]$ is Lipschitz continuous with constant $1$.
% \item For every $j \in \N_0$, $g_j : \left( Q^{(d-1)}, \noi{\cdot}\right) \rightarrow \left[0, \tfrac{1}{2} \right]$ is well-defined and Lipschitz continuous with constant $2^{jm}$.
% \item $f$ is well-defined and for every $n \in \N_0$ and $x' \in Q^{(d-1)}$, we have
% \begin{equation}
%     0 \leq f_n(x') \leq f_{n+1}(x') \leq f(x') \leq 1.
% \end{equation}
\item $f : \left(Q^{(d-1)}, \noi{\cdot}\right) \rightarrow \left[ 0, \infty \right)$ is $\gamma$-Hölder continuous with constant $3$.
\item For every $n \in \N_0$, $k \in K_n$ and $x', y' \in Q(n,k)$, we have 
\begin{equation}
    |f_{n-1}(x') - f_{n-1}(y')| \leq \frac{1}{8}2^{-\gamma m n }.
\end{equation}
\item For every $n \in \N_0$, $k\in K_n$ and $x' \in Q(n,k)$ with $2^{n m} x' - k \in \left[ 1/4,3/4\right]^{d-1}$,
% \begin{equation}
% 2^{n m} x' - k \in \left[ \frac{1}{4}, \frac{3}{4}\right]^{d-1}
% \end{equation}
we have
\begin{equation}
    f(x') - a_{n,k} \geq \frac{1}{8} 2^{- \gamma m n}.
\end{equation}
\item For every $n \in \N_0$, $k \in K_n$ and $x' \in Q(n,k)$, we have
\begin{equation}
    f(x') - a_{n,k} \leq 2^{- \gamma m n} .
\end{equation}
\end{enumerate}
\end{lemma}
\begin{proof}
\textbf{Proof of (i). } Let $x'$, $y' \in Q^{(d-1)}$ with $x' \neq y'$ and denote by $n'$ the largest number in $\N_0$ such that $2^{-n' m } \geq \noi{x' - y'}$. In particular, 
\begin{equation}\label{eq:xydiffmn}
    2^{- (n' + 1) m} < \noi{x' - y'} \leq 2^{-n' m}.
\end{equation}
We have
\begin{align*}
    |f(x') - f(y')| &= \left\lvert\sum_{j=0}^\infty 2^{- \gamma m j} g_j (x') - \sum_{j=0}^\infty 2^{- \gamma m j}g_j(y')\right\rvert \\
    &\leq \sum_{j=0}^{n'} 2^{- \gamma m j} |g_j(x') - g_j(y')| + \sum_{j=n'+1}^\infty 2^{- \gamma m j} |g_j(x') - g_j(y')| 
    \end{align*}
For the first term, we use the Lipschitz continuity of $g_j$, \eqref{eq:xydiffmn} and  $m(1-\ga)\ge 4$ by Definition \ref{def:exbasic} to get
\begin{align*}
 & \sum_{j=0}^{n'} 2^{- \gamma m j} |g_j(x') - g_j(y')| \le \sum_{j=0}^{n'} 2^{- \gamma m j} \cdot 2^{j m} \noi{x' - y'}\\
& \le \noi{x'-y'}^\gamma 2^{- m(1-\gamma)  n'}\sum_{j=0}^{n'}2^{ m(1-\gamma)  j}\le \noi{x'-y'}^\gamma \sum_{j=0}^{n'}2^{- m(1-\gamma)  j} 
    %& \le \noi{x'-y'}^\gamma\frac{1}{1-2^{-(1-\gamma)m}}\
    \le 2\noi{x'-y'}^\gamma\,.
\end{align*}
For the second term, we use the Lipschitz continuity of $g_j$, \eqref{eq:xydiffmn} and  $m\ga\ge 1$  to get
\begin{align*}
\sum_{j=n'+1}^\infty 2^{- \gamma m j} |g_j(x') - g_j(y')|\le \ha\sum_{j=n'+1}^\infty 2^{- \gamma m j} =\ha 2^{- \gamma m (n'+1)} \sum_{j=0}^\infty 2^{- \gamma m j} \le \noi{x'-y'}^\gamma\,.
\end{align*}
Combining these two estimates, we obtain the  claim.
\\ 
\textbf{Proof of (ii). } 
Let $n \in \N_0$, $k\in K_n$ and $x',y' \in Q\left(n,k\right)$. We have
\begin{align*}
    |f_{n-1}(x')-f_{n-1}(y')| = &\left|\sum_{j=0}^{n-1} 2^{-\gamma m j} g_j(x') - \sum_{j=0}^{n-1}2^{-\gamma m j}g_j(y')\right| \leq \sum_{j=0}^{n-1} 2^{-\gamma m j}|g_j(x') - g_j(y')| \\ 
    &\leq \sum_{j=0}^{n-1} 2^{-\gamma m j} \cdot 2^{m j} \noi{x' - y'} \leq 2^{-mn}\sum_{j=0}^{n-1} 2^{(1-\gamma) m j}\le\frac{1}{8}2^{-\gamma m n}\,,
    % \sum_{j=0}^{n-1} 2^{-\gamma m j} \cdot 2^{m j} \cdot 2^{-mn} \\
    % &= 2^{- m n}\sum_{j=0}^{n-1} 2^{(1-\gamma) m j} = 2^{- m n}\frac{1 - 2^{(1-\gamma) m n}}{1 - 2^{(1-\gamma) m}} \\
    % &= 2^{- m n}\frac{2^{(1-\gamma) m n} - 1}{2^{(1-\gamma) m} - 1} \leq 2^{- m n}\frac{2^{(1-\gamma) m n}}{2^{(1-\gamma) m} - 1} = 2^{- \gamma m n}\frac{1}{2^{(1-\gamma) m} - 1}.
\end{align*}
where we used Lipschitz continuity of $g_j$ in the third step, and $m(1-\ga)\ge 4$ in the last step. \\
% \begin{align*}
%     |f_{n-1}(x')-f_{n-1}(y')| = &\left|\sum_{j=o}^{n-1} 2^{-\gamma m j} g_j(x') - \sum_{j=0}^{n-1}2^{-\gamma m j}g_j(y')\right| \leq \sum_{j=0}^{n-1} 2^{-\gamma m j}|g_j(x') - g_j(y')| \\ 
%     &\leq \sum_{j=0}^{n-1} 2^{-\gamma m j} \cdot 2^{m j} \noi{x' - y'} \leq \sum_{j=0}^{n-1} 2^{-\gamma m j} \cdot 2^{m j} \cdot 2^{-mn} \\
%     &= 2^{- m n}\sum_{j=0}^{n-1} 2^{(1-\gamma) m j} = 2^{- m n}\frac{1 - 2^{(1-\gamma) m n}}{1 - 2^{(1-\gamma) m}} \\
%     &= 2^{- m n}\frac{2^{(1-\gamma) m n} - 1}{2^{(1-\gamma) m} - 1} \leq 2^{- m n}\frac{2^{(1-\gamma) m n}}{2^{(1-\gamma) m} - 1} = 2^{- \gamma m n}\frac{1}{2^{(1-\gamma) m} - 1}.
% \end{align*}
% Here we used \textbf{(ii)} in the third step. \\

\textbf{Proof of (iii). }
First note that for all $y' \in \left[1/4,3/4 \right]^{d-1}$, we have  $\psi(y')\ge 1/4$. 
%\begin{equation}\label{eq:psiv}
%    \psi(y') = \frac{1}{2} - \noi{y' - \left(\frac{1}{2}, \dots, \frac{1}{2}\right)} \geq \frac{1}{2} - \frac{1}{4} = \frac{1}{4}.
%\end{equation}
Let $n \in \N_0$, $k \in K_n$ and $x' \in Q(n,k)$ with
$    2^{n m} x' - k \in \left[1/4,3/4\right]^{d-1}$. Then, 
\begin{equation}\label{eq:gnv}
    g_n(x') = \sum_{\tilde{k} \in K_n} \psi\left(2^{n m} x' - \tilde{k}\right) = \psi\left(2^{n m} x' - k\right) \geq \frac{1}{4}.
\end{equation}
% Since by definition $m$ is chosen large enough such that
% \begin{equation}
%     \frac{1}{2^{m(1-\gamma)} - 1} \leq \frac{1}{8},
% \end{equation}
% we have
% \begin{equation}\label{eq:gn}
%     g_n(x') - \frac{1}{2^{m(1-\gamma)} - 1} \geq \frac{1}{4} - \frac{1}{8} = \frac{1}{8}
% \end{equation}
% Recall that $a_{n,k}$ was defined by
% \begin{equation}
%     a_{n,k} = \underset{y' \in Q(n,k)}{\sup} f_{n-1}(y').
% \end{equation}
Therefore, we have
\begin{align*}
    f(x') - a_{n,k} &\geq f_n(x') - a_{n,k} = f_n (x') - f_{n-1}(x') + f_{n-1}(x') - a_{n,k} \\
    &\geq 2^{- \gamma m n} g_n(x') - |f_{n-1}(x') - a_{n,k}| \geq \frac{1}{4}2^{-\gamma m n}  - \frac{1}{8}2^{-\gamma m n}  = \frac{1}{8} 2^{- \gamma m n},
\end{align*}
where we used \eqref{eq:gnv}, (ii) and \eqref{eq:ankdef} in the fourth step. \\
% Therefore, we have
% \begin{align*}
%     f(x') - a_{n,k} &\geq f_n(x') - a_{n,k} = f_n (x') - f_{n-1}(x') + f_{n-1}(x') - a_{n,k} \\
%     &\geq 2^{- \gamma m n} g_n(x') - |f_{n-1}(x') - a_{n,k}| \geq 2^{-\gamma m n} g_n (x') - 2^{-\gamma m n} \frac{1}{2^{(1-\gamma) m} - 1} \\
%     &= 2^{- \gamma m n} \left(g_n(x') - \frac{1}{2^{(1-\gamma) m} - 1}\right) \geq 2^{- \gamma m n} \frac{1}{8} = \frac{1}{8} 2^{- \gamma m n},
% \end{align*}
% where we used \textbf{(vi)} in the fourth step and \eqref{eq:gn} in the second-to-last step. \\

\textbf{Proof of (iv). }
Let $n \in \N_0$, $k\in K_n$, $x' \in Q(n,k)$. Then by \eqref{eq:ankdef}, $a_{n,k}\geq f_{n-1}(x')$.
% \begin{equation}
%     a_{n,k} = \underset{y' \in Q(n,k)}{\sup} f_{n-1}(y') \geq f_{n-1}(x').
% \end{equation}
By $g_j(x') \leq \tfrac{1}{2}$ and $m\ga\ge 1$, it follows that
\begin{align*}
    f(x') - a_{n,k} \leq f(x') - f_{n-1}(x') = \sum_{j = n}^\infty 2^{- \gamma m j} g_j (x') \leq  \frac{1}{2} 2^{- \gamma m n} \sum_{j = 0}^\infty 2^{-\gamma m j} \le 2^{- \gamma m n}.
\end{align*}
% where we used $g_j(x') \leq \tfrac{1}{2}$ in the third step, and $m\ga\ge 1$ in the last step.
% \begin{align*}
%     f(x') - a_{n,k} \geq &f(x') - f_{n-1}(x') = \sum_{j = n}^\infty 2^{- \gamma m j} g_j (x') \leq \sum_{j=n}^\infty 2^{- \gamma m j} \frac{1}{2} \\
%     &= \frac{1}{2} 2^{- \gamma m n} \sum_{j = 0}^\infty 2^{-\gamma m j} = \frac{1}{2} 2^{- \gamma m n} \frac{1}{1 - 2^{-\gamma m}}, 
% \end{align*}
% where we used $g_j(x') \leq \tfrac{1}{2}$, see \textbf{(iii)} in the third step.
\end{proof}
\begin{definition}[$\Omega$ and $\Omega_{n,k}$]\label{def:omexdef}
Define the $\gamma$-Hölder domain
\begin{equation*}
    \Omega := \left\{x = (x', x_d) \in \mathbb{R}^{d-1} \times \mathbb{R} \bigm| x' \in Q^{d-1},\ 0 \leq x_d < f(x') \right\} \cup\left( (-2,2)^{d-1} \times (-2, 0)\right)
\end{equation*}
and for all $n \in \N_0$, $k \in K_n$ define
\begin{equation}
    \Omega_{n,k} := \left\{x \in \Omega \bigm| x' \in Q(n,k),\ x_d \in \left(f_{n-1}(x') , f(x')\right)\right\}.
\end{equation}
\end{definition}

% \begin{lemma}
% Let $\Omega$ and $\Omega_{n,k}$ be as defined in Definition \ref{def:omexdef}. Let $n \in \N_0$ and  $k \in K_n$. Then
% $\Omega$ is an open bounded connected $\gamma$-Hölder continuous domain. Moreover, 
% $\Omega_{n,k}$ is open.
% \end{lemma}
% \begin{proof}
% % \textbf{Proof of (i). }
% (i) follows directly from the definition of $\Omega$ combined with the boundedness and $\gamma$-Hölder continuity of $f$, see Lemma \ref{le:fcont}(v).  For (ii), note that the openness of $\Omega_{n,k}$ directly follows from the continuity of $f_{n-1}$, see Lemma \ref{le:fcont}(iii), and the continuity of $f$, see Lemma \ref{le:fcont}(v).
% % \textbf{Proof of (i). }
% % This follows directly from the definition of $\Omega$ combined with the boundedness and $\gamma$-Hölder continuity of $f$, see Lemma \ref{le:fcont}(v).  \\ \\
% % \textbf{Proof of (ii). }
% % The openness of $\Omega_{n,k}$ directly follows from the continuity of $f_{n-1}$, see Lemma \ref{le:fcont}(iii), and the continuity of $f$, see Lemma \ref{le:fcont}(v).
% % \textbf{Proof of (i). }
% % The openness, boundedness and connectedness of $\Omega$ follows directly from the definition of $\Omega$ combined with the boundedness and continuity of $f$, see Lemma \ref{le:fcont}(v). The Hölder continuity of $\partial\Omega$ follows from the $\gamma$-Hölder continuity of $f$. \\
% % \textbf{Proof of (ii). }
% % The openness of $\Omega_{n,k}$ directly follows from the continuity of $f_{n-1}$, see Lemma \ref{le:fcont}(iii), and the continuity of $f$, see Lemma \ref{le:fcont}(v).
% \end{proof}
\begin{definition}[$u_{n,k}, b_2, b_\nabla,b_V$ and $V$]\label{de:unkbV}
\begin{enumerate}[(i)]
\item 
% Let $n\in \N_0$, $k \in K_n$. Define $u_{n,k} : \Omega \rightarrow \mathbb{R}$ by
For $n\in \N_0$, $k \in K_n$, let $u_{n,k} : \Omega \rightarrow \mathbb{R}$ with
\begin{equation}
    u_{n,k}(x) := \begin{cases} \sin\left(2^{\gamma m n}\left(x_d - a_{n,k}\right)\right) &\mathrm{\ for\ } x' \in Q(n,k),\, x_d \geq a_{n,k} \\ 0 &\mathrm{\ else}.\end{cases}
\end{equation}
\item
Define
\begin{equation}\label{eq:b2Vnabla}
    b_2 := 2^{-(d-1)} \int^{\frac{1}{8}}_0 \dr{t} |\sin(t)|^2 ,\quad b_\nabla := \int^1_0 \dr{t} |\cos(t)|^2 \quad \textrm{and}\quad b_V := 2 \frac{b_\nabla}{b_2}  .
\end{equation}
% \begin{align}
% &\ \ b_2 := 2^{-(d-1)} \int^{\frac{1}{8}}_0 \dr{t} |\sin(t)|^2 , \\
% &\ \ b_\nabla := \int^1_0 \dr{t} |\cos(t)|^2 ,
% \end{align}
% \begin{align}
% &\ \ b_2 := 2^{-(d-1)} \int^{\frac{1}{8}}_0 \dr{t} |\sin(t)|^2 , \\
% &\ \ b_\nabla := \int^{\frac{1}{2} \cdot \frac{1}{1 - 2^{-\gamma m}}}_0 \dr{t} |\cos(t)|^2 ,
% \end{align}
% and
% \begin{equation}
% b_V := 2 \frac{b_\nabla}{b_2}  .
% \end{equation}
% \begin{equation}
% b_V := 2 \frac{b_\nabla}{b_2} \left(\frac{1}{2} \cdot \frac{1}{1 - 2^{-\gamma m}}\right)^{- \frac{2}{d} \left(-1 + \epsilon\right)} .
% \end{equation}
\item
Let $\epsilon\in(0,1)$. Define $V: \Omega \rightarrow \mathbb{R}$ depending on $\epsilon$ by
\begin{equation}\label{eq:Vdef}
V(x) := \begin{cases}-b_V \left(f(x) - x_d\right)^{\frac{2}{d}(-1 + \epsilon)} &\mathrm{\ for\ } x' \in Q^{(d-1)},\, 0 \leq x_d < f(x') \\ 0 &\ \mathrm{else}.\end{cases}
\end{equation}
\end{enumerate}
\end{definition}

\begin{lemma}[Estimates for $u_{n,k}$]\label{le:gradunk}
Let $n \in \N_0$ and $k \in K_n$. Let $\epsilon \in \left(0,1\right)$. 
\begin{enumerate}[(i)]
\item Then
$u_{n,k} \in H^1(\Omega)$. Moreover,
\begin{align*}
    \int_\Omega |u_{n,k}|^2 \geq b_2 2^{-(d-1) m n} \cdot 2^{- \gamma m n}, \quad  \int_\Omega |\nabla u_{n,k}|^2  \leq b_\nabla 2^{-(d-1) m n} \cdot 2^{\gamma m n}.
\end{align*}
\item
For all $\lambda > 0$, we have
\begin{equation}
    \int_\Omega |\nabla u_{n,k}|^2 + \int_\Omega \lambda V |u_{n,k}|^2 \leq b_\nabla 2^{- (d-1) m n} \cdot 2^{\gamma m n} \left(1 - 2 \lambda 2^{- 2 \gamma m n} \cdot 2^{- \gamma m n \frac{2}{d}(-1 + \epsilon)}\right).
\end{equation}
\end{enumerate}
\end{lemma}
\begin{proof}
{\bf Proof of (i). }
A direct computation shows that $u_{n,k} \in H^1(\Omega)$.
% First note that $u_{n,k}$ is bounded, so $u_{n,k} \in L^2(\Omega)$. Since $\sin(0) = 0$, we also know that $u_{n,k}$ is continuous. Note that $u_{n,k}$ is almost everywhere differentiable with
% \begin{equation}\label{eq:gradunk}
%     \nabla u_{n,k}(x) = \begin{cases}
%         2^{\gamma m n} \cos\left(2^{\gamma m n}(x_d - a_{n,k})\right) &\mathrm{for\ } x'\in Q(n,k),\, x_d > a_{n,k} \\
%         0 &\mathrm{for\ } x\in \Omega \setminus \left(Q(n,k) \times \left[a_{n,k}, \infty\right)\right)
%     \end{cases}
% \end{equation}
% Using the definition of weak differentiability, integration by parts on both \newline $\Omega\cap\left(Q(n,k) \times \left(a_{n,k}, \infty\right)\right)$ and $\Omega \setminus \left(Q(n,k) \times \left[a_{n,k}, \infty\right)\right)$, and the continuity of $u_{n,k}$, one can deduce that $u_{n,k}$ is weakly differentiable and its weak derivative is given by \eqref{eq:gradunk}. The boundedness of $\nabla {u_{n,k}}$ implies $\nabla u_{n,k} \in L^2(\Omega)$, so $u_{n,k} \in H^1(\Omega)$. 
% Recall that by Lemma \ref{le:fcont}(vii), we know for all $x' \in Q(n,k)$ with
% \begin{equation}
%     2^{m n} x' - k \in \left[\frac{1}{4}, \frac{3}{4}\right]^{d-1}
% \end{equation}
% that
% \begin{equation}
%     f(x') - a_{n,k} \geq \frac{1}{8} 2^{- \gamma m n}.
% \end{equation}
Next, we compute
\begin{align*}
    \int_\Omega |u_{n,k}|^2& =  \int_{Q(n,k)}\dr{x'} \int_0^{f(x')-a_{n,k}} \dr{s} |\sin\left(2^{\gamma m n} s\right)|^2 \\
    &\geq \int_{\left\{x' \in Q(n,k) \bigm| 2^{m n} x' - k \in \left[1/4,3/4\right]^{d-1}\right\}} \dr{x'} \int_0^{\frac{1}{8} 2^{- \gamma m n}}\dr{s} |\sin(2^{\gamma m n} s)|^2 \\
    &= 2^{- (d-1) m n} \cdot 2^{- (d-1)} \int_0^{\frac{1}{8}}\dr{t}|\sin(t)|^2 2^{- \gamma m n} =  b_2 2^{-(d-1) m n} \cdot 2^{- \gamma m n}  ,
\end{align*}
where we used Lemma \ref{le:fcont}(iii) in the third step. 
% \begin{align*}
%     \int_\Omega |u_{n,k}|^2 = &\int_{Q(n,k)}\dr{x'} \int_{\left(a_{n,k}, f(x')\right)}\dr{x_d} |\sin\left(2^{\gamma m n} (x_d - a_{n,k})\right)|^2 \\
%     &= \int_{Q(n,k)}\dr{x'} \int_0^{f(x')-a_{n,k}} \dr{s} |\sin\left(2^{\gamma m n} s\right)|^2 \\
%     &\geq \int_{\left\{x' \in Q(n,k) \bigm| 2^{m n} x' - k \in \left[\frac{1}{4}, \frac{3}{4}\right]^{d-1}\right\}} \dr{x'} \int_0^{\frac{1}{8} 2^{- \gamma m n}}\dr{s} |\sin(2^{\gamma m n} s)|^2 \\
%     &= 2^{- (d-1) m n} \cdot 2^{- (d-1)} \int_0^{\frac{1}{8}}\dr{t}|\sin(t)|^2 2^{- \gamma m n} = 2^{-(d-1) m n} \cdot 2^{- \gamma m n} \cdot b_2 ,
% \end{align*}
% where we used the change of variables $t = 2^{\gamma m n} s$. \\
% Recall that by Lemma \ref{le:fcont}(viii) for all $x' \in Q(n,k)$,
% \begin{equation}
%     f(x') - a_{n,k} \leq \frac{1}{2} 2^{-\gamma m n} .
% \end{equation}
% \begin{equation}
%     f(x') - a_{n,k} \leq \frac{1}{2} 2^{-\gamma m n} \frac{1}{1 - 2^{- \gamma m}}.
% \end{equation}
Moreover, by Lemma \ref{le:fcont} (iv),
% and \eqref{eq:gradunk},  
\begin{align*}
    &\qquad \int_\Omega |\nabla {u_{n,k}}|^2 = \int_{Q(n,k)}\dr{x'} \int_0^{f(x')-a_{n,k}}\dr{s} 2^{2\gamma m n} |\cos\left(2^{\gamma m n} s\right)|^2 \\
    &\leq \int_{Q(n,k)}\dr{x'}\int_0^{ 2^{- \gamma m n}} \dr{s} 2^{2\gamma m n} |\cos \left(2^{\gamma m n} s\right)|^2 = b_\nabla 2^{- (d-1) m n} \cdot 2^{\gamma m n},
\end{align*}
% \begin{align*}
%     \int_\Omega |\nabla {u_{n,k}}|^2 = &\int_{Q(n,k)}\dr{x'}\int_{a_{n,k}}^{f(x')}\dr{x_d}|2^{\gamma m n} \cos\left(2^{\gamma m n}\left(x_d - a_{n,k}\right)\right)|^2 \\
%     &= \int_{Q(n,k)}\dr{x'} \int_0^{f(x')-a_{n,k}}\dr{s} 2^{2\gamma m n} |\cos\left(2^{\gamma m n} s\right)|^2 \\
%     &\leq \int_{Q(n,k)}\dr{x'}\int_0^{\frac{1}{2} 2^{- \gamma m n} \frac{1}{1 - 2^{- \gamma m}}} \dr{s} 2^{2\gamma m n} |\cos \left(2^{\gamma m n} s\right)|^2 \\
%     &= 2^{-(d-1) m n} \cdot 2^{\gamma m n} \int_0^{\frac{1}{2} \cdot \frac{1}{1 - 2^{-\gamma m}}} \dr{t} |\cos(t)|^2 = b_\nabla 2^{- (d-1) m n} \cdot 2^{\gamma m n},
% \end{align*}
where we used the change of variables $t = 2^{\gamma m n} s$. \\ \\
{\bf Proof of (ii). }
Let $\lambda > 0$. For all $x \in \Omega(n,k)$, we have $    x_d \in \left(f_{n-1}(x'), f(x')\right)$, so by $0 \leq g_j \leq \tfrac{1}{2}$ for all $j \in \N_0$, we get as in the proof of Lemma \ref{le:fcont} (iv)
\begin{equation}\label{eq:fxdest}
     f(x') - x_d \leq f(x') - f_{n-1}(x') = \sum_{j = n}^\infty 2^{- \gamma m j} g_j(x')\le 2^{- \gamma m n}.
\end{equation}
% \begin{align*}
%     f(x') - x_d &\leq f(x') - f_{n-1}(x') = \sum_{j = n}^\infty 2^{- \gamma m j} g_j(x') \leq \frac{1}{2} 2^{- \gamma m n}\sum_{j = 0}^\infty 2^{- \gamma m j} \\
%     &= \frac{1}{2} 2^{- \gamma m n} \frac{1}{1 - 2^{- \gamma m}}
% \end{align*}
For all $x\in \supp (V |u_{n,k}|^2) \subset \Omega_{n,k}$, we have by \eqref{eq:fxdest} and $\epsilon < 1$
\begin{equation}
        |V(x)| = b_V \left(f(x') - x_d\right)^{\frac{2}{d} (-1 + \epsilon)} \geq b_V  2^{- \gamma m n\frac{2}{d} (-1 + \epsilon)} .
\end{equation}
% \begin{align*}
%     |V(x)| &= b_V \left(f(x') - x_d\right)^{\frac{2}{d} (-1 + \epsilon)} \geq b_V \left( 2^{- \gamma m n}\right)^{\frac{2}{d} (-1 + \epsilon)} \\
%     &= b_V \left(\frac{1}{2}\cdot \frac{1}{1 - 2^{- \gamma m}} \right)^{\frac{2}{d} (-1 + \epsilon)} 2^{- \gamma m n\frac{2}{d} (-1 + \epsilon)} .
% \end{align*}
Using (ii) and (iii), we get by \eqref{eq:b2Vnabla}
\begin{align*}
    &\int_\Omega |\nabla u_{n,k}|^2 + \int_\Omega \lambda V |u_{n,k}|^2 \leq b_\nabla 2^{- (d-1) m n} \cdot 2^{\gamma m n} - \lambda b_V 2^{- \gamma m n \frac{2}{d} (-1 + \epsilon)}b_2 2^{-(d-1) m n} 2^{-  \gamma m n}\\
    &= b_\nabla 2^{-(d-1) m n} \cdot 2^{\gamma m n} \left(1 - 2\lambda 2^{- \gamma m n \frac{2}{d} (-1 + \epsilon)} \cdot 2^{- 2 \gamma m n}\right). 
\end{align*}
% \begin{align*}
%     |V(x)| &= b_V \left(f(x') - x_d\right)^{\frac{2}{d} (-1 + \epsilon)} \geq b_V \left(\frac{1}{2}\cdot \frac{1}{1 - 2^{- \gamma m}} 2^{- \gamma m n}\right)^{\frac{2}{d} (-1 + \epsilon)} \\
%     &= b_V \left(\frac{1}{2}\cdot \frac{1}{1 - 2^{- \gamma m}} \right)^{\frac{2}{d} (-1 + \epsilon)} 2^{- \gamma m n\frac{2}{d} (-1 + \epsilon)} .
% \end{align*}
% Using \textbf{(ii)} and \textbf{(iii)}, we get
% \begin{align*}
%     &\int_\Omega |\nabla u_{n,k}|^2 + \int_\Omega \lambda^2 V |u_{n,k}|^2 \\
%     &\leq b_\nabla 2^{- (d-1) m n} \cdot 2^{\gamma m n} - \lambda^2 b_V \left(\frac{1}{2} \cdot \frac{1}{1 - 2^{- \gamma m}}\right)^{\frac{2}{d} (-1 + \epsilon)}2^{- \gamma m n \frac{2}{d} (-1 + \epsilon)}b_2 2^{-(d-1) m n} \\
%     &= b_\nabla 2^{-(d-1) m n} \cdot 2^{\gamma m n} \left(1 - 2\lambda^2 2^{- \gamma m n \frac{2}{d} (-1 + \epsilon)} \cdot 2^{- 2 \gamma m n}\right),
% \end{align*}
% where we used the definition of $b_\nabla$, see Definition \ref{de:unkbV}(ii) in the last step.
\end{proof}

\begin{remark}\label{re:selfadjointev}
Lemma \ref{le:gradunk}(ii) will be the starting point for the example that satisfies \eqref{eq:thexeq} we are looking for in this subsection. If we choose for $n \in \N_0$
\begin{equation}\label{eq:lambdachoice}
    \lambda := 2^{2\gamma m n} \cdot 2^{\gamma m n \tfrac{2}{d} \left(-1 + \epsilon \right)}
\end{equation}
in Lemma \ref{le:gradunk}(ii), then for all $k \in K_n$, we have 
\begin{align*}
    \int_\Omega |\nabla u_{n,k}|^2 + \int_\Omega \lambda V |u_{n,k}|^2 \leq b_\nabla 2^{-(d-1) m n}\cdot 2^{\gamma m n}\left(1- 2 \lambda 2^{- \gamma m n \tfrac{2}{d}(-1 + \epsilon)} \cdot 2^{-2 \gamma m n}\right) <0
%    &= -b_\nabla 2^{-(d-1) m n} \cdot 2^{\gamma m n}<0 .
\end{align*}

Now suppose $- \Delta^N_\Omega + \lambda V$ was a self-adjoint operator with quadratic form domain $H^1(\Omega)$. In fact, this will be shown under suitable assumptions in Lemma \ref{le:selfadjointev}. Then, we deduce by the min-max principle \cite[Theorem 12.1, version 3]{liebloss} and since the $\left\{u_{n,k}\right\}_{k \in K_n}$ have disjoint interior of their support that
\begin{equation}
    \evv{\Omega}{\lambda V} \geq |K_n| = 2^{(d-1) m n} .
\end{equation}
If $\epsilon = (d-1)\left(1/\gamma-1\right)$, then $|K_n| =\lambda^{\frac d 2}$. 
But we can apply Lemma \ref{le:gradunk} with $0<\epsilon < (d-1) (1/\gamma-1)<1$ and $\lambda$ as in \eqref{eq:lambdachoice}, then we get  
\begin{equation}
   \lambda^{-\frac{d}{2}} \evv{\Omega}{\lambda V} \geq \lambda^{-\frac{d}{2}} |K_n|  = \lambda^{-\frac{d}{2}} 2^{(d-1) m n} \rightarrow \infty \mathrm{\ as\ } n \rightarrow \infty. 
\end{equation}
% Now if $\gamma \in \left(\tfrac{d-1}{d}, 1\right)$ and $\epsilon < (d-1) \left(1/\gamma- 1\right)$, by \eqref{eq:dgammaplusgamma}, we have
% \begin{equation}
%     d \left(\gamma m n + \gamma m n \frac{1}{d}(-1+\epsilon)\right) < (d-1) m n ,
% \end{equation}
% so
% \begin{equation}
%     \frac{\evv{\Omega}{\lambda^2 V}}{\lambda^d} \geq \frac{|K_n|}{\lambda^d} = \frac{2^{(d-1) m n}}{\lambda^d} \rightarrow \infty \mathrm{\ as\ } n \rightarrow \infty
% \end{equation}
Since $\lambda \rightarrow \infty$ as $n \rightarrow \infty$, we have shown \eqref{eq:thexeq}. % that 
%if $\gamma \in \left(\tfrac{d-1}{d}, 1\right)$, $0  <\epsilon < (d-1) \left(1/\gamma- 1\right)$ and the parameters are chosen in such a way that $-\Delta^N_\Omega + \lambda V$ is self-adjoint with quadratic form domain $H^1(\Omega)$, then
%\begin{equation}
%    \limsup_{\lambda \rightarrow \infty} \lambda^{-\frac{d}{2}} \evv{\Omega}{\lambda V}= \infty .
%\end{equation}
% \begin{equation}
%     \limsup_{\lambda \rightarrow \infty} \frac{\evv{\Omega}{\lambda^2 V}}{\lambda^d} = \infty .
% \end{equation}
\end{remark}
\begin{lemma}[Self-adjointness of $- \Delta^N_\Omega + \lambda V$ and $\no{V}{\tilde{p}, \beta}=\infty$]\label{le:selfadjointev}
Let $\gamma \in \left(\tfrac{d-1}{d}, 1\right)$.
\begin{enumerate}[(i)]
\item 
Then there exists $0 < \epsilon < (d-1)\left(1/\gamma- 1\right)$ such that $V \in L^{p^*}(\Omega) \subset L^{\frac{d}{2}}(\Omega)$ and for every $\lambda > 0$ the operator $    - \Delta^N_\Omega + \lambda V$
is bounded from below, has finitely many negative eigenvalues and it is self-adjoint with quadratic form domain $H^1 (\Omega)$. 
\item
For every $0 < \epsilon < (d-1)\left(1/\gamma- 1\right)$, we have $\no{V}{\tilde{p}, \beta}=\infty$.
% $V \notin L^{\hat{p}, \beta}(\Omega)$.
\end{enumerate}
\end{lemma}
% \begin{remark}
% \begin{enumerate}
% \item Lemma \ref{le:selfadjointev}(i) together with Remark \ref{re:selfadjointev} finish the proof of the existence of an example with non-semiclassical behaviour as formulated at the beginning of the subsection.
% \item Lemma \ref{le:selfadjointev}(ii) shows that this example does not contradict Theorem \ref{th:main}.
% \end{enumerate}
% \end{remark}

\begin{remark}
    Lemma \ref{le:selfadjointev}(ii) shows that this example does not contradict Theorem \ref{th:main}.
\end{remark}

\begin{proof}

\textbf{Proof of (i). }
Let us first find $0 < \epsilon < (d-1)\left(1/\gamma- 1\right)$ such that $V \in L^{p^*}(\Omega) \subset L^{\frac{d}{2}}(\Omega)$. Note that $L^{p^*}(\Omega) \subset L^{\frac{d}{2}}(\Omega)$ since $\Omega$ is bounded and $p^* > \tfrac{d}{2}$. By \eqref{eq:Vdef},
the definition of $V$ 
and the boundedness of $\Omega$, it suffices to find $0 < \epsilon < (d-1)\left(1/\gamma- 1\right)$ such that
\begin{equation} \label{eq:eps-check}
   p^* \frac{2}{d}\left(-1 + \epsilon\right) > - 1 .
\end{equation}
By a continuity argument, it suffices to show that \eqref{eq:eps-check} holds for $\eps=(d-1)\left(1/\gamma- 1\right)$, namely $\mu(\mu - (d + 1))>-d$ with $\mu := \tfrac{d-1}{\gamma} + 1\in (d,d+1)$, which is true. 

\bigskip
Fix $\epsilon = \epsilon(d,\gamma)$ as above and let $\lambda > 0$ be arbitrary. In order to show the self-adjointness of the operator $- \Delta^N_\Omega + \lambda V$, we show that the corresponding quadratic form on $H^1(\Omega)$ is well-defined, bounded from below and that it has the $H^1(\Omega)$-norm as its quadratic form norm, hence it is closed. The claim then follows from Friedrich's theorem. \\

For every $x \in \Omega$ with $x_d \geq 0$ and $\delta > 0$, let the oscillatory domain $D_x(\delta)$ and $a_x = a$ be defined as in Definition \ref{de:smalldomain}. Let $M :=  \delta /a$
% Let $a_x = a$ also be defined as in Definition \ref{de:smalldomain} and let
% \begin{equation}
%     M := M_x := \frac{\delta}{a}
% \end{equation}
as in the proof of Lemma \ref{le:bdryev}(ii). Then
\begin{align}\label{eq:Vmdeltaex}
\begin{split}
    \nor{V}{p^*, D_x(\delta)}{p^*} &\leq a^{d-1}\int_0^\delta\dr{t}\left|b_V t^{\frac{2}{d}(-1 + \epsilon)}\right|^{p^*} = M^{- (d-1)}\delta^{d-1} |b_V|^{p^*}\int_0^\delta \dr{t} t^{p^* \frac{2}{d}(-1 + \epsilon)} \\
    &=  C(d, \gamma) M^{-(d-1)} \delta^{d + p^* \frac{2}{d} (-1 + \epsilon)} \to 0 \quad \mathrm{\ as\ }\delta \rightarrow 0.  
    \end{split}
\end{align}
% \begin{align*}
%     \nor{V}{p^*, D_x(\delta)}{p^*} &\leq M^{-(d-1)} \delta^{d-1}\int_a^\delta\dr{t}\left|b_V t^{\frac{2}{d}(-1 + \epsilon)}\right|^{p^*} \\
%     &= M^{- (d-1)}|b_V|^{p^*}\delta^{d-1} \int_a^\delta \dr{t} t^{p^* \frac{2}{d}(-1 + \epsilon)} \\
%     &= C(d, \gamma) M^{-(d-1)} \delta^{d-1} \delta^{p^* \frac{2}{d} (-1 + \epsilon) + 1} \\
%     &= C(d, \gamma) M^{-(d-1)} \delta^{d + p^* \frac{2}{d} (-1 + \epsilon)} .
% \end{align*}
In the third step we used that $p^*$, $b_V$, $\epsilon$ only depend on $d$ and $\gamma$. We also used \eqref{eq:eps-check} to ensure that the integral is finite.
% \begin{equation}
%     p^* \frac{2}{d} (-1 + \epsilon) > -1,
% \end{equation}
% which follows from the choice of $\epsilon = \epsilon(d,\gamma)$.
%Note that by \eqref{eq:eps-check} and $d\ge2$, we have
%\begin{equation}\label{eq:deltato0}
%    \delta^{d + p^* \frac{2}{d} (-1 + \epsilon)} \rightarrow 0 \mathrm{\ as\ }\delta \rightarrow 0 .
%\end{equation}

\bigskip
Let $K=K(d,\ga)$ be the constant in the covering theorem for oscillatory domains (Lemma \ref{le:cover})). By \eqref{eq:Vmdeltaex}, we can choose $\delta>0$ small enough such that
\begin{equation}\label{eq:4lambdakpstar}
    (4\lambda K)^{p^*} C(d, \gamma) \delta^{d + p^* \frac{2}{d} (-1 + \epsilon)} < \left(\frac{1}{C_{PS}}\right)^{p^*} .
\end{equation}
Note that $\delta$ only depends on $d$, $\gamma$, $\lambda$ but \emph{not} on $x$. For each $x \in \Omega$ with $x_d \geq 0$ let $D_x := D_x(\delta)$ with $\delta$ defined in \eqref{eq:4lambdakpstar}. By \eqref{eq:4lambdakpstar} and \eqref{eq:Vmdeltaex},
% the choice of $\delta > 0$, we know that
\begin{equation}
    \nor{4\lambda K V}{p^*, D_x}{p^*} M^{d-1} < \left(\frac{1}{C_{PS}}\right)^{\ps} ,
\end{equation}
so for every $v \in H^1(\Omega)$ with $\int_{D_x} v = 0$, we have by \eqref{eq:evVintcondMV}
\begin{equation} \label{eq:int4klambdasquaredv}
    \int_{D_x} |\nabla v|^2 + \int_{D_x} 4 K \lambda V |v|^2 \geq 0  .
\end{equation}
Let $\mathcal{F}_1, \dots, \mathcal{F}_K$ be the families of oscillatory domains we get from applying the covering theorem for oscillatory domains (Lemma \ref{le:cover}) to $\left\{D_x\right\}_{x \in \Omega \mathrm{\ with\ } x_d \geq 0}$. Note that
\begin{equation}\label{eq:DcoversuppV}
    \bigcup_{k = 1}^K \bigcup_{D \in \mathcal{F}_k} D \supset \supp(V) .
\end{equation}
Using \eqref{eq:int4klambdasquaredv} and \eqref{eq:DcoversuppV}, we obtain in exactly the same way as in \eqref{eq:quadrformest} in the proof of Lemma \ref{le:selfadj}
\begin{align*}
    -\Delta^N_\Omega + \lambda V &\geq \frac{1}{2} (- \Delta^N_\Omega) + \frac{1}{2K} \sum_{k = 1}^K \sum_{D \in \mathcal{F}_k} (- \Delta^N_D + 2 \lambda K V 1_D )\\
    &\geq \frac{1}{2} (- \Delta^N_\Omega) - 2 \lambda \left(\sum_{k = 1}^K \sum_{D \in \mathcal{F}_k}  \frac{1}{|D|}\int_D |V|\right)1_\om .
\end{align*}
For each $k \in \left\{1, \dots, K\right\}$ and $D \in \mathcal{F}_k$, we have
\begin{equation}
    |D| \geq \frac{1}{4} \delta a^{d-1} \geq \frac{1}{4} \delta \left(c_0(c_1\delta)^{\frac{1}{\gamma}}\right)^{d-1} = \frac{1}{4} c_0^{d-1} c_1^{\dg} \delta^{1 + \dg} .
\end{equation}
Since the domains $D \in \mathcal{F}_k$ for each $k \in \left\{1, \dots, K\right\}$ are all disjoint and $\Omega$ is bounded, we deduce that $\sum_{k = 1}^K |\mathcal{F}_k| < \infty$.
% \begin{equation}
%     \sum_{k = 1}^K |\mathcal{F}_k| < \infty .
% \end{equation}
Thus, since $V \in L^{p^*}(\Omega) \subset L^1(\Omega)$, we get
\begin{equation}
    \sum_{k = 1}^K \sum_{D \in \mathcal{F}_k} \int_D |V| \frac{1}{|D|} < \infty ,
\end{equation}
so $- \Delta^N_\Omega + \lambda V$ is a well-defined quadratic form on $H^1 (\Omega)$ that is bounded from below and has the $H^1(\Omega)$-norm as its quadratic form norm. We can deduce as in the proof of Lemma \ref{le:evkv} that $- \Delta^N_\Omega + \lambda V$ has finitely many negative eigenvalues. \\

\textbf{Proof of (ii). } 
By the definition of $V$ and of $\no{V}{\tilde{p}, \beta}$, it suffices to consider the case $\epsilon := (d-1)\left(1/\gamma- 1\right)$. Recall that
\begin{equation}
    \supp(V) = \left\{x = (x', x_d) \in \R^{d-1} \times \R \,\middle\vert\, x_d \geq 0\right\} \cap \Omega .
\end{equation}
For all $x \in \Omega$ with $x_d \geq 0$ we use the shorthand notation
\begin{equation}
    h_x := f(x') - x_d  .
\end{equation}
By the definition of $\no{V}{\tilde{p}, \beta}$,  
% \begin{equation}
%     \no{V}{\tilde{p}, \beta} = |V|_{\tilde{p}, \beta} + \no{V}{\frac{d}{2}} ,
% \end{equation}
it suffices to show $|V|^{\tilde{p}}_{\tilde{p},\beta} = \infty$.
% \begin{equation}
%     |V|^{\tilde{p}}_{\tilde{p},\beta} = \infty .
% \end{equation}
% By the definition of $V$ and of $\no{V}{\tilde{p}, \beta}$, it suffices to consider the case $\epsilon := (d-1)\left(1/\gamma- 1\right)$. Recall that
% \begin{equation}
%     \supp(V) = \left\{x = (x', x_d) \in \R^{d-1} \times \R \,\middle\vert\, x_d \geq 0\right\} \cap \Omega .
% \end{equation}
% For all $x \in \Omega$ with $x_d \geq 0$ we use the shorthand notation
% \begin{equation}
%     h_x := f(x') - x_d  .
% \end{equation}
% Since
% \begin{equation}
%     \no{V}{\tilde{p}, \beta} = |V|_{\tilde{p}, \beta} + \no{V}{\frac{d}{2}} ,
% \end{equation}
% it suffices to show
% \begin{equation}
%     |V|^{\tilde{p}}_{\tilde{p},\beta} = \infty .
% \end{equation}
Recall that
\begin{equation}
    |V|^{\tilde{p}}_{\tilde{p},\beta} = \int_\Omega \dr{x} |V(x)|^{\tilde{p}} h_x^{-\beta}= \int_\Omega |b_V|^{\tilde{p}} h_x^{\tilde{p} \frac{2}{d} (-1 + \epsilon)} h_x^{- \beta}.
\end{equation}
Now, recall that by \eqref{eq:gnv}, we have
\begin{equation}
    f(x') =\sum_{j=0}^\infty 2^{-{\gamma}jm}g_j(x')\ge g_0(x')\geq \frac{1}{4}
\end{equation}
for all $x' \in \left[1/4,3/4\right]^{d-1}$, so
\begin{align*}
    |V|^{\tilde{p}}_{\tilde{p}, \beta} &= |b_V|^{\tilde{p}} \int_\Omega h_x^{\tilde{p} \frac{2}{d} (-1 + \epsilon) - \beta} = |b_V|^{\tilde{p}} \int_{Q^{(d-1)}} \dr{x'} \int_0^{f(x')} \dr{x_d} \left(f(x') - x_d\right)^{\tilde{p} \frac{2}{d} (-1 + \epsilon) -\beta} \\
    &\geq |b_V|^{\tilde{p}} \int_{\left[\tfrac{1}{4}, \tfrac{3}{4}\right]^{d-1}} \dr{x'} \int_0^{f(x')} \dr{t} t^{\tilde{p} \frac{2}{d} (-1 + \epsilon) - \beta}\geq |b_V|^{\tilde{p}} 2^{-(d-1)} \int_0^{\frac{1}{4}} \dr{t} t^{\tilde{p} \frac{2}{d} (-1 + \epsilon) - \beta} .
\end{align*}
Hence, if we can show that 
\begin{equation}\label{eq:pfrac-beta}
    \tilde{p} \frac{2}{d} (-1 + \epsilon) - \beta \leq -1 ,
\end{equation}
then we have $|V|^{\tilde{p}}_{\tilde{p}, \beta} = \infty$, so $\no{V}{\tilde{p}, \beta} = \infty$. Recalling \eqref{eq:ptdef}, the bound \eqref{eq:pfrac-beta} is equivalent to
% In the following, let us show \eqref{eq:pfrac-beta}.
%Using the notation $\mu := \dgo$ and \eqref{eq:ptdef}, we  have
% Hence, if we can show that $\tilde{p} \tfrac{2}{d} (-1 + \epsilon) -\beta \leq -1$, then we have $|V|^{\tilde{p}}_{\tilde{p}, \beta} = \infty$, so $\no{V}{\tilde{p}, \beta} = \infty$. \\

% In the following, let us show that
% \begin{equation}\label{eq:pfrac-beta}
%     \tilde{p} \frac{2}{d} (-1 + \epsilon) - \beta \leq -1 .
% \end{equation}
% We have using the notation $\mu := \dgo$ 
% \begin{align*}
%     &\quad \tilde{p} \frac{2}{d} (-1 + \epsilon) - \beta \\
%     &= \frac{1}{2d} \left(\dgo\right)^2 \frac{2}{d}\left(-1 + (d-1) \left(\frac{1}{\gamma} - 1\right)\right) \\
%     &\quad - \frac{1}{d + 1} \left(\dgo\right)\left[\frac{1}{d} \left(\dgo\right)^2 - d\right] \\
%     &= \frac{1}{d^2} \left(\dgo\right)^2 \left(-1 + \dg -d + 1\right) - \frac{1}{d + 1} \left(\dgo\right)\left[\frac{1}{d} \left(\dgo\right)^2 - d\right] \\
%     &= \frac{1}{d^2} \left(\dgo\right)^2 \left(\left(\dgo\right) - (d + 1)\right) \\
%     &\quad - \frac{1}{d + 1} \left(\dgo\right)\left[\frac{1}{d} \left(\dgo\right)^2 - d\right] \\
%     &= \frac{1}{d^2} \mu^2 \left(\mu -(d + 1)\right) - \frac{1}{d + 1} \mu \left[\frac{1}{d} \mu^2 - d\right] .
% \end{align*}
\begin{equation} \label{eq:f-finall}
   f(\mu)=\frac{1}{d^2} \mu^2 \left(\mu -(d + 1)\right) - \frac{1}{d + 1} \mu \left[\frac{1}{d} \mu^2 - d\right] \le -1
\end{equation}
with $\mu := \dgo \in (d,d+1)$. Note that 
$$
 f'(\mu)= \frac{3 \mu^2}{d^2(d + 1)} - \frac{2 (d + 1)}{d^2} \mu + \frac{d}{d + 1}  
$$
is a convex function with $f'(d)<0$ and $f'(d+1)<0$. Therefore, $f$ is monotone decreasing in $(d,d+1)$, and hence \eqref{eq:f-finall} follows from the fact that $f(d)=-1$. The proof of Lemma \ref{le:selfadjointev} is complete. %his finishes the proof.
\end{proof}
\begin{proof}[Proof of Theorem \ref{th:example}]
    Remark \ref{re:selfadjointev} combined with Lemma \ref{le:selfadjointev}(i) show Theorem \ref{th:example}.
\end{proof}

\printbibliography

%\bibliographystyle{siam}
%\bibliography{Literature}

\end{document}